\documentclass[11pt]{article}   
\usepackage{geometry}
 
\usepackage{amsmath,amsfonts,amsthm,enumitem}
\usepackage[dvipsnames]{xcolor}
\usepackage{empheq,tensor,cancel,braket,tcolorbox,blkarray,authblk,cite}
\usepackage{tikz-cd}
\usetikzlibrary{decorations.pathreplacing}
\usepackage[colorlinks=true, linkcolor=blue, citecolor=blue, linktoc=all]{hyperref}
\usepackage{bm}
\usepackage{amssymb}
\usepackage{pifont}
\usepackage{tabularx}
\usepackage{ltablex}
\usepackage{booktabs}

\newcommand{\uiuc}[1]{
\\ \bigskip
Illinois Center for Advanced Studies of the Universe \& Department of Physics,\\ University of Illinois, 1110 West Green St., Urbana IL 61801, U.S.A.
}

\def\pa{\partial}
\renewcommand{\comment}[1]{}
\newcommand{\ti}{\tilde }
\newcommand{\bb}{\mathbb}
\newcommand{\td}{\textnormal{d}}
\newcommand{\E}{\textnormal{e}}

\newcommand{\un}[1]{\underline{#1}}

\newcommand{\be}{\begin{equation}}
\newcommand{\ee}{\end{equation}}
\newcommand{\beq}{\begin{eqnarray}}
\newcommand{\eeq}{\end{eqnarray}}

\newcommand{\RR}{\mathbb{R}}

\newcommand{\nn}{\nonumber}
\newcommand{\p}{\partial}

\theoremstyle{definition}
\newtheorem{defn}{Definition} 

\theoremstyle{plain}
\newtheorem{theorem}{Theorem}[section]

\newtheorem{prop}[theorem]{Proposition}
\newtheorem{lemma}[theorem]{Lemma}

\allowdisplaybreaks[1]

\title{\huge Weyl-Ambient Geometries}
\author{
	Weizhen Jia\footnote{weizhen2@illinois.edu},\; 
	Manthos Karydas\footnote{karydas2@illinois.edu}\; and 
	Robert G.~Leigh\footnote{rgleigh@illinois.edu}
	{\small \uiuc{}}
}
\date{}
\begin{document}
\maketitle

\begin{abstract}
Weyl geometry is a natural extension of conformal geometry with Weyl covariance mediated by a Weyl connection. We generalize the Fefferman-Graham (FG) ambient construction for conformal manifolds to a corresponding construction for Weyl manifolds. We first introduce the Weyl-ambient metric motivated by the Weyl-Fefferman-Graham (WFG) gauge. From a top-down perspective, we show that the Weyl-ambient space as a pseudo-Riemannian geometry induces a codimension-2 Weyl geometry. Then, from a bottom-up perspective, we start from promoting a conformal manifold into a Weyl manifold by assigning a Weyl connection to the principal $\mathbb{R}_+$-bundle realizing a Weyl structure. We show that the Weyl structure admits a well-defined initial value problem, which determines the Weyl-ambient metric. Through the Weyl-ambient construction, we also investigate Weyl-covariant tensors on the Weyl manifold and define extended Weyl-obstruction tensors explicitly.
\end{abstract}

\newpage
\begingroup
\hypersetup{linkcolor=black}
\tableofcontents
\endgroup

\newpage
\section{Introduction}
Conformal geometry is a very rich area of mathematics with its history deeply intertwined with that of physics. Historically, the subject was initiated at the beginning of the twentieth century with the work of Weyl \cite{Weyl:1918pdp}, Cartan \cite{CartanLesE} and Thomas \cite{thomas1925invariants}. In physics, there have been numerous applications of conformal geometry, from conformal compactification \cite{Penrose:1962ij} and conformal gravity \cite{Mannheim:2011ds} to the AdS/CFT correspondence \cite{Maldacena:1997re,Witten:1998qj}. 
 \par
The fundamental structure appearing in conformal geometry is a manifold $M$ endowed with a \emph{conformal class} of metrics $[g]$. Two metrics belong in the same conformal class $[g]$ if one metric is a smooth positive multiple of the other. Local rescalings of the metric tensor by an arbitrary smooth positive function are called \emph{Weyl transformations}. Compared to pseudo-Riemannian manifolds $(M,g)$, conformal manifolds are endowed with an enlarged symmetry group with both diffeomorphisms and Weyl transformations, denoted by $\text{Diff}(M)\ltimes \text{Weyl}$. A tensor $T$ on a conformal manifold $(M,[g])$ is said to be conformally covariant if it transforms covariantly under a Weyl transformation:
\begin{equation}\label{conformal_Weyl_tensor}
T\to {\cal B}(x)^{w_{T}}T\,,\qquad\text{when}\qquad g\to {\cal B}(x)^{-2}g\,,
\end{equation}
where $w_{T} $ is  the Weyl weight of the tensor $T$. On the physics side, conformal-covariant tensors appear as expectation values of operators in conformal field theories coupled to a background metric. As an important example, the expectation value of the trace of the energy-momentum tensor acquires an anomalous term after quantization, namely the celebrated Weyl anomaly \cite{Capper:1974ic}. By investigating the effective action in dimensional regularization, Deser and Schwimmer \cite{Deser:1993yx} made a conjecture regarding the possible candidates for the Weyl anomaly, which are global conformal invariants. This conjecture was later proven in \cite{alexakis2012decomposition,Boulanger:2007ab,Boulanger:2007st}.\footnote{The analysis in \cite{alexakis2012decomposition} concerns local conformal invariants, corresponding to the type B Weyl anomaly, while \cite{Boulanger:2007ab,Boulanger:2007st} deals with the type A Weyl anomaly.}

Just as diffeomorphism-covariant quantities, i.e.,\ tensors, on pseudo-Riemannian manifolds can easily be constructed out of the metric, Riemann tensor and covariant derivatives, one might expect to find conformal-covariant tensors on conformal manifolds.  
However, unlike the abundance of diffeomorphism-covariant quantities on $(M,g)$, it is significantly harder to construct conformal-covariant tensors on $(M,[g])$. Before the work of Fefferman and Graham, known examples of conformal tensors were the Weyl tensor in any dimension, the Cotton tensor \cite{cotton1899varietes} in $3d$ and the Bach tensor \cite{bach1921weylschen} in $4d$. In their seminal work \cite{AST_1985__S131__95_0,Fefferman:2007rka} Fefferman and Graham introduced the ambient metric construction based on previous work by Fefferman \cite{fefferman1979parabolic}, which provided a systematic method of finding conformal-covariant tensors. The basic idea of the construction was to associate  a $(d+2)$-dimensional ``ambient" pseudo-Riemannian manifold to a $d$-dimensional conformal manifold. One can then find a specific class of ambient diffeomorphisms that induces Weyl transformations on the conformal manifold. In the context of AdS/CFT,  diffeomorphisms that induce a Weyl transformation of the boundary metric are the Penrose-Brown-Henneaux (PBH) transformations \cite{Imbimbo:1999bj}. Thus, conformal-covariant tensors can descend from ambient Riemannian tensors, and their Weyl transformations can be derived from certain ambient diffeomorphisms. 
\par
An important outcome of the ambient construction was to define extended obstruction tensors from covariant derivatives of the ambient Riemann tensor \cite{graham2009extended}. Obstruction tensors are the generalization to higher (even) dimension of the Bach tensor. For each even dimension, the corresponding obstruction tensor is the only irreducible conformal-covariant tensor in that dimension \cite{graham2005ambient}. Defined through the ambient space, the $k^{th}$ extended obstruction tensor $\Omega^{(k)}_{ij}$ has a simple pole at $d=2k+2$, whose residue is the obstruction tensor in that dimension. For example, the first obstruction tensor reads
\be
\Omega^{(1)}_{ij}=-\frac{1}{d-4}B_{ij}\,,
\ee
where $B_{ij}$ is the Bach tensor, namely the obstruction tensor in $4d$. The extended obstruction tensors also play an integral role in the context of holography as the basic building blocks of the holographic Weyl anomaly \cite{Henningson:1998gx,graham2009extended}, or equivalently the Q-curvature \cite{branson1991explicit,graham2005ambient,Anastasiou:2020zwc}.
 \par
A different perspective on conformal geometry was introduced by Weyl \cite{Weyl:1918pdp}, whose idea was to make the physical scale a local quantity. The Weyl connection was introduced so that one can transport the physical scale between two points of the manifold. Although Weyl's initial attempt to identify the Weyl connection with the electromagnetic gauge field failed, the consistent mathematical structure he introduced was developed further in \cite{10.4310/jdg/1214429379,doi:10.1063/1.529582}. In this approach, a Weyl connection $a$ is introduced on the conformal manifold which transforms together with the metric $g$ under a Weyl transformation. One can modify the conformal class $[g]$ to a \emph{Weyl class} $[g,a]$, which is the equivalence class formed by the pairs $(g,a)\sim ({\cal B}(x)^{-2}g, a- \td\ln{\cal B}(x))$. This defines a Weyl manifold $(M,[g,a])$, and the conformal geometry is promoted to \emph{Weyl geometry} \cite{10.4310/jdg/1214429379,doi:10.1063/1.529582,scholz2018unexpected}. Equivalently, a Weyl connection can be thought of as a connection on the \emph{Weyl structure}, which is a principal bundle with the Weyl symmetry group as the structure group \cite{10.4310/jdg/1214429379}. 
\par
Similarly to a conformal-covariant tensor, one can define a Weyl-covariant tensor $T$ on a Weyl manifold $(M,[g,a])$ to be a tensor that transforms covariantly under a Weyl transformation:
\begin{equation}\label{Weyl_tensor}
\begin{split}
T\to {\cal B}^{w_{T}}(x)T\,,\qquad\text{when}\qquad g\to {\cal B}(x)^{-2}g\,,\quad a\to a- \td \ln {\cal B}(x)\,.
\end{split}
\end{equation}
Although conformal-covariant tensors on a conformal manifold $(M,[g])$ are hard to find, Weyl-covariant tensors on a Weyl manifold $(M,[g,a])$ can be constructed quite easily. Recall that on a pseudo-Riemannian manifold $(M,g)$, one can define a Levi-Civita (LC) connection $\nabla$, and it is well-known that diffeomorphism-covariant quantities can be constructed from the metric, Riemann curvature, and covariant derivatives  of the Riemann curvature. On a Weyl manifold $(M,[g,a])$, one can define a Weyl-Levi-Civita connection $\hat\nabla$, and a plethora of Weyl-covariant quantities can similarly be constructed from the metric, Weyl-Riemann curvature, and  Weyl-covariant derivatives $\hat\nabla$ of the Weyl-Riemann curvature. This indicates that the $\text{Diff}(M)\ltimes \text{Weyl}$ symmetry is manifested more naturally on a Weyl manifold, and the representation has a similar structure as that of  $\text{Diff}(M)$ on  pseudo-Riemannian manifolds. 
There are corresponding notions of Weyl metricity, Weyl torsion and a uniqueness theorem giving a Weyl-LC connection \cite{10.4310/jdg/1214429379}. 
\par
A significant step towards incorporating the Weyl connection in the formalism of AdS/CFT was taken in \cite{Ciambelli:2019bzz}. By modifying the Fefferman-Graham (FG) ansatz of an asymptotically locally AdS (AlAdS) spacetime to the Weyl-Fefferman-Graham (WFG) form, it was shown \cite{Ciambelli:2019bzz}  that  the bulk LC connection induces a Weyl connection on the conformal boundary. Thus, the AlAdS bulk geometry in the WFG gauge induces a Weyl geometry instead of only a conformal geometry on the conformal boundary. Applying the holographic dictionary \cite{Witten:1998qj}, they also calculated the holographic Weyl anomaly in the WFG gauge, and found it can be organized in a Weyl-covariant fashion. In the FG ambient construction, the conformal boundary $(M,[g])$ of a $(d+1)$-dimensional AlAdS bulk is associated with a $(d+2)$-dimensional ambient space, and the AlAdS bulk in the FG gauge can be considered as a hypersurface in the ambient space. A natural question to ask is whether such a construction exists for the conformal boundary as a Weyl manifold. In this paper we will provide such a construction. We introduce the Weyl-ambient space $(\tilde M,\tilde g)$ as a modification of the FG ambient space, in which the AlAdS bulk in the WFG gauge is a hypersurface and its boundary is associated with a codimension-2 Weyl manifold $(M,[g,a])$.
\par
Following \cite{Ciambelli:2019bzz}, the AlAdS bulk expansion in the WFG gauge was further investigated in \cite{Jia:2021hgy}. Using the technique of dimensional regularization, the Weyl-obstruction tensors and extended Weyl-obstruction tensors were introduced as the poles in the on-shell metric expansion. It was also found in \cite{Jia:2021hgy} that the holographic Weyl anomaly can be expressed in terms of extended Weyl-obstruction tensors. Although it is convenient to read off the extended Weyl-obstruction tensors from the pole of the AlAdS metric expansion, this should not be regarded as a precise definition since the pole may have an ambiguity when shifted by a finite term. One of the results of the present  paper is to provide a definition of  Weyl-obstruction tensors on a Weyl manifold $(M,[g,a])$ through the Weyl-ambient space $(\tilde M,\tilde g)$, in a way analogous to how extended obstruction tensors were defined in \cite{graham2009extended,Fefferman:2007rka}. Many properties of the extended Weyl-obstruction tensors can also be derived from the Weyl-ambient space.
\par
The main goal of this paper is to provide an ambient construction for Weyl manifolds. We start by introducing the Weyl-ambient metric as a modification of the FG ambient metric. We will then present two perspectives. The first one is a top-down approach. We will see that one naturally obtains a codimension-2 Weyl manifold $(M,[g,a])$. A more formal approach is the bottom-up perspective, where we start from a $d$-dimensional conformal manifold $(M,[g])$, which is then enhanced into a Weyl manifold $(M,[g,a])$ by introducing a connection on the Weyl structure over $M$. A $(d+2)$-dimensional Weyl-ambient space can then be constructed by taking the Weyl structure as an initial surface, which follows the rigorous ambient space construction in \cite{Fefferman:2007rka}.
\par
This paper will be organized as follows. In Section \ref{Sec2} we first briefly review the ambient metric of Fefferman-Graham before we introduce the Weyl-ambient metric $\tilde g$ at the end. To get some intuition on the construction, we start from the example of the flat ambient metric and then generalize to Ricci-flat ambient metrics. Some different coordinate systems presented in Section \ref{Sec2} are described in Appendix \ref{AppA}. Then, from a top-down perspective, in Section \ref{sec:topdown} we demonstrate how $(\tilde M,\tilde g)$ induces a codimension-2 Weyl manifold $(M,[g,a])$. We also discuss how the Weyl-covariant tensors on $(M,[g,a])$ can be derived from the Riemann tensor of $(\tilde M,\tilde g)$, and define the extended Weyl-obstruction tensors as a special example. We work in first order formalism in Section \ref{sec:topdown} by introducing a null frame, with some details of the calculation given in Appendix \ref{App:Null}. Section \ref{sec:bottomup} is the bottom-up construction of the Weyl-ambient metric. We show rigorously that the Weyl-ambient metric has a well-defined perturbative initial value problem, where the Ricci-flatness condition plays the role of the equation of motion. To show this we follow \cite{Fefferman:2007rka} closely using the second order formalism and generalize their definitions properly for the Weyl-ambient space. We also extend some major theorems of \cite{Fefferman:2007rka} with suitable modifications for the Weyl-ambient space.  The details of some proofs are presented in Appendix \ref{AppC}. After that we discuss   Weyl-covariant tensors and extended Weyl-obstruction tensors from the point of view of the second order formalism, and prove the equivalence of the extended Weyl-obstruction tensors defined from both approaches. Finally, our results are summarized in Section \ref{sec:conclu}.

\subsection*{Notation}
We will label the indices in a $d$-dimensional manifold $M$ by lowercase Latin letters $i,j,\cdots$, in a $(d+1)$-dimensional AlAdS bulk by lowercase Greek letters $\mu,\nu,\cdots$, and in a $(d+2)$-dimensional ambient space $\tilde M$ by uppercase Latin letters $I,J,\cdots$. The vectors on $M$ are denoted by $ \un U,\un V$, on the Weyl structure ${\cal P}_W$ over $M$ are denoted by $\un u,\un v$, and on the ambient manifold $\tilde M$ are denoted by $\un {\cal U},\un{\cal V}$.

In Section \ref{sec:topdown}, we mainly use the dual frame $\{\bm e^I\}$, and the ambient frame indices are $I=+,1,\cdots,d,-$. Unless otherwise indicated, in Section \ref{sec:bottomup} we mainly use the ambient coordinate system $\{t,x^i,\rho\}$, and the indices are $I=0,1,\cdots,d,\infty$, where $0$ labels the $t$-component and $\infty$ labels the $\rho$-component. The notation $(0,x^i,\infty)$ is also used for the components in a trivialization ${\cal P}_W\times \bb R\simeq \bb R_+\times M\times\bb R$, even without specifying a choice of coordinates on $M$. The above-mentioned notation is summarized in Table \ref{t1}.

\begin{table}[!h]
\centering
\caption{Notation}
\begin{tabularx}{\textwidth}{c|c|c|X}
\toprule
Dimension & Manifold & Vectors & Indices \\
\midrule
$d$ & $M$ & $ \un U,\un V$ & $i,j,\cdots$ $\quad\{x^i\}$ $\quad i=1,\cdots,d$\\
\midrule
$d+1$ & AlAdS$_{d+1}$ &  & $\mu,\nu,\cdots$ $\quad\{x^\mu\}=\{z,x^i\}$  $\quad i=1,\cdots,d$ \\
\midrule
$d+1$ & ${\cal P}_W$ & $\un u,\un v$ &   \\
\midrule
$d+2$ & $\tilde M$ & $\un {\cal U},\un{\cal V}$ & $I,J,\cdots$\newline  In the frame $\{\bm e^I\}=\{\bm e^+,\bm e^i,\bm e^-\}$, $I=+,1,\cdots,d,-$. \newline In the coordinates $\{x^I\}=\{t,x^i,\rho\}$, $I=0,1,\cdots,d,\infty$.\\
\bottomrule
\end{tabularx}
\label{t1}
\end{table}

\section{Ambient Metrics}
\label{Sec2}
In this section we will review the FG ambient metric and introduce the Weyl-ambient metric. To build up some intuition, we begin with the flat ambient metric and then generalize to Ricci-flat ambient metrics.
\subsection{Flat Ambient Metrics}\label{Sec2:Subs1}
The simplest example of an ambient space is the flat ambient space. Consider the $(d+2)$-dimensional Minkowski spacetime $\mathbb{R}^{1,d+1}$ with the metric
\begin{equation}\label{Flat_Ambient metric}
\eta= - (\td X^{0})^{2}+\sum_{i=1}^{d+1}(\td X^{i})^{2}\,.
\end{equation}
One can describe $(d+1)$-dimensional Euclidean AdS spaces as the following codimension-1 hyperboloids:\footnote{One can also take the signature in \eqref{Flat_Ambient metric} to be $(2,d)$. Then, $g^+$ will be the Lorentzian signature AdS spacetime and the $\delta_{ij}$ in \eqref{Flat_Ambient_3} becomes $\eta_{ij}$. More generally, if one takes the signature in \eqref{Flat_Ambient metric} to be $(p,d+2-p)$, then the signature of $g^+$ will be $(p-1,d+2-p)$.}
\be
(X^{0})^2-R^2=L^2\,,\qquad R^2=\sum_{i=1}^{d+1}(X^{i})^{2}\,,
\ee
where $L$ represents the AdS radius. The hyperboloids with different $L$ form a one-parameter family of hypersurfaces foliating the interior of the future light cone, denoted by ${\cal N}^+$, emanating from the origin of the Lorentzian coordinate system $\{X^{0}, X^{i}\}$. Then, one can also write the Minkowski metric in the following ``cone'' form:
\begin{equation}\label{Flat_Ambient_4}
\eta = -\td \ell^2 + \frac{\ell^2}{L^2} g^{+}\,,\qquad \ell>0\,,
\end{equation}
where the coordinate $\ell=\sqrt{(X^{0})^2-R^2}$, and $g^{+}$ is the $(d+1)$-dimensional Euclidean AdS metric. Now the Euclidean AdS space is represented by the hyperbola $\ell=L$. The metric $g^{+}$ can be expressed in the Fefferman-Graham (FG) form in the following different ways (see Appendix \ref{AppA} for details):
\begin{align}
\label{eq:adsglob}
\qquad g^{+}_{S}&=  \frac{L^2}{z^2} \Big(\td z^2 + L^2(1- \frac{1}{4}(z/L)^2)^2\td\Omega_{d}^2\Big)\,,\qquad 0<z<2L\,,\\
\label{eq:adsPoin}
\qquad g^{+}_{F}&= \frac{L^2}{z^2}\left(\td z^2 + \delta_{ij}\td x^{i}\td x^{j}\right)\,,\qquad i=1,\cdots,d\,,\qquad z>0\,.
\end{align}
The metric \eqref{Flat_Ambient_4} with $g^+=g^{+}_{S}$ or $g^{+}_{F}$ is defined in the whole interior of the light cone ${\cal N}^+$,\footnote{Note that for Lorentzian signature AdS spacetime, the metric \eqref{Flat_Ambient_4} with $g^{+}_{F}$ only covers half of the interior of the future light cone.} while their AdS boundaries have different topologies. It is easy to see that the AdS boundary at $z\to0^+$ of $g^+_S$ in \eqref{eq:adsglob} is conformally a $d$-sphere while that of $g^+_F$ in \eqref{eq:adsPoin} is conformally flat. 

While the metric \eqref{Flat_Ambient_4} is singular in the limit $z\to 0^{+}$ with $\ell$ fixed, it is well-defined when taking both $z$ and $\ell$ to zero with $z/\ell$ fixed. To make this evident we introduce a new coordinate system $\{t,x^{i},\rho\}$, called the \emph{ambient coordinate system}, with $t=\ell/z$ and $\rho=-z^{2}/2$. First we look at the metric \eqref{Flat_Ambient_4} with $g^+_S$ in \eqref{eq:adsglob}, which in the ambient coordinate system becomes
\begin{align}\label{Flat_Ambient_2}
\eta &= 2\rho \td t^2 + 2t \td t \td\rho + t^2 (1+ \frac{\rho}{2 L^2})^2 L^2 \td\Omega_{d}^2 \,.
\end{align} 
The coordinate patch of $\{\ell,x^i,z\}$ which covers the interior of the light cone surface ${\cal N}^{+}$, corresponds to $t\in (0,\infty) $, $\rho \in (-2L^2,0)$ (see Figure~\ref{fig:cones}). However, it is apparent now that the limit $\rho\to 0^{-}$ of the above metric is well-defined, and thus we can extend the coordinate patch of $\{t,x^i,\rho\}$ to include an open neighborhood of the surface ${\cal N}^{+}$ at $\rho=0$. Hence, ${\cal N}^{+}$ is parametrized by $\{t,x^{i}\}$, where $t\in \mathbb{R}_{+}$ and $x^{i}$ are the coordinates of the $d$-sphere $S_{d}$. In other words, ${\cal N}^{+}$ can be regarded as a line bundle over $S^d$ whose fibres are parametrized by $t$.
\par
Suppose $\phi$ is a function on $\bb R^{1,d+1}$, which defines a hypersurface $\Sigma$ by the locus of points $p\in \mathbb{R}^{1,d+1}$ such that $\phi(t,x^{i},\rho)|_p=0$. In order to find the intersection $\Sigma \cap \mathcal{N}^{+}$, one can set $\rho=0$ and solve for $t$ as a function $t(x^{i})$ of the $d$-sphere coordinates from $\phi(t,x,\rho=0)=0$. The pullback metric on the intersection submanifold is $\eta |_{\Sigma \cap {\cal N}^{+}} = t (x)^2L^2 \td\Omega_{d}^2$. The function $t(x)$ depends on the choice of function $\phi$ (which is arbitrary) that defines $\Sigma$, and thus we see that the pullback metric is conformally equivalent to the metric of $S_{d}$. An example is to take $\phi=\ln t$, and to consider the pull back of the metric at $\rho=0$, $t=1$, namely $\eta|_{\rho=0,t=1}= L^2\td\Omega^{2}_{d}$. If we perform a diffeomorphism $t={\cal B}(x)^{-1}t'$ and pull back the metric at $\rho=0$, $t'=1$, then we find $\eta|_{\rho=0,t'=1}= {\cal B}(x)^{-2}L^2\td\Omega^{2}_{d}$. Therefore, at $\rho=0$ we have a \emph{conformal class} $[g]$ of $d$-dimensional metrics, and the $(d+2)$-dimensional Minkowski metric expressed in \eqref{Flat_Ambient_2} is said to be the ambient metric of $[g]$. This implies that the null surface ${\cal N}^{+}$ at $\rho=0$ is associated with a metric bundle, which will be important for the formal construction later in Section \ref{sec:bottomup}.
\par
Similarly, the metric \eqref{Flat_Ambient_4} with $g^+_{F}$ in \eqref{eq:adsPoin} can also be expressed in the ambient coordinates as
\begin{figure}[!htb]
\begin{center}
\includegraphics[width=3.5in]{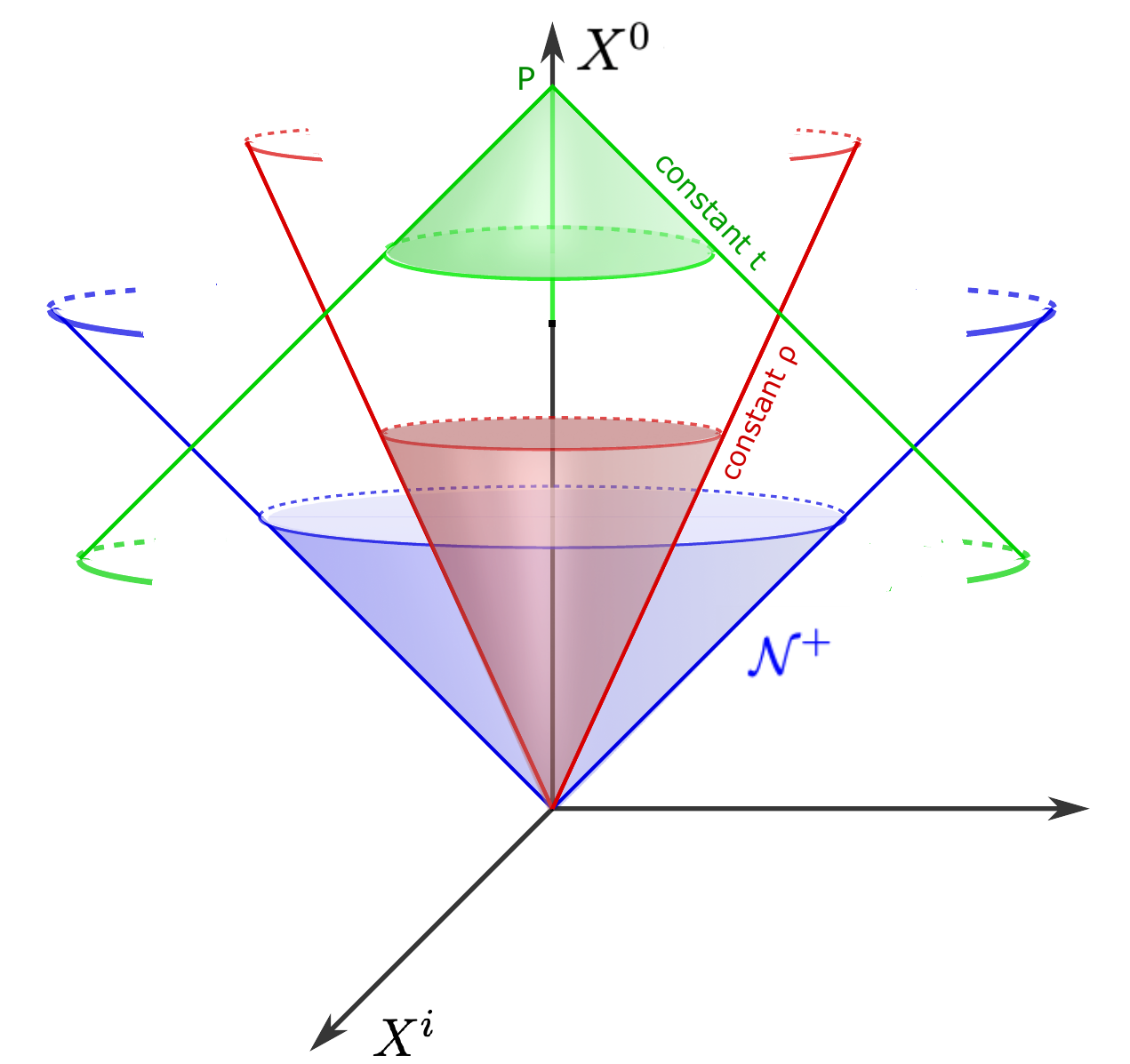}
\caption{Sketch of a constant-$\rho$ surface (red) and a constant-$t$ surface (green) of the flat ambient metric \eqref{Flat_Ambient_2} in the Lorentzian coordinate system $\{X^{0}, X^{i}\}$. Constant-$t$ surfaces are past directed light cones. Changing $t$ moves the apex $P$ of the cone along the $X^{0}$-axes. Constant-$\rho$ surfaces are future directed timelike cones. When $\rho\to 0^{-}$ the constant $\rho$ surface becomes the light cone ${\cal N}^{+}$ (blue).
} 
\label{fig:cones}
\end{center}
\end{figure}

\begin{equation}\label{Flat_Ambient_3}
\eta = 2\rho \td t^2 + 2t\td t \td\rho + t^2\delta_{ij} \td x^{i}\td x^{j}\,,\qquad i=1,\cdots,d\,.
\end{equation}
In this case, the original coordinate patch of $\{\ell,x^i,z\}$ corresponds to $t\in(0,\infty)$, $\rho\in(-\infty,0)$, and the null surface ${\cal N}^+$ is again covered by the $\{t,x^i,\rho\}$ system at $\rho=0$. Intersecting the null surface with a hypersurface and taking the pullback metric on the intersection, we now obtain a $d$-dimensional metric $\td s^2= t(x)^2 \delta_{ij}\td x^{i}\td x^{j}$ that is conformally flat. This metric is also in the conformal class $[g]$ but the topology is different from the $d$-dimensional metric obtained from \eqref{Flat_Ambient_2}. Note that the flat ambient metric in either \eqref{Flat_Ambient_2} or \eqref{Flat_Ambient_3} is homogeneous of degree 2 with respect to the $t$-coordinate; that is, under a constant scaling $t\to s t$ the metric transforms as $\eta \to s^{2}\eta$, or in the infinitesimal form,
\begin{equation}
\label{eq:homo}
{\cal L}_{\un T}\eta= 2\eta \,,\qquad\un T=t\un\pa_{t}\,.
\end{equation}
We will retain this property also for Ricci-flat ambient metrics and the Weyl-ambient metric. For relaxation of this homogeneity condition, see \cite{graham2008inhomogeneous}.

\subsection{Ricci-Flat Ambient Metrics}
The flat ambient metric combines hyperbolic metrics and their conformal boundaries in a unified  framework. Before we describe its utility, we will review the  generalization of flat ambient metrics to Ricci-flat ambient metrics. This will allow us to consider $(d+1)$-dimensional asymptotically locally Anti-de Sitter (AlAdS) spaces which are especially relevant in holographic theories.
\par
The main observation that allows an extension to Ricci-flat ambient metrics is that \eqref{Flat_Ambient_4} can be generalized in the following form:
\begin{equation}\label{R_Flat_Ambient}
\ti g= -\td \ell^2 + \frac{\ell^2}{L^2} g^{+}_{\mu\nu}(x)\td x^{\mu}\td x^{\nu},\qquad \mu,\nu=1,\cdots d+1\,,\quad \ell>0\,,
\end{equation}
where now $g^{+}(x)$ is an arbitrary $(d+1)$-dimensional metric independent of $\ell$. We will refer to this $(d+1)$-dimensional geometry as the ``bulk''. The ambient Ricci tensor $\tilde Ric(\tilde g)$ can be decomposed in terms of the Ricci tensor of $g^+$ as \cite{Fefferman:2007rka,Graham:1991jqw}
\begin{equation}
\ti Ric(\ti g)= Ric(g^{+}) + \frac{d}{L^2} g^{+} \,.
\end{equation}
The right-hand side of the above equation can also be written as $G_{\mu\nu}(g^{+})+ \Lambda g^{+}_{\mu\nu}$ with $\Lambda= -\frac{d(d-1)}{2L^2}$.  Therefore, when the ambient metric $\ti g$ is Ricci-flat, $g^{+}$ is an Einstein metric and thus satisfies the vacuum Einstein equations.
\par
According to the Fefferman-Graham theorem \cite{AST_1985__S131__95_0,Graham:1991jqw}, any AlAdS Einstein metric can be expressed in the Fefferman-Graham form 
\begin{equation}
\label{eq:g+FG}
g^{+}= L^2\frac{\td z^2}{z^2}+ \frac{L^2}{z^2}\gamma_{ij}(x,z)\td x^{i}\td x^{j}\,,\qquad i,j=1,\cdots ,d\,,\quad z>0\,,
\end{equation}
which is a generalization of \eqref{eq:adsPoin}. Then, by a coordinate transformation $t=\ell/z$ and $\rho=-z^{2}/2$, the metric \eqref{R_Flat_Ambient} takes the form
\begin{equation}\label{ambient_metric}
\ti g = 2\rho \td t^2 + 2t\td t\td\rho + t^2 \gamma_{ij}(x,\rho)\td x^{i}\td x^{j} \,,\qquad t>0\,.
\end{equation}
We can see that the flat ambient metrics \eqref{Flat_Ambient_2} and \eqref{Flat_Ambient_3} are nothing but special cases of \eqref{ambient_metric} when $\tilde g=\eta$. The codimension-2 metric is now generalized to an arbitrary $\gamma_{ij}(x,z)$ whose corresponding $g^+$ in \eqref{eq:g+FG} is an Einstein metric.
\par
Note that the advantages of the ambient coordinate system $\{t,x^i,\rho\}$ mentioned before for the flat ambient space are now carried over to the Ricci-flat case. One can see that the surface at $\rho=0$ is still a null hypersurface, denoted by $\cal N$, which is a coordinate singularity in the original $\{\ell,x^i,z\}$ coordinate system. Hence, the ambient coordinate system permits one to extend the spacetime region to include an open neighborhood of the null surface ${\cal N}$. Denoting the extended spacetime manifold as $\tilde M$, then ${\cal N}$ is a hypersurface in $\tilde M$ parametrized by $\{t,x^{i}\}$, which furnishes a conformal class $[\gamma]$ of codimension-2 metrics. Suppose $M$ is a $d$-dimensional manifold equipped with the conformal class $[\gamma]$, then $(\tilde M,\tilde g)$ is called the $(d+2)$-dimensional ambient space of $(M,[\gamma])$. 
\par
Being part of the Ricci-flat ambient space, ${\cal N}$ can be regarded as an initial value surface. Then given the initial data $\gamma_{ij}(x,\rho)|_{\rho=0}$, the Ricci-flatness condition can be used to ``propagate'' the metric beyond the initial surface to a neighborhood around $\rho=0$. That is, the Ricci-flatness condition $\ti Ric(\ti g)=0$ is a set of differential equations for $\tilde g_{ij}(x,\rho)$, which can be solved iteratively in a series around $\rho=0$ given the initial value $\ti g(x,\rho)|_{\rho=0}$. The initial value problem for the Ricci-flat ambient space has been defined and evaluated rigorously in \cite{Fefferman:2007rka}, the results of which will be carried over to the Weyl-ambient space in later sections.

\subsection{Weyl-Ambient Metric}
Now we are ready to introduce the Weyl-ambient metric. We start from the $(d+2)$-dimensional ambient metric in the form of \eqref{R_Flat_Ambient}. The expression of $g^+$ in \eqref{eq:g+FG} is the FG ansatz for an AlAdS spacetime, which is not  preserved under a Weyl diffeomorphism $z\to z/{\cal B}(x)$, $x^i\to x^i$. This motivated the authors of \cite{Ciambelli:2019bzz}  to introduce the Weyl-Fefferman-Graham (WFG) gauge by adding an additional mode $a_\mu$ to \eqref{eq:g+FG} as follows:
\begin{equation}
\label{eq:g+WFG}
g ^+_{\text{WFG}}=L^2\left(\frac{\td z}{z} - a_{i}(x,z)\td x^{i}\right)^{2}+ \frac{L^2}{z^2}\gamma_{ij}(x,z)\td x^{i}\td x^{j}\,,\qquad z>0\,.
\end{equation}
Now we substitute the $g^+$ in \eqref{R_Flat_Ambient} with the WFG ansatz \eqref{eq:g+WFG}, then transforming back to the ambient coordinates $\{t,x^i,\rho\}$, we obtain the \emph{Weyl-ambient metric}\footnote{The form of the metric \eqref{Weyl_ambient} with $a_i$ independent of $\rho$ has been introduced in \cite{Manvelyan:2007tk} in the context of Ricci gauging. We thank Omar Zanusso for pointing this out to us.}
\begin{equation}\label{Weyl_ambient}
\ti g = 2\rho \td t^2 + 2t^2\td\rho\left(\frac{\td t}{t} + a_{i}(x,\rho)\td x^{i}\right) + t^2 g_{ij}(x,\rho)\td x^{i}\td x^{j}\,,\qquad t>0\,,
\end{equation}
where $g_{ij}(x,\rho):= \gamma_{ij}(x,\rho)- 2\rho a_{i}(x,\rho)a_{j}(x,\rho)$. We call the pseudo-Riemannian space $(\tilde M,\tilde g)$  a \emph{Weyl-ambient space}. Having the form of the Weyl-ambient metric,  the \emph{ambient Weyl diffeomorphism}\footnote{In terms of the coordinates $\ell,z$, the ambient Weyl diffeomorphism acts as $(\ell',x'^i,z')=(\ell, x^i, {\cal B}(x)^{-1}z)$.}
\begin{equation}
\label{eq:Weyldiff}
t'={\cal B}(x)t\,,\qquad x'^i=x^i\,,\qquad \rho'={\cal B}(x)^{-2}\rho
\end{equation}
induces a change in the constituents $a_i$ and $\gamma_{ij}$ of the form
\begin{equation}
\label{eq:Weyldiff2}
a'_i(x',\rho')=a_i(x,\rho)-\p_i\ln{\cal B}(x)\,,\qquad\gamma'_{ij}(x',\rho')={\cal B}(x)^{-2}\gamma_{ij}(x,\rho)\,.
\end{equation}
If we regard the AlAdS bulk as a hypersurface of the Weyl-ambient space, the above transformation gives rise to the Weyl diffeomorphism which preserves the WFG ansatz. In addition, we want to point out that just as the ambient metric \eqref{ambient_metric} is homogeneous with respect to $t$, the homogeneity property \eqref{eq:homo}  also pertains for the Weyl-ambient metric \eqref{Weyl_ambient} since both $a_{i}(x,\rho)$ and $\gamma_{ij}(x,\rho)$ are independent of $t$.   This homogeneity property will be repeatedly used throughout this paper. In the following we use this property in order to show how an induced Weyl class arises from the Weyl-ambient metric; it is also crucial for the bottom-up construction and for proving Propositions~\ref{prop1} and \ref{Proposition_6.5_Graham}.
\par
The Ricci-flatness condition $\ti Ric(\ti g)=0$ for the Weyl-ambient metric \eqref{Weyl_ambient}, similar to that for the ambient metric \eqref{ambient_metric}, is a set of differential equations for $\tilde g_{ij}(x,\rho)$ which can be solved order by order in a neighborhood of $\rho=0$ given the initial value $\ti g_{ij}(x,\rho)|_{\rho=0}$. To be precise, in a neighborhood of $\rho=0$ we can expand $\gamma_{ij}$ and $a_{i}$ as\footnote{There will be a second series starting from the $\rho^{d/2}$ order in the expansion \eqref{eq:gexpan}:
\begin{align*}
\gamma_{ij}(x,\rho)&=( \gamma^{(0)}_{ij}(x)+ \gamma^{(1)}_{ij}(x)\rho+ \cdots)+\rho^{{d/2}}( \pi^{(0)}_{ij}(x)+ \pi^{(1)}_{ij}(x)\rho  + \cdots)\,.
\end{align*}
However, to solve for the second series in $\gamma_{ij}$ order by order one needs the interior data $\pi_{ij}^{(0)}$ of the ambient space. This is related to the obstruction tensors as explained in \cite{Jia:2021hgy}.}
\begin{align}
\label{eq:gexpan}
\gamma_{ij}(x,\rho)&= \gamma^{(0)}_{ij}(x)+ \gamma^{(1)}_{ij}(x)\rho +\gamma^{(2)}_{ij}(x)\rho^2 + \cdots\,,\\
\label{eq:aexpan} 
a_{i}(x,\rho)&= a^{(0)}_{i}(x) + a^{(1)}_{i}(x)\rho + a^{(2)}_{i}(x)\rho^2 + \cdots\,.
\end{align}
From the equation $\ti Ric(\ti g)=0$, one can solve for $\gamma^{(n)}_{ij}(x)$ in terms of $\gamma^{(k)}_{ij}(x)$ and $a^{(k)}_{i}(x)$ with $k$ up to $n-1$. However, the modes $a^{(n)}_{i}(x)$ are not determined by the Ricci flatness condition and hence we regard $a^{(k)}_{i}(x,\rho)$ as input data. This initial value problem will be examined in detail in Section~\ref{sec:bottomup} after the Weyl-ambient space is defined in terms of the Weyl structure and the ansatz in \eqref{Weyl_ambient} will be shown to be the uniquely determined Weyl-ambient metric for any given $\gamma^{(0)}_{ij}(x)$ and $a_i(x,\rho)$. 
\par
From the transformation \eqref{eq:Weyldiff2} and the expansions \eqref{eq:gexpan} and \eqref{eq:aexpan}, we can see that $\gamma_{ij}^{(k\geqslant0)}$ and $a^{(k\geqslant1)}_{i}(x)$ transform covariantly under the ambient Weyl diffeomorphism \eqref{eq:Weyldiff}, with Weyl weights $2k-2$ and $2k$, respectively:
\be
\gamma^{(k\geqslant0)}_{ij}(x)\to{\cal B}(x)^{2k-2}\gamma^{(k\geqslant0)}_{ij}(x)\,,\qquad a_i^{(k\geqslant1)}(x)\to{\cal B}(x)^{2k}a_i^{(k\geqslant1)}(x)\,.
\ee
On the other hand, $a^{(0)}_i$ transforms as $a^{(0)}_i\to a^{(0)}_i-\p_i\ln{\cal B}$. Therefore, we should anticipate that $a^{(0)}_i$ can be interpreted as a Weyl connection on the codimension-2 geometry.  An important result of \cite{Ciambelli:2019bzz} was to show that the bulk metric of an AlAdS spacetime in the WFG gauge provides a Weyl geometry on the conformal boundary. In the next section we will show that by introducing $a_i(x,\rho)$ in the ambient metric, we indeed obtain a Weyl geometry at codimension-2, where $\gamma^{(0)}_{ij}$ and $a^{(0)}_i$ play the role of a metric and a Weyl connection, respectively.
\par
Closing this section, we remark that the codimension-1 surface ${\cal N}$ at $\rho=0$ is again a null surface parametrized by $(t,x)$ with $t\in \mathbb{R}_{+}$, just like the case of the ambient metric \eqref{ambient_metric}. This surface in fact has the structure of a line bundle with each fibre parametrized by $t$, which turns out to be a principal bundle with the structure group $\bb R_+$. The new ingredient $a_i$ in the Weyl-ambient metric \eqref{Weyl_ambient} induces naturally a connection on this principal bundle, represented by $a^{(0)}_i=a_i|_{\rho=0}$. We will explore this in Section~\ref{sec:bottomup}.

\section{Weyl-Ambient Space: Top-Down Perspective}
\label{sec:topdown}
\subsection{Induced Weyl Geometry}\label{sec:WeylAmb}
The goal of this section is to analyze the Weyl-ambient metric from a top-down perspective before we introduce the more formal bottom-up construction of the Weyl-ambient space in Section \ref{sec:bottomup}. We will show explicitly that the Weyl-ambient metric \eqref{Weyl_ambient} leads to a Weyl geometry at codimension-2.
\par
Define a dual frame $\{\bm e^{P}\}$ on the $(d+2)$-dimensional manifold $\tilde M$ as follows:
\begin{align}
\label{e+-def}
\bm e^+&=\td t+ta_i(x,\rho)\td x^i\,,\qquad \bm e^i=\td x^i\,,\qquad\bm e^-=t\td\rho+\rho \td t-t\rho a_i(x,\rho)\td x^i\,,
\end{align}
where now $P=\{+,i,-\}$. In this frame the Weyl-ambient metric \eqref{Weyl_ambient} can be written as
\begin{align}
\label{eq:metricnull}
\tilde g=\bm e^+\otimes\bm e^-+\bm e^-\otimes\bm e^++t^2\gamma_{ij}\bm e^i\otimes\bm e^j\,.
\end{align}
It is easy to check that the 1-forms defined in \eqref{e+-def} are covariant under \eqref{eq:Weyldiff} and \eqref{eq:Weyldiff2}, and thus the form of $\tilde g$ in \eqref{eq:metricnull} is preserved under an ambient Weyl diffeomorphism. The corresponding frame $\{\un D_{P}\}$ of \eqref{e+-def} reads
\begin{align}
\label{eq:D+-i}
\un D_+&=\un\p_t-\frac{\rho}{t}\un\p_\rho\,,\qquad\un D_i=\un\p_i-ta_i(x,\rho)\un\p_t+2\rho a_i(x,\rho)\un\p_\rho\,,\qquad\un D_-=\frac{1}{t}\un\p_\rho\,.
\end{align}
From \eqref{eq:metricnull} it is clear that $\un D_+$ and $\un D_-$ are null vectors. $\{\un D_i\}$ form a basis of a $d$-dimensional \emph{distribution} $C_d\subset T\tilde M$, defined as
\begin{align}
C_d=\big\{\un{\mathcal V}\in T\tilde M\,|\,i_{\un{\mathcal V}}\bm e^\pm=0\big\}\,.
\end{align}
It follows from \eqref{eq:D+-i} that
\begin{align}
[\un D_i,\un D_j]=-tf_{ij}\un D_++t\rho f_{ij}\un D_-\,,
\end{align}
where $f_{ij}=D_ia_j-D_ja_i$ is the curvature of $a_i(x,\rho)$. The Frobenius theorem implies that the distribution $C_d$ is integrable when $f_{ij}=0$, though we will not generally assume this to be the case. One should note that the codimension-1 distribution spanned by $\{\un D_i,\un D_+\}$ is integrable at $\rho=0$, and thus defines a codimension-1 subspace (see Appendix \ref{App:Null} for relevant details).

Suppose $M$ is a $d$-dimensional manifold with a local coordinate system $\{y^i\}$ on $U\subset M$, and a point $\tilde p\in\tilde M$ has coordinates $(t,x^i,\rho)$. One can consider the coordinate patch $\tilde U$ of the ambient coordinate system $\{t,x^i,\rho\}$ as a fibre bundle with the projection $\pi:\tilde U\to U$ such that $\pi(\tilde p)=p\in M$ has coordinates $y^i=x^i$, i.e.\ each fibre in $\tilde U$ is parametrized by $(t,\rho)$. For simplicity, in what follows we will refer to $\tilde U$ as $\tilde M$ and $U$ as $M$, and we will not distinguish $\{x^i\}$ and $\{y^i\}$. Now that we have a bundle structure $\pi:\tilde M\to M$, we can see that $a_i(x,\rho)$ plays the role of an Ehresmann connection that specifies the horizontal subspace $H_{\tilde p}=C_d|_{\tilde p}\subset T_{\tilde p}\tilde M$, which defines the horizontal lift $T_pM\to H_{\tilde p}$ with $\un\p_i\mapsto\un D_i$. In general then, we are describing an isolated surface.

\par
Since we have a bundle structure $\pi:\tilde M\to M$, each section defines an embedding $\phi:M\to\tilde M$ such that a point $p\in M$ with coordinates $x^i$ is mapped to $\phi(p)=(t(x),x^i,\rho(x))$. With the horizontal subspace defined, we have $\pi_*:H_p\to T_pM$ such that $\pi_*(\un D_i)=\un\p_i$. Now consider the embedding $\phi$ with $\phi(p)=(t=1,x^i,\rho=0)$. We can define an induced metric $\gamma_{ij}^{(0)}(x)$ on $M$ by ``pulling back''\footnote{Note that we abuse the term as this is technically not a standard pullback by the embedding $\phi$, because $\un D_i$ is not tangent to $\phi[M]$.} $\tilde g_{ij}(t,x,\rho)= \tilde g(\un D_i,\un D_j)$ from the subspace of $\tilde M$ at $t=1$ and $\rho=0$ similar to what we did for the flat ambient space:
\be
\label{eq:inducedmetr}
\gamma^{(0)}_{ij}= \tilde g_{ij}|_{t=1,\rho=0}\,.
\ee
Under the coordinate transformation \eqref{eq:Weyldiff} in $\tilde M$ induced by an ambient diffeomorphism, we can consider the pullback $\gamma'^{(0)}(x')$ of $\tilde g'(t',x',\rho')$ by $\phi'(p)=(t'=1,x'^i,\rho'=0)$:
\be
\label{eq:inducedmetr'}
\gamma'^{(0)}_{ij}= \tilde g'_{ij}|_{t'=1,\rho'=0}\,,
\ee
where $\tilde g'_{ij}=g'(\un D'_i,\un D'_j)$, with $\un D'_i=\un\p'_i-t'a'_i(x',\rho')\un\p'_t+2\rho' a'_i(x',\rho')\un\p'_\rho$. Since $\tilde g'_{ij}=t'^2\gamma'_{ij}(x',\rho')$, we have
\be
\gamma'^{(0)}_{ij}={\cal B}(x)^{-2} \tilde g'_{ij}|_{t'={\cal B}(x),\rho'=0}={\cal B}(x)^{-2} \tilde g_{ij}|_{t=1,\rho=0}={\cal B}(x)^{-2}\gamma^{(0)}_{ij}\,.
\ee
That is, under the ambient Weyl diffeomorphism in $\tilde M$, we obtain two induced metrics which are related by a Weyl transformation in $M$. Hence, the ambient Weyl diffeomorphisms acting on the surface $\rho=0$, namely the null surface $\cal N$, gives rise to a conformal class of metrics on $M$.\footnote{If one only performs a local scaling in the coordinate $t$, i.e.\ $t'=B(x)t, x'^i=x^i, \rho'=\rho$, then one can also get a conformal class of metrics from other constant-$\rho$ surfaces. However, to obtain the induced Weyl connection and a Weyl class, one needs to perform the ambient Weyl diffeomorphism, and thus needs the restriction of $\rho=0$.}

\par
Having a conformal class of induced metrics on $M$, now let us look at how a connection is induced from $\tilde M$ onto $M$. Suppose $\tilde\nabla$ is the Levi-Civita connection of the ambient space $(\tilde M,\tilde g)$, i.e.\ it is torsion-free and has zero metricity $\tilde\nabla_{\un D_P}\tilde g_{MN}=0$. The ambient connection coefficients $\tilde\Gamma^P{}_{MN}$ of $\tilde\nabla$ are defined with respect to the frame $\un D_M$ of $T\tilde M$ as:
\begin{align}
\tilde\nabla_{\un D_M}\un D_N=\tilde\Gamma^i{}_{MN}\un D_i+\tilde\Gamma^+{}_{MN}\un D_++\tilde\Gamma^-{}_{MN}\un D_-\,.
\end{align}
In the following discussion we will denote the covariant derivative $\tilde\nabla_{\un D_P}$ along $\un D_P$ as $\tilde\nabla_P$ for brevity ($P=+,i,-$); we emphasize that these are not however the coordinate frame components.
The ambient connection 1-form $\tilde{\bm\omega}^{M}{}_{N}=\tilde\Gamma^M{}_{PN}\bm e^P$ in this frame is then found to be (the matrix elements are arranged in the order of $+,i,-$)
\begin{align}
\tilde{\bm\omega}^{M}{}_{N}=&\left(\begin{array}{ccc}a_k & -t\psi_{kj}& 0  \\\frac{1}{t}(\delta_k{}^i-\rho\psi_k{}^i) & \tilde\Gamma^i{}_{kj}& \frac{1}{t}\psi_k{}^i  \\0 & -t(\gamma_{kj}-\rho\psi_{kj})& -a_k \end{array}\right)\bm e^k\nn\\
\label{eq:conn1form}
&+\left(\begin{array}{ccc}0 & \rho\varphi_j& 0  \\\frac{\rho^2}{t^2}\varphi^i &\frac{1}{t}(\delta_j{}^i-\rho\psi_{j}{}^i)& -\frac{\rho}{t^2}\varphi^i  \\0& -\rho^2\varphi_j & 0 \end{array}\right)\bm e^++\left(\begin{array}{ccc}0 & -\varphi_j & 0 \\-\frac{\rho}{t^2}\varphi^i &\frac{1}{t}\psi_j{}^i & \frac{1}{t^2}\varphi^i  \\0  & \rho \varphi_j & 0 \end{array}\right)\bm e^- \,,
\end{align}
where the upper $i,j$ indices are raised by $\gamma^{ij}\equiv(\gamma_{ij})^{-1}$, and
\begin{align}
\label{eq:psiphi}
\psi_{ij}&=\frac{1}{2}(\p_\rho \gamma_{ij}+f_{ij})\,,\qquad\varphi_i=\p_\rho a_i \,,\qquad f_{ij}=D_ia_j-D_ja_i\,,\\
\label{eq:Gammaijk}
\tilde\Gamma^{i}{}_{jk}&=\frac{1}{2}\gamma^{im}(D_j\gamma_{mk}+D_k\gamma_{jm}-D_m\gamma_{jk})-(a_j\delta^i{}_k+a_k\delta^i{}_j-a^i\gamma_{jk})\,.
\end{align}

We note that the Levi-Civita condition $\tilde\nabla_{i}\tilde g_{jk}=0$ evaluates to $\nabla_{i}\gamma_{jk}=2a_i\gamma_{jk}$, where $\nabla$ is the connection on the distribution $C_d$ induced by $\tilde\nabla$, with $\nabla_{i}\gamma_{jk}:=D_i\gamma_{jk}-\tilde\Gamma^m{}_{ij} \gamma_{m k}-\tilde\Gamma^m{}_{ik} \gamma_{jm}$. Hence, if we interpret $\gamma_{ij}$, i.e.\ $\tilde g_{IJ}$ restricted to the $i,j$ indices, as giving rise to a metric on the distribution $C_d$ spanned by $\{\un D_i\}$ in $\tilde M$, then the connection $\nabla$ on $C_d$ has a nonvanishing metricity $2a_i\gamma_{jk}$. Equivalently, we say that this connection has vanishing {\it Weyl metricity}, and it is therefore convenient and natural to introduce a connection $\hat\nabla$ on $C_d$, such that \[\hat\nabla_i\gamma_{jk}:=\nabla_i\gamma_{jk}-2a_i\gamma_{jk}=0.\] The vanishing of the Weyl metricity is a Weyl-covariant condition, whereas the vanishing of the usual metricity $\nabla_i\gamma_{jk}$ is not. More generally, for any tensor $T$ defined on $C_d$ (i.e.,\ $ T$ has no $+,-$ components) that transforms covariantly under an ambient Weyl diffeomorphism as $T(t,x^i,\rho)\to {\cal B}(x)^{w_T} T({\cal B}(x)^{-1}t,x^i,{\cal B}^{2}(x)\rho)$, the derivative
\be
\label{eq:hatnabla}
\hat\nabla_i T:=\nabla_iT+w_Ta_iT
\ee
will also transform covariantly with the same weight. For example, it follows from the definitions in \eqref{eq:psiphi} that $\varphi_i(x,\rho)\to{\cal B}(x)^2\varphi_i(x,{\cal B}(x)^2\rho)$ and $\psi_{ij}(x,\rho)\to\psi_{ij}(x,{\cal B}(x)^2\rho)$, and thus we can write their Weyl-covariant derivatives as
\be
\hat\nabla_i\varphi_j=\nabla_i\varphi_j+2a_i\varphi_j\,,\qquad 
\hat\nabla_i\psi_{jk}=\nabla_i\psi_{jk}\,.
\ee
From the above behavior of the induced connection on $C_d$, we can naturally expect that the induced connection on $M$ will give us a codimension-2 Weyl geometry. However, since $\{\un D_i\}$ is not an integrable distribution when $a_i$ is turned on, the connection coefficients \eqref{eq:Gammaijk} cannot be pulled back directly to $M$. As we will see below, this problem does not exist if we focus on the surface at $\rho=0$.

Notice that $\tilde\Gamma^{i}{}_{jk}$ does not depend on $t$, and thus at any value of $t$ at $\rho=0$, the induced connection coefficients can be expressed as
\begin{align}
\label{eq:Weylconn}
\Gamma_{(0)}^{i}{}_{jk}\equiv\tilde\Gamma^{i}{}_{jk}|_{\rho=0}=\frac{1}{2}\gamma_{(0)}^{im}(\p_j\gamma^{(0)}_{mk}+\p_k\gamma^{(0)}_{jm}-\p_m\gamma^{(0)}_{jk})-(a^{(0)}_j\delta^i{}_k+a^{(0)}_k\delta^i{}_j-a_{(0)}^i\gamma^{(0)}_{jk})\,.
\end{align}
To define an induced connection on $M$, let us take $t=1$ as a representative, i.e.\ take $\phi(M)$ to be a $d$-dimensional surface in $\tilde M$ at $\rho=0$ and $t=1$. At  first sight, the connection defined by \eqref{eq:Weylconn} is still an induced connection on the distribution spanned by $\{\un D_i\}$, which does not lie on the codimension-2 surface $\phi[M]$ when $a_i$ is turned on. However, when the dual frame $\{\bm e^P\}$ gets pulled back on $M$, we get $\{\bm e^i=\td x^i\}$, and the corresponding vector basis on $TM$ is $\{\un \p_i\}$. Hence, the ambient LC connection $\tilde\nabla$ defined on $T^*\tilde M$ induces a connection $\nabla^{(0)}$ on $T^*M$ in the following natural manner
\begin{align}
\nabla^{(0)}_{\un \p_j}\bm e^i\equiv\nabla_{\un D_j}\bm e^i|_{\rho=0,t=1}=-\Gamma_{(0)}^i{}_{jk}\bm e^k\,.
\end{align}
Then, $\nabla^{(0)}$ can also be defined on $TM$, which defines the parallel transport of a vector along a curve on $M$:
\begin{align}
\nabla^{(0)}_{\un\p_i}\un\p_j=\Gamma_{(0)}^k{}_{ij}\un\p_k\,.
\end{align}
In this way we get a connection $\nabla^{(0)}$ on $M$ whose connection coefficients are given by \eqref{eq:Weylconn}. This is a connection that satisfies $\nabla^{(0)}_i\gamma^{(0)}_{jk}=2a^{(0)}_i\gamma^{(0)}_{jk}$, i.e.\ it has vanishing Weyl metricity, and $a^{(0)}_i$ plays the role of a Weyl connection on $M$. One can also define a metricity-free connection $\hat\nabla^{(0)}$ on $M$ satisfying $\hat\nabla^{(0)}_{i}\gamma^{(0)}_{jk}=\nabla^{(0)}_i\gamma^{(0)}_{jk}-2a^{(0)}_i\gamma^{(0)}_{jk}=0$, which can be referred to as a Weyl-LC connection. 
\par
An ambient Weyl diffeomorphism in $\tilde M$ induces on $M$ a Weyl transformation $\gamma^{(0)}_{ij}\to{\cal B}^{-2}\gamma^{(0)}_{ij}$, $a^{(0)}_i\to a^{(0)}_i-\p_i\ln{\cal B}$.\footnote{If one considers a more general version of the diffeomorphism \eqref{eq:Weyldiff} where $x'=x'(x)$, then 
\begin{equation*}
\frac{\p x'^j}{\p x^i}a'^{(0)}_j(x')=a_i^{(0)}(x)-\p_i\ln{\cal B}(x)\,,\qquad\frac{\p x'^i}{\p x^k}\frac{\p x'^j}{\p x^l}\gamma_{kl}'^{(0)}(x')={\cal B}(x)^{-2}\gamma_{ij}^{(0)}(x)\,.
\end{equation*}
The transformation $(t,x^i,\rho)\to(t,x'^i(x),\rho)$ realizes the Diff$(M)$ part of the $\text{Diff}(M)\ltimes \text{Weyl}$ symmetry on $M$.} This means that we get a \emph{Weyl class} $[\gamma^{(0)},a^{(0)}]$, which is the equivalence class formed by all the pairs of $\gamma^{(0)}$ and $a^{(0)}$ that are connected by Weyl transformations, i.e.,\
\be
(\gamma^{(0)}_{ij},a^{(0)}_i)\sim({\cal B}(x)^{-2}\gamma^{(0)}_{ij},a^{(0)}_i-\p_i\ln{\cal B}(x))\,.
\ee
With the Weyl class defined on $M$, we obtain a $d$-dimensional Weyl manifold $(M,[\gamma^{(0)},a^{0}])$ (see \cite{10.4310/jdg/1214429379,doi:10.1063/1.529582,Ciambelli:2019bzz}) induced by the Weyl-ambient space $(M,\tilde g)$, where the geometric quantities defined in terms of the Weyl connection are Weyl covariant. For example, one can define on $M$ the Weyl-Riemann tensor $\hat R_{(0)}^i{}_{jkl}$, Weyl-Ricci tensor $\hat R^{(0)}_{ij}$, Weyl-Ricci scalar $\hat R^{(0)}$, etc. (See Appendix A of \cite{Jia:2021hgy} for a review on Weyl geometry and the details of these Weyl-covariant quantities.)

\subsection{Weyl-Obstruction Tensors: First Order Formalism}
\label{sec:Frame}
A very useful property of the ambient metric introduced in \cite{2001math.....10271F} in the context of conformal geometry is the ability to construct conformal-covariant tensors from the ambient Riemann tensor, including the (extended) obstruction tensors, which were generalized to (extended) Weyl-obstruction tensors in \cite{Jia:2021hgy}. The Weyl-obstruction tensors on a $d$-dimensional Weyl manifold $(M,[\gamma^{(0)},a^{(0)}])$ were introduced in \cite{Jia:2021hgy} as the poles of the metric expansion of $\gamma_{ij}$ in the AlAdS bulk. However, they can also be defined in a more explicit way from the Weyl-ambient space $(\tilde M,\tilde g)$. Before we construct the Weyl-ambient space rigorously, we would like to demonstrate how the Weyl-obstruction tensors on $M$ can be derived from $(\tilde M,\tilde g)$ in the first order formalism using the frame introduced in \eqref{e+-def}. 
\par
Starting from the metric \eqref{eq:metricnull}, one can solve $\ti Ric(\ti g)=0$ order by order (which is equivalent to solving the Einstein equations in the AlAdS bulk \cite{Jia:2021hgy}) to find\footnote{Note that the $\gamma^{(k)}_{ij}$ and $a^{(k)}_{i}$ defined here correspond to $(-2)^k\gamma^{(2k)}_{\mu\nu}/L^{2k}$ and $(-2)^ka^{(2k)}_{\mu}/L^{2k}$ in \cite{Jia:2021hgy}, respectively.}
\begin{align}
\label{Pgf}
\gamma^{(1)}_{ij}={}&2\hat P_{(ij)}=2\hat P_{ij}-f^{(0)}_{ij}\,.\\
\label{g4}
\gamma^{(2)}_{ij}={}&\hat{\Omega}^{(1)}_{ij}+\hat P^{k}{}_{i}\hat P_{kj}+\hat\nabla^{(0)}_{(i} a_{j)}^{(1)}\,,\\
\gamma^{(3)}_{ij}={}&\tfrac{1}{3}\hat{\Omega}^{(2)}_{ij}+\tfrac{4}{3}\hat\Omega^{(1)}_{k(i}\hat P^k{}_{j)} +\tfrac{2}{3}\hat\nabla^{(0)}_{(i}a^{(2)}_{j)}\nn+2a^{(1)}_{i}a^{(1)}_{j}-\tfrac{1}{3}a^{(1)}\cdot a^{(1)}\gamma^{(0)}_{ij}\\
\label{g6}
&+\tfrac{1}{3}P^{k}{}_{(i}\hat\nabla^{(0)}_{j)}a^{(1)}_k-\tfrac{1}{3}a_{(1)}^k(\hat\nabla^{(0)}_i\hat P_{kj}+\hat\nabla^{(0)}_i\hat P_{jk}-\hat\nabla^{(0)}_k\hat P_{ji}+2\hat\nabla^{(0)}_j\hat P_{ik}-2\hat\nabla^{(0)}_k\hat P_{ij})\,,
\end{align}
where $f^{(0)}_{ij}=\p_ia^{(0)}_j-\p_ja^{(0)}_i$, and $\hat P_{ij}$ is the Weyl-Schouten tensor on $(M,[\gamma^{(0)},a^{(0)}])$. Treating $d$ as an continuous complex variable, the solution for each $\gamma^{(k\geqslant2)}_{ij}$ has a pole at $d=2k$ (see Proposition \ref{prop1}) represented by $\hat{\Omega}^{(k-1)}_{ij}$. For now one should simply regard $\hat{\Omega}^{(k-1)}_{ij}$ in the above equations as denoting the pole terms of $\gamma^{(k)}_{ij}$ at $d=2k$ ($\hat P_{ij}$ also represents the ``pole'' of $\gamma^{(1)}_{ij}$ at $d=2$, which identically vanishes in $2d$). Later in this section we will recognize them as extended Weyl-obstruction tensors through a precise definition. In terms of $\gamma^{(0)}_{ij}$, these quantities can be written as
\begin{align}
\label{eq:Wsch}
\hat P_{ij}={}&\frac{1}{d-2}\bigg(\hat R^{(0)}_{ij}-\frac{\hat R^{(0)}}{2(d-1)}\gamma^{(0)}_{ij}\bigg)\,,\\
\label{eq:Omega1}
\hat\Omega^{(1)}_{ij}={}&\frac{1}{d-4}\Big(-\hat\nabla^{(0)}_k\hat\nabla_{(0)}^k \hat P_{ij}+\hat\nabla^{(0)}_k\hat\nabla^{(0)}_{j} \hat P_{i}{}^{k}+\hat W^{(0)}_{kjil}\hat P^{lk}\Big)\,,\\
\label{eq:Omega2}
\hat\Omega^{(2)}_{ij}={}&\frac{1}{d-6}\Big(-\hat\nabla_{(0)}^k\hat\nabla^{(0)}_k\hat\Omega^{(1)}_{ij}+2\hat W^{(0)}_{kjil}\hat\Omega_{(1)}^{lk}+4\hat P\hat\Omega^{(1)}_{ij}-2\hat P_{k(j}\hat\Omega_{(1)}^{k}{}_{i)}+2\hat\Omega_{(1)}^{k}{}_{(i}\hat P_{j)k}\nn\\
&\qquad\quad+2\hat\nabla_{(0)}^k\hat  C_{kl(i}\hat P^{l}{}_{j)} -2\hat P^{kl}\hat\nabla^{(0)}_{(i}\hat  C_{j)lk}+4\hat P^{(kl)}\hat\nabla^{(0)}_l\hat C_{(ij)k}+2\hat\nabla^{(0)}_l\hat P^{kl}\hat C_{(ij)k}\nn\\
&\qquad\quad-2\hat C^{k}{}_{i}{}^{l}\hat C_{ljk}+2 \hat\nabla_{(0)}^l\hat P^k{}_{(i}\hat C_{j)kl}-2\hat W^{(0)}_{k(ji)l}\hat P^{l}{}_m\hat P^{mk}\Big)\,,
\end{align}
where $\hat W_{(0)}^i{}_{jkl}$ is the Weyl curvature tensor and {$\hat C_{ijk}\equiv\hat\nabla_{k}^{(0)}\hat P_{ij}-\hat\nabla_{j}^{(0)}\hat P_{ik}$ is the Weyl-Cotton tensor}. Note that indices are lowered with $\gamma^{(0)}_{ij}$ as necessary.
\par
We first look at how the Weyl-Schouten tensor $\hat P_{ij}$ is derived from the Weyl-ambient geometry. Consider the expansion of $\gamma_{ij}$. At $\rho=0$ and $t=1$, the ambient connection 1-form \eqref{eq:conn1form} becomes
\begin{align}
\tilde{\bm\omega}_{(0)}^{M}{}_{N}=&\left(\begin{array}{ccc}a^{(0)}_k & -\hat P_{jk} & 0  \\ \delta^i{}_k & \Gamma_{(0)}^i{}_{kj}& \hat P^i{}_k  \\0 & -\gamma^{(0)}_{jk} & -a^{(0)}_k \end{array}\right)\bm e^k
+\left(\begin{array}{ccc}0 & 0& 0  \\0 &\delta_j{}^i& 0  \\0& 0 & 0 \end{array}\right)\bm e^++\left(\begin{array}{ccc}0 & 0 & 0 \\ 0&\psi_j{}^i & 0  \\0  &0& 0 \end{array}\right)\bm e^- \,.
\end{align}
Notice that the first term, which is the pullback of $\tilde{\bm\omega}_{(0)}^{M}{}_{N}$ from $T^*\tilde M$ to $T^*M$, can be recognized as the Cartan normal conformal connection \cite{10.2996/kmj/1138845392,kobayashi2012transformation}. From here we can see that the Weyl-Schouten tensor of the boundary appears in the leading order ($\rho=0$) of the ambient connection. 
\par
From the connection 1-form \eqref{eq:conn1form}, we can also find the ambient curvature 2-form in the frame $\{\bm e^+,\bm e^i,\bm e^-\}$:
\begin{align}
\label{eq:curv2form}
\tilde{\bm R}^{M}{}_{N}=&\left(\begin{array}{ccc}0 & -t\bm{\mathcal C}_j & 0  \\-\frac{\rho}{t}\bm{\mathcal C}^i & {\bm{\mathcal W}}^i{}_{j} & \frac{1}{t}\bm{\mathcal C}^i \\0 & \rho t\bm{\mathcal C}_j & 0 \end{array}\right)+\left(\begin{array}{ccc}0 &\bm{\mathcal B}_j & 0 \\\frac{\rho}{t^2}\bm{\mathcal B}^i &\frac{1}{t}{\cal C}_{kj}{}^i\bm e^k & -\frac{1}{t^2}\bm{\mathcal B}^i\\0 & -\rho\bm{\mathcal B}_j & 0 \end{array}\right)\wedge(\bm e^--\rho\bm e^+)\,.
\end{align}
Here we defined $\bm{\mathcal B}_i={\cal B}_{ij}\bm e^j$, $\bm{\mathcal C}_i=\frac{1}{2}{\cal C}_{ikj}\bm e^j\wedge \bm e^k$, $\bm{\mathcal W}^i{}_j={\cal W}^i{}_{jkl}\bm e^k\wedge \bm e^l$, with
\begin{align}
\label{eq:Bij}
{\cal B}_{ij}&\equiv\p_\rho\psi_{ij}-\psi_{ik}\psi_j{}^k-\hat\nabla_i\varphi_j-2\rho\varphi_i\varphi_j\,,\\
{\cal C}_{ikj}&\equiv\hat\nabla_j\psi_{ki}-\hat\nabla_k\psi_{ji}-2\rho\varphi_if_{jk}\,,\\
{\mathcal W}^i{}_{jkl}&\equiv\bar R^i{}_{jkl}+\delta_j{}^if_{kl}-\delta_k{}^i\psi_{lj}-\psi_k{}^i\gamma_{lj}+\delta_l{}^i\psi_{kj}+\psi_l{}^i\gamma_{kj}+2\rho(\psi_k{}^i\psi_{lj}-\psi_l{}^i\psi_{kj}-\psi_j{}^if_{kl})\,,
\end{align}
where $\hat\nabla$ was introduced in \eqref{eq:hatnabla}, and
\begin{align}
\bar R^{i}{}_{jkl}=&D_{k}\ti \Gamma^{i}{}_{lj}-D_{l}\ti \Gamma^{i}{}_{kj}+\ti \Gamma^{i}{}_{km}\ti \Gamma^{m}{}_{lj}-\ti \Gamma^{i}{}_{lm}\ti \Gamma^{m}{}_{kj}\,.
\end{align}
Plugging in \eqref{Pgf} and \eqref{g4} from the $\rho$-expansion of $\gamma_{ij}$, one obtains at the leading order
\begin{align}
\label{eq:BCW}
{\cal B}_{ij}^{(0)}&=\hat\Omega_{ij}^{(1)}\,,\qquad{\cal C}_{ikj}^{(0)}=\hat C_{ijk}\,,\qquad
{\mathcal W}_{(0)}^i{}_{jkl}=\hat W_{(0)}^i{}_{jkl}\,.
\end{align}
Therefore, when pulled back from $\tilde M$ to $M$ the Riemann curvature of the Weyl-ambient space gives us on $M$ the Weyl tensor $\hat W_{(0)}^i{}_{jkl}$, Weyl-Cotton tensor $\hat C_{ijk}$ and the tensor $\hat\Omega_{ij}^{(1)}$ we obtained in \eqref{eq:Omega1} as follows:
\begin{align}
\label{eq:ambRiem}
\tilde R_{-ij-}|_{\rho=0,t=1}=\hat{\Omega}^{(1)}_{ij}\,,\qquad \tilde R_{-ijk}|_{\rho=0,t=1}=\hat C_{ijk}\,,\qquad\tilde R_{ijkl}|_{\rho=0,t=1}=\hat{W}^{(0)}_{ijkl}\,.
\end{align}
The corresponding curvature 2-form at $\rho=0, t=1$ can be expressed as
\begin{align}
\label{ea:ambRiem0}
\tilde{\bm R}_{(0)}^{M}{}_{N}=&\left(\begin{array}{ccc}0 & -\hat{\bm{C}}_j & 0  \\0 & \hat{\bm{W}}_{(0)}^i{}_{j} & \hat{\bm{C}}^i \\0 & 0 & 0 \end{array}\right)+\left(\begin{array}{ccc}0 &\hat{\bm\Omega}^{(1)}_j & 0 \\0 &\hat{C}_{kj}{}^i\bm e^k & -\hat{\bm\Omega}_{(1)}^i\\0 & 0 & 0 \end{array}\right)\wedge\bm e^-\,,
\end{align}
where $\hat{\bm\Omega}^{(1)}_i ={\hat\Omega}^{(1)}_{ij}\bm e^j$, $\hat{\bm C}_i=\frac{1}{2}\hat{C}_{ikj}\bm e^j\wedge \bm e^k$, $\hat{\bm W}_{(0)}^i{}_j={\hat W}^i{}_{jkl}\bm e^k\wedge \bm e^l$. As expected, the first matrix in \eqref{ea:ambRiem0}, which represents the components of $\tilde{\bm R}_{(0)}^{M}{}_{N}$ in the $\bm e^i\wedge\bm e^j$ directions, is the curvature 2-form of the Cartan normal connection. The $\bm e^i\wedge\bm e^-$ components, on the other hand, give rise to the tensor $\hat \Omega^{(1)}_{ij}$ on $M$, which is expected to be the first extended Weyl-obstruction tensor. This implies that we can define the extended Weyl-obstruction tensors on the $d$-dimensional manifold $M$ by means of the $(d+2)$-dimensional Weyl-ambient space. Before getting to that, we first provide the following proposition, which shows that diffeomorphism-covariant tensors in the Weyl-ambient space are Weyl-covariant tensors when pulled back to $M$.

\begin{prop}
\label{prop1}
Let $IJKLM_{1}\dots M_{r}$ be a list of indices, $s_{+}$ of which are $+$, $s_{M}$ of which correspond to $x^i$, and $s_{-}$ of which are $-$, then under the ambient Weyl diffeomorphism \eqref{eq:Weyldiff}, we have
\begin{align}
\label{eq:prop1}
\tilde\nabla_{M_1}\cdots\tilde\nabla_{M_r} \tilde R'_{IJKL}|_{\rho'=0,t'=1}={\cal B}(x)^{2s_{-}-2} \tilde\nabla_{M_1}\cdots\tilde\nabla_{M_r} \tilde R_{IJKL}|_{\rho=0,t=1}\,.
\end{align}
\end{prop}
\begin{proof}
Under the ambient Weyl diffeomorphism \eqref{eq:Weyldiff}, the vector basis $\{\un D_P\}$ transforms as
\be
\un D'_+={\cal B}(x)^{-1}\un D_+\,,\qquad\un D'_i=\un D_i\,,\qquad\un D'_-={\cal B}(x)\un D_-\,,
\ee
where
\begin{align}
\label{eq:D+-i'}
\un D'_+&=\un\p'_t-\frac{\rho'}{t'}\un\p'_\rho\,,\qquad\un D'_i=\un\p'_i-t'a'_i(x',\rho')\un\p'_t+2\rho' a'_i(x',\rho')\un\p'_\rho\,,\qquad\un D'_-=\frac{1}{t'}\un\p'_\rho\,.
\end{align}
Hence, 
\begin{align}
\label{eq:R'R}
\tilde\nabla_{M_1}\cdots\tilde\nabla_{M_r} \tilde R'_{IJKL}|_{\rho'=0,t'={\cal B}(x)}={\cal B}(x)^{s_{-}-s_+} \tilde\nabla_{M_1}\cdots\tilde\nabla_{M_r} \tilde R_{IJKL}|_{\rho=0,t=1}\,.
\end{align}
Noticing the fact that $\tilde g$ is homogeneous in $t$ with degree 2, and considering the $t$-dependence of $\un D_+$ and $\un D_-$ in \eqref{eq:D+-i}, we have
\begin{align}
\label{eq:R'R'}
\tilde\nabla_{M_1}\cdots\tilde\nabla_{M_r} \tilde R'_{IJKL}|_{\rho'=0,t'=1}={\cal B}(x)^{s_{-}+s_+-2} \tilde\nabla_{M_1}\cdots\tilde\nabla_{M_r} \tilde R'_{IJKL}|_{\rho'=0,t'={\cal B}(x)}\,.
\end{align}
Combining \eqref{eq:R'R} and \eqref{eq:R'R'} we obtain \eqref{eq:prop1}.
\end{proof}
Since diffeomorphism-covariant tensors can be constructed out of the Riemann tensor and its covariant derivatives \cite{Graham2007jet}, this proposition implies that the pullback of an ambient tensor $\tilde T_{M_1\cdots M_k}$ to $M$:
\be
T_{i_1\cdots i_{s_M}}\equiv \tilde T_{M_1\cdots M_k}|_{\rho=0,t=1},
\ee
is Weyl covariant with Weyl weight $2s_{-}-2$, where among the indices $M_1\cdots M_k$, $s_-$ of which are $-$, and $s_M$ of which correspond to $x^i$. For instance, from Proposition \ref{prop1} we can see that the tensors obtained in \eqref{eq:ambRiem} are all Weyl-covariant tensors on $M$, and the Weyl weights of $\hat{\Omega}^{(1)}_{ij}$, $\hat C_{ijk}$ and $\hat{W}^{(0)}_{ijkl}$ can be read off to be 2, 0, and $-2$, respectively, which are indeed the correct Weyl weights (see Table 1 in Appendix A of \cite{Jia:2021hgy}).
\par
As a special kind of Weyl-covariant tensor, we introduce the extended Weyl-obstruction tensors as follows.
\begin{defn}\label{def1}
Suppose $k$ is a positive integer. The $k^{th}$ extended Weyl-obstruction tensor $\hat \Omega^{(k)}_{ij}$ is defined as
\begin{align}
\label{eq:def2}
\hat \Omega^{(k)}_{ij}= \underbrace{\tilde\nabla_-\cdots\tilde\nabla_-}_{k-1} \tilde R_{- ij- }|_{\rho=0,t=1}.
\end{align}
\end{defn}
Some properties of Weyl-obstruction tensors can be readily seen from the above definition. From the symmetry of the Riemann tensor we can see that $\hat \Omega^{(k)}_{ij}$ is a symmetric tensor. It follows from Proposition~\ref{prop1} that $\hat \Omega^{(k)}_{ij}$ is Weyl covariant with Weyl weight $2k$. Also, from the Ricci-flatness condition we obtain that $\tilde g^{IJ}\tilde\nabla_{M_1}\cdots\tilde\nabla_{M_r}\tilde R_{IKJL}=0$, which gives rise to $\gamma_{(0)}^{ij}\hat \Omega^{(k)}_{ij}=0$, i.e.\ $\hat \Omega^{(k)}_{ij}$ is traceless.
\par
We have seen in \eqref{eq:BCW} that when $k=0$, this definition gives the $\hat\Omega^{(1)}_{ij}$ in \eqref{eq:Omega1}. By computing $\tilde\nabla_-\tilde R_{-ij-}$, one also finds that $\hat \Omega^{(2)}_{ij}$ defined in this way gives exactly the expression in \eqref{eq:Omega2} (see Appendix \ref{App:Null}). Notice again that before Definition \ref{def1} (also in \cite{Jia:2021hgy}), when we referred to $\hat\Omega^{(k)}_{ij}$ as the $k^{th}$ extended Weyl-obstruction tensor, we simply meant that they denote the pole of $\gamma^{(k+1)}_{ij}$ at $d=2k+2$; however, there is an ambiguity since the pole can be shifted by a finite term, and so that was not a precise definition for extended Weyl-obstruction tensors. Now the $\hat\Omega^{(k)}_{ij}$ defined through the Weyl-ambient space is uniquely determined. The proposition below will show that each $\hat\Omega^{(k)}_{ij}$ defined through the Weyl-ambient space indeed has a pole at $d=2k+2$, whose residue is the same as the pole in $\gamma^{(k+1)}_{ij}$. Therefore, the ambiguity of the pole in $\gamma^{(k+1)}_{ij}$ can be fixed by taking it to be the extended Weyl-obstruction tensor in Definition \ref{def1}.

\begin{prop}
\label{prop:obspole}
Let $k\geqslant 2$ be an integer. Both the extended Weyl-obstruction tensor $\hat\Omega^{(k-1)}_{ij}$ and $\gamma^{(k)}_{ij}$ in the expansion \eqref{eq:gexpan} have a simple pole at $d=2k$. The residues satisfy
\be
\textnormal{Res}_{d=2k}\hat\Omega^{(k-1)}_{ij}=\frac{k!}{2}\textnormal{Res}_{d=2k}\gamma^{(k)}_{ij}\,.
\ee
More specifically, $\hat\Omega^{(k-1)}_{ij}$ has the following form:
\be
\hat\Omega^{(k-1)}_{ij}=\frac{(-1)^{k-1}\Gamma(d/2-k)}{2^{k-1}\Gamma(d/2-1)}(\Delta_{(0)}^{k-1}\hat P_{ij}-\Delta_{(0)}^{k-2}\hat\nabla^{(0)}_i\hat\nabla^{(0)}_kP_{j}{}^k+\cdots)\,,
\ee
where $\Delta_{(0)}\equiv\hat\nabla^{(0)}_k\hat\nabla_{(0)}^k$ and the ellipsis represents the terms with fewer number of $\hat\nabla^{(0)}$. The terms inside the brackets represent the Weyl-obstruction tensor.
\end{prop}
\begin{proof}
First, let us show that $\gamma^{(k\geqslant2)}_{ij}$ has a pole at $d=2k$, which has the form
\be
\label{eq:gammapole}
\gamma^{(k)}_{ij}=\frac{(-1)^{k-1}\Gamma(d/2-k)}{2^{k-2}k!\Gamma(d/2-1)}(\Delta_{(0)}^{k-1}\hat P_{ij}-\Delta^{k-2}_{(0)}\hat\nabla^{(0)}_i\hat\nabla^{(0)}_kP_{j}{}^k+\cdots)\,.
\ee
We have seen this previously for $k=2$ and $3$. Using mathematical induction, now we will prove the following equation for $k\geqslant2$:
\be
\label{eq:psipole}
(d-2k)\p^{k-1}_{\rho}\psi_{ji}=\frac{(-1)^{k-1}\Gamma(d/2-k+1)}{2^{k-2}\Gamma(d/2-1)}(\Delta^{k-1}\psi_{ji}-\Delta^{k-2}\hat\nabla_i\hat\nabla_k\psi^k{}_{j}+{\cdots})+2\rho\p^{k}_{\rho}\psi_{ij}+O(\rho)\,,
\ee
where $\Delta\equiv\hat\nabla_k\hat\nabla^k$. This relation leads to \eqref{eq:gammapole} when $\rho=0$ since $\psi_{ij}=\frac{1}{2}(\p_\rho\gamma_{ij}+f_{ij})$ (the $f_{ij}$ in the left-hand side are combined in the ellipsis). Differentiating the Ricci-flatness condition of the form \eqref{eq:Rij0} with respect to $\rho$ and use the expression \eqref{eq:Rijprho} we can see that
\be
\label{eq:psipole2}
(d-4)\p_{\rho}\psi_{ji}=-(\Delta\psi_{ji}-\hat\nabla_i\hat\nabla_k\psi^k{}_{j}+{\cdots})+2\rho\p^2_{\rho}\psi_{ij}+O(\rho)\,,
\ee
which is \eqref{eq:psipole} in the case $k=2$. Now we assume \eqref{eq:psipole} holds for $k=n$. Differentiating both sides of \eqref{eq:psipole2}  for $n-1$ times with respect to $\rho$ yields
\be
(d-2n-2)\p^{n}_{\rho}\psi_{ji}=-\p^{n-1}_\rho(\Delta\psi_{ji}-\hat\nabla_i\hat\nabla_k\psi^k{}_{j}+\cdots)+2\rho\p^{n+1}_{\rho}\psi_{ij}+O(\rho)\,.
\ee
Note that $\p_\rho$ produces two $\hat\nabla$ when acting on $\psi$, while it only produces one $\hat\nabla$ when acting on $\tilde\Gamma^i{}_{jk}$, and thus when we commute $\p_\rho$ with $\hat\nabla$, the new terms only contribute to the ellipsis. Hence,
\begin{align}
&(d-2n-2)\p^{n}_{\rho}\psi_{ji}=-(\Delta\p^{n-1}_\rho\psi_{ji}-\hat\nabla_i\hat\nabla^k\p^{n-1}_\rho\psi_{kj}+\cdots)+2\rho\p^{n+1}_{\rho}\psi_{ij}+O(\rho)\nn\\
={}&\frac{(-1)^{n}\Gamma(d/2-n)}{2^{n-1}\Gamma(d/2-1)}(\Delta^{n}\psi_{ji}-\Delta^{n-1}\hat\nabla_i\hat\nabla_k\p_\rho\psi^k{}_{j}+\cdots)+2\rho\p^{n+1}_{\rho}\psi_{ij}+O(\rho)\,,
\end{align}
where {we used \eqref{eq:R+i0} and} the assumption that \eqref{eq:psipole} holds for $k=n$. This is exactly \eqref{eq:psipole} for $k=n+1$, and thus \eqref{eq:psipole} is proved for any $k\geqslant2$. Therefore, at $\rho=0$ we have
\be
\p^k_{\rho}\psi_{ji}|_{\rho=0}=\frac{(-1)^{k-1}\Gamma(d/2-k-1)}{2^{k}\Gamma(d/2-1)}(\Delta_{(0)}^{k}\hat P_{ij}-\Delta_{(0)}^{k-1}\hat\nabla^{(0)}_i\hat\nabla^{(0)}_k\hat P_{j}{}^k+\cdots)\,.
\ee
From \eqref{eq:curv2form} we can read off that
\be
\tilde R_{-ij-}={\cal B}_{ij}=\p_\rho\psi_{ij}-\psi_{ik}\psi_j{}^k-\hat\nabla_i\varphi_j-2\rho\varphi_i\varphi_j\,.
\ee
Hence, the Weyl-obstruction tensor $\hat\Omega^{(k)}_{ij}$ has the form
\begin{align}
\hat\Omega^{(k-1)}_{ij}&=\underbrace{\tilde\nabla_-\cdots\tilde\nabla_-}_{k-2}\tilde R_{-ij-}|_{\rho=0,t=1}=\p^{k-1}_\rho\psi_{ij}|_{\rho=0}+\cdots\nn\\
\label{eq:poleOmega}
&=\frac{(-1)^{k-1}\Gamma(d/2-k)}{2^{k-1}\Gamma(d/2-1)}(\Delta_{(0)}^{k-1}\hat P_{ij}-\Delta_{(0)}^{k-2}\hat\nabla^{(0)}_i\hat\nabla^{(0)}_kP_{j}{}^k+{\cdots})\,,
\end{align}
where finite terms at $d=2k$ are shifted into the pole. On the other hand, from \eqref{eq:poleOmega} we also have
\begin{align}
\text{Res}_{d=2k}\hat\Omega^{(k-1)}_{ij}&=\text{Res}_{d=2k}\p^k_\rho\psi_{ij}|_{\rho=0}=\frac{k!}{2}\text{Res}_{d=2k}\gamma^{(k)}_{ij}\,,
\end{align}
where in the second equality we considered that $f_{ij}$ does not contribute to the pole.
\end{proof}
This proposition indicates that both the extended Weyl-obstruction tensor $\hat\Omega^{(k-1)}_{ij}$ and $\gamma^{(k)}_{ij}$ are meromorphic functions, which are holomorphic in the whole complex plane except at even integers $d=4,6,\cdots,2k$. We have seen that the pole at $d=2k$ is a simple pole, while the pole at a lower even dimension could be of higher order. These two tensors only differ by terms that are finite at $d=2k$. Therefore,  we can express $\gamma^{(k)}_{ij}$ in terms of $\hat\Omega^{(k-1)}_{ij}$ plus finite terms as we have seen for $k=1,2$ in \eqref{g4} and \eqref{g6}.
\par

Later in the next section, we will introduce the extended Weyl-obstruction tensors in the second order formalism \`a la \cite{Fefferman:2007rka} and show that the two definitions are equivalent.

\section{Weyl-Ambient Space: Bottom-Up Perspective}
\label{sec:bottomup}
\subsection{Formal Theory}
\label{sec:Formal}
In this section we will present a geometric interpretation of the Weyl-ambient metric \eqref{Weyl_ambient} as well as the Weyl connection therein in terms of a bottom-up construction. By ``bottom-up'' we mean to construct a $(d+2)$-dimensional Weyl-ambient space from a $d$-dimensional manifold $M$. The majority of this section is based on Section 2 and Section 3 of \cite{Fefferman:2007rka} where a more detailed exposition of the ambient construction can be found. We will generalize the main definitions and theorems there with the inclusion of a Weyl connection on the principal $\RR_+$-bundle. The resulting Weyl structure together with the metric bundle, viewed as an associated bundle, will be then used to define the Weyl-ambient metric. For this section to be self-contained we repeat some of the definitions and proofs of \cite{Fefferman:2007rka} when necessary while generalizing them appropriately. 
\par
We start with a $d$-dimensional manifold $M$ and introduce a principal $\bb{R}_{+}$-bundle $\mathcal{P}_W$ over $M$ that we call a {\it Weyl structure}.\footnote{We use this name since ${\cal P}_W$ can be regarded as a $G$-structure of the frame bundle, in which the structure group is reduced from $GL(d,\RR)$ to $\RR_+$.} 

\begin{defn}\label{Weylstructure}
Given a $d$-dimensional manifold $M$, a \emph{Weyl structure} is a $(d+1)$-dimensional manifold ${\cal P}_W$ together with the structure group $\bb R_+$, which is equipped with
\par
(2.1) a free right action $\delta:\mathcal{P}_W\times\bb R_+\to\mathcal{P}_W$, such that $\delta_s(p)= p\cdot s$, $\forall p\in\mathcal{P}_W$, $s\in\bb R_+$;
\par
(2.2) a projection map $\pi: \mathcal{P}_W\rightarrow M$, such that $\pi(p)=\pi(p\cdot s)$, $\forall p\in\mathcal{P}_W$, $s\in\bb R_+$;
\par
(2.3) a local trivialization $T_i:\pi^{-1}(U_i)\to U_i\times \bb{R}_{+}$ for each open set $U_i\subset M$ with $T_i(p)=(\pi(p),t_i(p))$, where $t_i:\pi^{-1}(U_i)\to \bb{R}_{+}$ satisfies $t_i(p\cdot s)=t_i(p)\cdot s$ for all $s\in\bb{R}_{+}$.
\par
For brevity, suppose $U_i\subset M$ has local coordinates $\{x^i\}$, we can express a point $p\in{\cal P}_W$ as $(x,t)$ with $t\in \bb{R}_{+}$. 
\end{defn}
A connection on the Weyl structure can be described as follows. First we note that the push forward $\pi_*:T{\cal P}_W\to TM$ defines the vertical sub-bundle $V\subset T{\cal P}_W$ given at any point $p\in{\cal P}_W$ by
\beq
V_p=\ker(\pi_*)\equiv\{\un v\in T_p{\cal P}_W\,|\,\pi_*(\un v)=\un 0\}.
\eeq
In the present case $V_p$ is a one-dimensional vector space spanned by the fundamental vector field which generates the  group action along the fibres; in the local trivialization, it is expressed as $\un T=t\un\pa_t$. From the perspective of ${\cal P}_W$, we can then think of the action of $\bb{R}_+$ as corresponding to a dilatation of the fibres. To assign a connection on ${\cal P}_W$ is to specify a horizontal sub-space $H_p\subset T_p{\cal P}_W$ such that $T_p{\cal P}_W=H_p\oplus V_p$ at any $p$. In the local trivialization given above, the horizontal bundle can be described as the span of vectors of the form $\un D_i=\un\pa_i-a_i(x)t\un\pa_t$.\footnote{Here we have required that $a(x)$ be independent of $t$ in order to make the Weyl-ambient metric homogeneous of degree 2 with respect to $t$.} Equivalently, it can be described as the kernel of a form $n:=t^{-1}\td t+a_i(x)\td x^i\in T^*{\cal P}_W$, i.e.
\beq
\label{eq:HV}
H_p:= \{\un u\in T_p{\cal P}_W\,|\,i_{\un u}n=0\}\qquad \forall p\in{\cal P}_W.
\eeq
We note that under the Abelian group action $(x,t(x))\mapsto (x,t'(x))=(x,t(x)s(x))$, we have
\beq
\label{eq:horizontaln}
n'= n+\big(a_i'(x)-a_i(x)+\pa_i\ln s(x)\big)\td x^i\,,
\eeq
and so we see that the coefficients $a_i(x)$ transform as connection coefficients.
Note also that it is natural to introduce the projector ${\bf a}:T{\cal P}_W\to V$ as
\beq\label{projector_triv_1}
{\bf a}=t\un\pa_t\otimes \big(t^{-1}\td t+a_i(x)\td x^i\big)\,,
\eeq
which is an alternative way to express the connection on ${\cal P}_W$. We will refer to both $\mathbf a$ and $a_i(x)$ as the \emph{Weyl connection}.

This line bundle has an important representation given by a conformal class of metrics. Indeed, all the non-trivial representations are one-dimensional, and thus a representation of $\bb{R}_+$ is given by specifying a Weyl weight $w$. We call the corresponding associated bundle ${\cal E}_{w}$ and its sections respond to the group action as
\beq
T_x\mapsto s(x)^w T_x\,.
\eeq
Equivalently, this determines the transition functions on the associated bundle.

Suppose a conformal class $[g]$ of smooth metrics of signature $(p,q)$ is given on $M$, in which any two representatives $g$ and $g'$ are related by a smooth function ${\cal B}(x)$ as $g'_x={\cal B}(x)^{-2}g_x$, where $g_x$ is the value of $g$ at a point $x\in M$. Then, $(M,[g])$ is a conformal manifold. One can define a \emph{metric bundle} ${\cal G}$ as follows \cite{Fefferman:2007rka}:
\begin{defn}\label{Metric Bundle}
A \emph{metric bundle} $\mathcal{G}$ is the collection of pairs $(x,h)$ where $h= s^2 g_{x}$, $\forall s\in \mathbb{R}_{+}$ and $\forall x\in M$, which is equipped with 
\par
(3.1) a dilatation map $\tilde\delta_{s}: \mathcal{G}\to \mathcal{G}$ such that $\tilde\delta_{s}(x,h)= (x,s^2 h)$, $\forall s\in \mathbb R_{+}$. 
\par
(3.2) a projection map $\tilde\pi: \mathcal{G}\rightarrow M$ such that $(x,h)\mapsto x$;
\end{defn}
This definition simply identifies a conformal class of metrics with a bundle associated to the Weyl structure given by the weight $w=-2$ representation of $\bb{R}_+$. We note that it is isomorphic to the Weyl structure ${\cal P}_W$, as is any non-trivial associated bundle of ${\cal P}_W$.\footnote{Note that in \cite{Fefferman:2007rka}, the metric bundle $\cal G$ itself is treated as the principal $\bb R_+$-bundle through an isomorphism. Here we introduced the Weyl structure ${\cal P}_W$ and distinguish it from $\cal G$ in order to emphasize that a conformal class of metrics furnishes a representation of the group $\bb{R}_+$ with $w=-2$.} Under a trivialization, assigning an isomorphism between ${\cal P}_W$ and the metric bundle $\cal G$ can be thought of as a choice of representative $g$ of the conformal class $[g]$ if we identify
\begin{equation}
\label{eq:trivial}
(x,t)\in U_i\times\mathbb{R}_{+} \qquad \text{with }\qquad (x,t^2 g_{x})\in {\cal G}\,.
\end{equation}
Given $g\in[g]$, for any $p\in{\cal P}_W$, by means of the corresponding $(x,h)\in{\cal G}$ one can define a symmetric tensor $\mathbf{g}_{0}$ of type $(0,2)$ called the \emph{tautological tensor} that acts on vector fields $\un w_1,\un w_2 \in T_{p}{\cal P}_W$ as follows:
\begin{equation}\label{tautological_tensor}
\mathbf{g}_{0} (\un w_1,\un w_2)\equiv h(\pi_{*}\un w_1,\pi_{*}\un w_2)\,,
\end{equation}
which can be expressed as $\mathbf{g}_0=t^2\pi^*g$ under the identification in \eqref{eq:trivial}. 

If we pick another representative $g'_{x}={\cal B}(x)^{-2}g_{x}$ of the conformal class $[g]$, following the identification in \eqref{eq:trivial}, we obtain another isomorphism between ${\cal P}_W$ and ${\cal G}$ by identifying
\begin{equation}
(x,t')\in U_i\times\mathbb{R}_{+} \qquad \text{with }\qquad (x,t'^2 g'_{x})\in {\cal G}.
\end{equation}
It is easy to see that the two isomorphisms are related by setting $t'= {\cal B}(x)t$. To preserve the horizontal subspace on ${\cal P}_W$, from \eqref{eq:horizontaln} we can see that $a'_{i}(x)$ satisfies
\begin{equation}
a_{i}'(x)= a_{i}(x)- \pa_{i}\text{ln}{\cal B}(x)\,.
\end{equation}
In the present circumstances, it is natural to replace the notion of conformal class $[g]$ by the {\it Weyl class} $[g,a]$, with the property
\beq
\forall (g,a), (g',a')\in [g,a],\quad \exists\; {\cal B}(x)\; {\rm such \; that} \; (g'_x,a'_x)=({\cal B}(x)^{-2}g_x,a_x-\td\ln{\cal B}(x))\,,
\eeq
where $\td$ is the exterior derivative on $M$.
\par
Before we proceed to define the Weyl-ambient space based on the Weyl structure ${\cal P}_W$, we would like to make a few remarks. Recall that for the Weyl-ambient metric \eqref{Weyl_ambient}, the coordinates $t$ and $x^{i}$ parametrize a codimension-1 null hypersurface ${\cal N}$ located at $\rho=0$. One can see that this surface is exactly a Weyl structure. In Section \ref{sec:topdown}, the degenerate ``induced metric" of $\tilde g$ on ${\cal N}$ is the tautological tensor, the induced metric $\gamma^{(0)}$ on $M$ is a representative $g$ in the conformal class, and the Weyl connection $a^{(0)}_i(x)$ on $M$ is the $a_i(x)$ in \eqref{projector_triv_1}. Thus, the Weyl class $[\gamma^{(0)},a^{(0)}]$ corresponds to $[g,a]$ in this section, and $(M,[g,a])$ defines a Weyl manifold. We will discuss more details of the role of the Weyl connection and the horizontal subspace it defines in Theorem \ref{diff_Weyl_normal_form} below. It is noteworthy that the projector in \eqref{projector_triv_1}, which defines the Weyl connection on ${\cal P}_W$, is a special case of the construction presented in \cite{Ciambelli:2019lap} with restricted diffeomorphisms.

\par
Now we will define a Weyl-ambient space for a Weyl manifold generalizing the definition of a Fefferman-Graham ambient space for a conformal manifold introduced in \cite{Fefferman:2007rka}. Consider a $(d+2)$-dimensional space $\tilde M$ which looks at least locally like ${\cal P}_W \times\bb{R}$ where each point can be labeled by $(p,\rho)$ with $\rho\in\bb R$. The inclusion map $\iota:{\cal P}_W\to\tilde M$ is defined such that $p\mapsto (p,0)$. By letting the  map $\delta_s$  act only on $p\in{\cal P}_W$,  we can extend $\delta_s$ to a map on $\tilde M$, which commutes with $\iota$. The vector field $\un T$ which generates the Weyl group action is extended to a vector field $\un{\cal T}=\iota_*\un T=t\un\p_t$ on $\tilde M$.
\begin{defn}\label{Ambient}
Suppose $M$ is a $d$-dimensional manifold equipped with a Weyl class $[g,a]$, and ${\cal P}_W$ is a Weyl structure over $M$. A pseudo-Riemannian space $(\tilde M,\tilde g)$ is called the \emph{Weyl-ambient space} for $(M,[g,a])$ if
\par
(4.1) $\tilde M$ is a dilatation-invariant open neighborhood of ${\cal P}_W\times\{0\}$ in ${\cal P}_W\times\bb R$, and the pullback $\iota^*\tilde g$ is the tautological tensor $\mathbf{g}_0$ defined above;
\par
(4.2) $\tilde g$ is a smooth metric on $\tilde M$ of signature $(p+1,q+1)$, which is homogeneous of degree 2 on $\tilde M$, i.e.,\ $\delta^*_s\tilde g=s^2\tilde g$, $\forall s\in\bb R_+$;
\par
(4.3) $Ric(\tilde g)$ vanishes to infinite order at every point of ${\cal P}_W\times\{0\}$.
\end{defn}
\noindent Without condition (4.3), $(\tilde M,\tilde g)$ is called a \emph{Weyl pre-ambient space} for $(M,[g,a])$. Note that the condition (3) in \cite{Fefferman:2007rka} is presented differently when $d$ is even and odd, and $Ric(\tilde g)$ has an obstruction in the order $O(\rho^{d/2-1})$ for even $d$. 
Here we take the dimension to be a continuous complex variable, and so the Ricci-flatness condition always holds to infinite order. The obstruction at even dimension will be manifested by the pole of the expansion of $\tilde g$ at even $d$, which is identified as the extended Weyl-obstruction tensor. This formulation was described in \cite{Ciambelli:2019bzz} and the equivalence of the two formulations in the level of AlAdS bulk was addressed in \cite{Jia:2021hgy}.
\par
Now we introduce the final ingredient in our Weyl-ambient construction---the \emph{Weyl-normal form}, which is a generalization of the \emph{normal form} defined in \cite{Fefferman:2007rka}.
 \begin{defn}\label{Weyl_normal_form}
 A Weyl pre-ambient space $(\ti M,\ti g)$ for $(M,[g,a])$ is said to be in \emph{Weyl-normal form} with acceleration $\un{\cal A}$ if
 \par
(5.1) For each fixed $p\in \mathcal{P}_W$, the set of $\rho \in \mathbb{R}$ such that $(p,\rho)\in \ti M$ is an open interval $I_{p}\in\bb R$ containing $0$.
\par
(5.2) For each $p\in {\cal P}_W$, the parametrized curve $C_p:I_{p} \to \tilde M$, $\rho\mapsto(p,\rho)$ has a tangent vector $\un{\cal U}$, whose acceleration $\un {\cal A}\equiv \tilde\nabla_{\un {\cal U}}\un {\cal U}$ satisfies $\tilde g(\un {\cal T}, \un {\cal A})=0$, where $\tilde\nabla$ is the Levi-Civita connection of $(\tilde M,\tilde g)$.
\par
(5.3) Let $(t,x,\rho)$ represent a point in $\mathbb{R}_{+}\times M\times \mathbb{R}\simeq {\cal P}_W\times \mathbb{R}$ under the local trivialization induced by $g$. Then, at each point $(t,x,0)\in {\cal P}_W\times \{0\}$, the metric $\ti g$ takes the form
 \begin{equation}\label{form_1}
 \ti g|_{\rho=0}= \mathbf{g}_{0}  + 2t^2(t^{-1}\td t + a_{i}(x)\td x^{i})\td \rho\,,
 \end{equation}
 where $\mathbf{g}_{0}$ is the  tautological symmetric tensor defined in \eqref{tautological_tensor}.
 \end{defn} 
Definition \ref{Weyl_normal_form} is engineered for the purpose of generating the Weyl-ambient metric from the ``initial surface'' at $\rho=0$. At $\rho=0$, the Weyl-ambient metric we have seen in \eqref{Weyl_ambient} has the form \eqref{form_1}, which motivates condition (5.3). Since $\un{\cal T}=t\un\p_t$ everywhere in $\tilde M$, condition (5.2) implies that the covector ${\cal A}$ of the acceleration does not have a $t$-component. Furthermore, one can also parametrize the accelerated curve $C_p$ such that $\tilde g(\un{\cal A},\un{\cal U})=0$, and let $\cal A$ have no $\rho$-component either.\footnote{Suppose $C_p$ has a parameter $\lambda$, then under a reparametrization $\lambda\to f(\lambda)$ we have $\un{\cal U}\to f'\un{\cal U}$, and the acceleration vector transforms  $\un{\cal A}\to f'^2 \un{\cal A} + f'\un{\cal U}(f)\un{\cal U}$, and thus $\tilde g(\un{\cal A}, \un{\cal U})$ can always be set to zero for non-null $\un{\cal U}$ by choosing an appropriate function $f$. For null $\un{\cal U}$ the condition holds automatically.} We will assume that $\rho$ is such a parametrization. Note that in the special case where $\un{\cal A}=0$, the $\rho$-coordinate lines are geodesics, and condition (5.2) goes back to that of normal form in \cite{Fefferman:2007rka}, while condition (5.3) will still be different as long as $a_{i}(x)$ are nonvanishing. The acceleration $\un{\cal A}$ encodes all the higher modes $a^{(k\geqslant 1)}_{i}(x)$ in the expansion \eqref{eq:aexpan} of $a_i(x,\rho)$, as we will see in Lemma \ref{Lemma}. In fact, if both $a_{i}(x)$ and $\un{\cal A}$ are zero, the mode $a_{i}(x,\rho)$ in \eqref{Weyl_ambient} vanishes. 
\par
The following Theorem is a generalization of Proposition 2.8 in \cite{Fefferman:2007rka}.
 \begin{theorem}\label{diff_Weyl_normal_form}
 Let $(M,[g,a])$ be a Weyl manifold, with $(g,a)$ a representative of the Weyl class. Let ${\cal P}_W$ be the Weyl structure over $M$, and $(\ti M,\ti g)$ be a Weyl pre-ambient space for $(M,[g,a])$. Then, there exists a dilatation-invariant open set $\tilde M'\subset {\cal P}_W\times \mathbb{R}$ containing ${\cal P}_W\times \{0\}$ on which there is a unique diffeomorphism $\phi:\tilde M'\to\ti M$ commuting with dilatations with $\phi |_{{\cal P}_W\times \{0\}}$ being the identity map, such that the Weyl pre-ambient space $(\tilde M', \phi^{*}\ti g)$ is in Weyl-normal form with acceleration $\un{\cal A}'$.
\end{theorem}
This theorem indicates that given a representative pair $(g,a)$, any Weyl pre-ambient space can be put into Weyl-normal form by a diffeomorphism $\phi$. $(\tilde M, \ti g)$ and $(\tilde M', \phi^{*}\ti g)$ are also said to be ambient-equivalent (see Definition 2.2 in \cite{Fefferman:2007rka} for the precise definition of ambient equivalence).
For a proof of the theorem, see Appendix \ref{Appthm5pt1}.

\par
Before we move on to the main result of this section, namely Theorem \ref{Weyl_ambient_existence}, let us introduce some useful notation. Given a local coordinate system $\{x^i\}$ ($i=1,\cdots,d$) on $M$, the fibre coordinate $t$ of ${\cal P}_W$ and the parameter $\rho$ naturally defines an \emph{ambient coordinate system} $\{t,x^i,\rho\}$ on $\tilde M$. Later on, we will follow \cite{Fefferman:2007rka} and use $I,J,\cdots=(0,i,\infty)$ to label the ambient coordinate indices, where $0$ labels the $t$-component and $\infty$ labels the $\rho$-component. It is also convenient to interpret the notations $(0,i,\infty)$ as representing the components in a trivialization ${\cal P}_W\times \bb R\simeq \bb R_+\times M\times\bb R$, even without specifying a choice of coordinates on $M$.

We will now present Theorem \ref{Weyl_ambient_existence}, which is a natural generalization of Theorem 2.9 of \cite{Fefferman:2007rka}, based on our definition of Weyl-normal form. As a corollary of this theorem, we will show that for a Weyl-ambient space in Weyl-normal form, the Weyl-ambient metric \eqref{Weyl_ambient} emerges from the initial surface uniquely under the Ricci-flatness condition. We emphasize again that we consider the dimension $d$ of the manifold $M$ formally as a complex parameter, and do not need to distinguish between even and odd dimensions. This is in alignment with previous considerations in \cite{Ciambelli:2019bzz,Jia:2021hgy}. 

\begin{theorem}\label{Weyl_ambient_existence}
Let $(M,[g,a])$ be a Weyl manifold, and let $(g,a)$ be a representative in the Weyl class.
 \begin{enumerate}[label=$\rm (\Alph*)$]
 \item There exists a Weyl-ambient space $(\ti M,\ti g)$ for $(M,[g,a])$ which is in Weyl-normal form with acceleration $\un{\cal A}$.
 \item Suppose that $(\ti M_{1},\ti g_{1})$ and $(\ti M_{2},\ti g_{2})$ are two Weyl-ambient spaces for $(M,[g,a])$, both of which are in Weyl-normal form with acceleration $\un{\cal A}$. Then $\ti g_{1}- \ti g_{2}$ vanishes to infinite order at every point of ${\cal P}_W\times \{0\}$.
 \end{enumerate}
 \end{theorem}
\par
The proof of Theorem \ref{Weyl_ambient_existence} employs the following lemma. 
\begin{lemma}\label{Lemma}
Let $({\ti M},\ti g)$ be a Weyl pre-ambient space for $(M,[g,a])$. Suppose for each $p\in {\cal P}_W$, the set of all $\rho\in \mathbb{R}$ such that $(p,\rho)\in \ti  M$ is an open interval $I_{p}$ containing $0$. Let $g$ be a metric in the representative $(g,a)$ of the Weyl class, which provides a local trivialization ${\cal P}_W\times \mathbb{R}\simeq \mathbb{R}_{+}\times M\times \mathbb{R}$. Then $(\ti M,\ti g)$ is in Weyl-normal form with acceleration $\un{\cal A}$ if and only if one has on $\ti M$:
\begin{equation}\label{lemma_conditions}
\ti g_{0\infty}=t\,,\qquad \ti g_{i\infty}=t^{2} a_{i}(x,\rho)\,,\qquad \ti g_{\infty\infty}=0\,,
\end{equation}
where $a_{i}(x,\rho)\equiv a_{i}(x)+t^{-2} \int_{0}^{\rho}{\cal A}_{i}(t,x,\rho')\td\rho' $.
\end{lemma}

\begin{proof}
Suppose $\tilde g$ satisfies \eqref{lemma_conditions}, then it follows from the condition $\iota^{*}\ti g= \mathbf{g}_{0}$ for the pre-ambient space that $\ti g|_{\rho=0}$ must have the form \eqref{form_1}. Thus, all we have to prove is that for $\ti g$ satisfying \eqref{form_1} at $\rho=0$, the condition that the $\rho$-coordinate lines have acceleration $\un{\cal  A}$ with $\tilde g(\un{\cal T}, \un{\cal A})=0$ is equivalent to \eqref{lemma_conditions}. The fact that the $\rho $-coordinate lines have an acceleration $\un{\cal A}$ implies
\begin{equation}\label{Gamma_condition_rho_lines}
\ti\Gamma\indices{_{\infty \infty I}}= \mathcal{A}_{I}\,,
\end{equation}
where $\ti\Gamma_{IJK}\equiv\ti g_{KL}\ti\Gamma\indices{^{L}_{IJ}}$.
The condition $\tilde g(\un{\cal T},\un{\cal A})=0$ leads to ${\cal A}_0=0$. As we have mentioned, one can also parametrize the curve $C_p:I_p\to\tilde M$ such that $\tilde g(\un{\cal U}, \un{\cal A})=0$, then we also have ${\cal A}_{\infty}=0$, and thus $\un{\cal A}_{I}=\left({\cal A}_{0}, {\cal A}_{i},{\cal A}_{\infty}\right)= \left(0,t^{2}\varphi_{i}(x,\rho),0\right)$. The functions $\varphi_{i}(x,\rho)$ are considered as external input and cannot be determined from the initial conditions. The factor $t^2$ is derived from the homogeneity property of $\ti g$ and \eqref{Gamma_condition_rho_lines}. If we set $I=\infty$ in \eqref{Gamma_condition_rho_lines} we get 
\begin{equation}
\ti\Gamma\indices{_{\infty\infty \infty}}= {\cal A}_{\infty}= 0\,\implies\,\pa_{\rho}g_{\infty\infty}=0\,\implies\, g_{\infty\infty}=0\,,
\end{equation}
where in the last step we used the initial condition $g_{\infty\infty}|_{\rho=0}=0$. Similarly, setting $I=0$ in \eqref{Gamma_condition_rho_lines} we find
\begin{equation}
\pa_{\infty }g_{\infty 0}=0\, \implies\, g_{\infty 0}= t\,,
\end{equation}
where we used the initial condition $g_{0\infty}|_{\rho=0}=t$. Finally, setting $I=i$ yields
\begin{equation}
\pa_{\rho}g_{\infty i}={\cal A}_{i}(t,\rho;x)\implies g_{\infty i}= t^{2}a_{i}(x)+ t^{2}\int_{0}^{\rho}\varphi_{i}(\rho;x)\td\rho \equiv  t^2 a_{i}(\rho;x)\,,
\end{equation}
where we used the initial condition $\ti g_{\infty i}|_{\rho=0}= t^2 a_{i}(x)$.
 \end{proof}
The main logic of the proof of Theorem \ref{Weyl_ambient_existence} will follow part of Section 3 in \cite{Fefferman:2007rka}. To show part (A) of Theorem \ref{Weyl_ambient_existence}, namely the existence of the Weyl-ambient space ${\tilde M}$ in Weyl-normal form, we need to show the following: for a Weyl manifold $(M,[g,a])$, given a representative $(g,a)$ of the Weyl class and $a_i(x,\rho)$ determined by $\un{\cal A}$, there exists a metric $\ti g$ on an open neighborhood $\ti M$ of ${\cal P}_W\times\{0\}$ with the following properties:
 \begin{enumerate}[label={(\arabic*)}]
 \item $\delta^{*}_{s}\ti g= s^2 \ti g$, $\forall s>0$ (homogeneity property);
 \item $\ti g = t^2 g(x)+ 2t^2(t^{-1}\td t+ a_{i}(x)\td x^{i})\td\rho$ when $\rho=0$;
 \item $\ti g_{0\infty}=t,\quad \ti g_{i\infty}=t^{2} a_{i}(x,\rho),\quad \ti g_{\infty\infty}=0$;
 \item $\ti Ric(\ti g)=0$ to infinite order at $\rho=0$.
 \end{enumerate}
 The first property above is the homogeneity property which is still taken to be true for the Weyl-ambient metric. Property (3) is equivalent to condition (5.2) of  Definition \ref{Weyl_normal_form} due to Lemma \ref{Lemma}, which indicates that $\ti g_{I\infty}$ components are known, while the rest are now regarded as unknown functions. Property (2) can be considered as the initial data of these components at the initial surface at $\rho=0$, while the Ricci-flatness property (4) is a system of partial differential equations that one can solve to find the metric components beyond the initial surface. We will show that this is a well defined initial value problem so that the unknown components of the Weyl-ambient metric can be uniquely determined in a series expansion in $\rho$, which will prove part (B) of Theorem \ref{Weyl_ambient_existence}. The complete proof can be found in Appendix \ref{Appthm5pt1}.

As an important corollary, we now show in Theorem \ref{thm:Weyl_ambient} that the metric $\tilde g$ determined from Theorem \ref{Weyl_ambient_existence} has exactly the form of the Weyl-ambient metric \eqref{Weyl_ambient}. First we need the following lemma.
\begin{lemma}\label{Lemma_Ricci_Weyl_ambient}
Suppose a metric $\ti g$ has the following form:
\begin{equation}\label{WFG_ambient}
 \ti g_{IJ}=\left(\begin{array}{ccc}
2\rho & 0 &  t\\
0 & t^2 g_{ij}(x,\rho)& t^2 a_{j}(x,\rho)\\
t&t^2 a_{i}(x,\rho)&0
\end{array}\right)\,.
\end{equation}
Then the Ricci curvature of $\ti g$ satisfies $\ti R_{0I}=0$.
\end{lemma}
\begin{proof}
For $\ti g$ of the form \eqref{WFG_ambient}, we can write the inverse metric as
\begin{equation}\label{inverse_Weyl_ambient}
\begin{split}
\ti g^{IJ}&= \frac{1}{1+ 2\rho a^2}
\left(\begin{array}{ccc}
a^2 & -t^{-1}a^{j} &t^{-1}\\
-t^{-1}a^{i} &t^{-2}(1+ 2\rho a^2 )g^{ij}- 2t^{-2}\rho a^{i}a^{j}&2t^{-2}\rho a^{i}\\
t^{-1}&2t^{-2}\rho a^{j}& -t^{-2}2\rho
\end{array}\right)\,,
\end{split}
\end{equation}
and the Christoffel symbols $\ti\Gamma_{IJK}=\ti g_{KL}\ti\Gamma\indices{^{L}_{IJ}}$ are given by
\begin{equation}\label{Christoffel_Weyl_ambient}
\begin{split}
\ti \Gamma_{IJ0}&=
\left(\begin{array}{ccc}
0& 0 &1\\
0 &-tg_{ij}&-ta_{i}\\
1&-ta_{j}& 0
\end{array}\right)\,,\quad
\ti \Gamma_{IJ\infty}=
\left(\begin{array}{ccc}
0& t a_{j} &0\\
 t a_{i} &-t^{2}\left(\frac{1}{2}\p_\rho g_{ij} - \pa_{(i}a_{j)}\right)&0\\
0&0& 0
\end{array}\right)\,,\\
\ti \Gamma_{IJk}&=
\left(\begin{array}{ccc}
0& t g_{jk} &t a_{k}\\
t g_{ik} &t^{2}\Gamma_{ijk}&\frac{t^2}{2}\left(\p_\rho g_{ik} + F_{ik}\right)\\
 t a_{k}&\frac{t^2}{2}\left(\p_\rho g_{jk} + F_{jk}\right)& t^{2}\pa_{\rho}a_{k}
\end{array}\right)\,,
\end{split}
\end{equation}
where $\Gamma_{ijk}=g_{kl}\Gamma\indices{^{l}_{ij}}$ are the Christoffel symbols of $g_{ij}(x,\rho)$, and $F_{jk}= \pa_{j}a_{k}- \pa_{k}a_{j}$. Substituting \eqref{inverse_Weyl_ambient} and  \eqref{Christoffel_Weyl_ambient} into \eqref{Ricci} we can compute $\ti R_{0I}$ explicitly and find that $\ti R_{0I}=0$.
\end{proof}

\begin{theorem}
\label{thm:Weyl_ambient}
Suppose $(M,[g,a])$ is a Weyl manifold. Let $(\ti M,\ti g)$ be the unique ambient space for $(M,[g,a])$ which is in Weyl-normal form with acceleration $\un{\cal A}$. Then, for any representative $(g,a)$, the uniquely determined metric $\tilde g$ has the following form 
\begin{equation}
\ti g = 2\rho \td t^2 + 2td\rho\left(\frac{\td t}{t} +  a_{i}(x,\rho)\td x^{i}\right) + t^2 g_{ij}(x,\rho)\td x^{i}\td x^{j}\,,
\end{equation}
where $a_{i}(x,\rho)\equiv a_{i}(x)+t^{-2} \int_{0}^{\rho}{\cal A}_{i}(t,x,\rho')$. This metric is exactly the Weyl-ambient metric introduced in \eqref{Weyl_ambient}. 
\end{theorem}
\begin{proof}
Based on Theorem \ref{Weyl_ambient_existence}, all we have to prove is that $\tilde g_{00}=2\rho$ and $\tilde g_{0i}=0$ to all orders. Let $\ti g^{(m)}$ be the $m^{th}$ order of $\ti g$, and let  $\ti g^{[k]}$ represent $\tilde g$ with all the orders higher than $O(\rho^{k})$ in the $\rho$-expansion excluded, i.e. $\tilde g=\tilde g^{[k]}+O(\rho^{k+1})$. From \eqref{metric_components_order_1} we find to the first order that $\tilde g^{[1]}_{00}=2\rho$ and $\tilde g^{[1]}_{0i}=0$. Assuming that $\tilde g^{[m-1]}_{00}=2\rho$ and $\tilde g^{[m-1]}_{0i}=0$, it follows from Lemma \ref{Lemma_Ricci_Weyl_ambient} that $\tilde R^{[m-1]}_{00}=\tilde R^{[m-1]}_{0i}=0$. Then, from \eqref{Ricci_inductive} we obtain that $\phi_{00}=\phi_{0i}=0$, and hence $\tilde g^{(m)}_{00}=\tilde g^{(m)}_{0i}=0$ [see \eqref{eq:g+phi}], $\forall m>1$. Therefore, by induction we can deduce to infinite order that $\tilde g_{00}=2\rho$ and $\tilde g_{0i}=0$, which completes the proof.
\end{proof}

\subsection{Weyl-Obstruction Tensors: Second Order Formalism}
\label{sec:WeylObs}
In Section \ref{sec:topdown} we have seen that Weyl-obstruction tensors can be defined as the derivatives of the ambient Riemann tensor in the first order formalism. In this section we will follow the setup of the present section in the second order formalism and show that appropriate ambient tensors constructed from the Weyl-ambient Riemann tensor on $\tilde M$ behave as Weyl-covariant tensors on $M$, through which Weyl-obstruction tensors can again be defined as a special case.   Then we will show that the Weyl-obstruction tensors defined in this way agree with the Weyl-obstruction tensors we defined previously in Definition \ref{def1}. 
\par

We have proven in the last subsection that for any pair of $(g,a)$ on $M$, there exists a unique Weyl-ambient space $(\tilde M,\tilde g)$ for the Weyl manifold $(M,[g,a])$ where $\tilde g$ has the form of \eqref{Weyl_ambient}. In Section \ref{sec:topdown} we saw that the ambient Weyl diffeomorphism
\be
\label{eq:ambientdiff}
(t',x'^i,\rho')=({\cal B}(x)t,x^i,{\cal B}^{-2}(x)\rho)
\ee
induces a Weyl transformation on $M$. Therefore, to find a Weyl-covariant tensor on $(M,[g,a])$, we can find an ambient tensor which is covariant under an ambient Weyl diffeomorphism, and its pullback on $M$ will be Weyl covariant.
\par
The first main result of this section is the following proposition. This  provides  the Weyl transformations of tensors constructed from covariant derivatives of the Riemann tensor of a Weyl-ambient metric, from which we can see which tensors are Weyl covariant when pulled back to $M$.

\begin{prop}\label{Proposition_6.5_Graham}
Suppose $(\tilde M,\tilde g)$ is the Weyl-ambient space for $(M,[g,a])$, and let $(g,a)$ and $(g',a')$ be two representatives of $[g,a]$, with $g'_{ij}={\cal B}^{-2}g_{ij}$ and $a'_i=a_i-\p_i\ln{\cal B}$.
Let $IJKLM_{1}\dots M_{r}$ be a list of indices, $s_{0}$ of which are $0$, $s_{M}$ of which are $x^i$ on $M$, and $s_{\infty}$ of which are $\infty$. Then, the following components of the covariant derivatives of the Riemann tensor $\tilde R_{ABCD}$ of $\ti g$ in the trivialization defined by $g$ satisfy the transformation law
\begin{equation}\label{Proposition_6.5_Graham_main}
\ti R'_{IJKL;M_{1}\cdots M_{r}}|_{\rho'=0,t'=1}={\cal B}(x)^{2(s_{\infty }-1)} \ti R_{ABCD;F_{1}\cdots F_{r}}|_{\rho=0,t=1}p\indices{^{A}_{I}}\cdots p\indices{^{F_{r}}_{M_{r}}}
\end{equation}
under an ambient Weyl diffeomorphism \eqref{eq:ambientdiff}, where $p\indices{^{A}_{I}}$ is the matrix 
\begin{equation}\label{p_matrix}
p\indices{^{I}_{J}}= \begin{blockarray}{cccc}
 & 0 & j & \infty  \\
\begin{block}{c(ccc)}
  0 & 1& \Upsilon_{j} &0\\
   i & 0  & \delta\indices{^{i}_{j}}& 0 \\
  \infty & 0&0& 1\\
\end{block}
\end{blockarray}\,\,,
\end{equation}
and $\Upsilon(x)\equiv - \ln {\cal B}(x)$, $\Upsilon_{j}\equiv \pa_{j}\Upsilon(x)$.
$\ti R'_{IJKL;M_{1}\dots M_{r}}$ denotes covariant derivatives of the Riemann tensor of $\tilde g$ in the coordinates $X'^{I}=(t', x'^{i},\rho')$ given by the trivialization provided by $g'$.
\end{prop}
\begin{proof}
The logic for the proof of this Proposition follows the proof of Proposition 6.5 in \cite{Fefferman:2007rka} closely. We start by observing that the ambient Weyl diffeomorphism $\psi:(t',x'^{i},\rho')\mapsto(t,x^{i},\rho)$ has the following properties:
\begin{equation}\label{psi_requirements}
\psi(t',x'^{i},0)= \big(t' \E^{\Upsilon (x)},x'^{i},0 \big)\,,\qquad\psi^{*}\ti g|_{\rho'=0}= 2t'\td\rho'\td t'+ t'^2 g'_{ij}\td x'^{i}\td x'^{j}+ 2t'^2 a'_{i}\td x'^{i}\td\rho'\,,
\end{equation}
where the Weyl-ambient metric $\ti g$ has the form of \eqref{Weyl_ambient}, and $g'_{ij}= {\cal B}(x)^{-2}g_{ij}$, $a'_{i}= a_{i}+\Upsilon_i$. The Jacobian $(\psi)\indices{^{A}_{I}}= \left(\frac{\pa X^{I}}{\pa X'^{J}}\right)$ of this diffeomorphism is
\begin{equation}\label{Jacobian_Weyl_diff}
(\psi)\indices{^{I}_{J}}\equiv 
\left(\begin{array}{ccc}
\psi\indices{^{t}_{t'}} & \psi\indices{^{t}_{j'}} & \psi\indices{^{t}_{\rho'}}\\
\psi\indices{^{i}_{t'}} & \psi\indices{^{i}_{j'}} & \psi\indices{^{i}_{\rho'}}\\
\psi\indices{^{\rho}_{t'}}& \psi\indices{^{\rho}_{j'}}& \psi\indices{^{\rho}_{\rho'}}
\end{array}\right)
=
\left(\begin{array}{ccc}
\E^{\Upsilon(x) } &t' \E^{\Upsilon (x)}  \Upsilon_{j} & 0\\
0 & \delta\indices{^{i}_{j}}& 0\\
0&-2\rho' \E^{-2\Upsilon(x)}\Upsilon_{j}& \E^{-2\Upsilon(x)}
\end{array}\right)\,,
\end{equation}
where $\Upsilon(x)\equiv -\text{ln}{\cal B}(x)$ and $\Upsilon_{i}\equiv \pa_{i}\Upsilon(x)$.
At $\rho'=0$ the Jacobian matrix \eqref{Jacobian_Weyl_diff} reads
\begin{equation}\label{Jacobian_Weyl_diff_rho=0}
(\psi)\indices{^{A}_{I}}|_{\rho' =0}= 
\left(\begin{array}{ccc}
\E^{\Upsilon(x) } & t' \E^{\Upsilon (x)}  \Upsilon_{j} & 0\\
0 & \delta\indices{^{i}_{j}}& 0\\
0&0& \E^{-2 \Upsilon(x)}
\end{array}\right)\,.
\end{equation}
The above matrix can be written as the following matrix product:
\begin{equation}
(\psi)\indices{^{A}_{I}}|_{\rho' =0} =d_{1}p d_{2}\,,
\end{equation}
with 
\begin{equation}\label{pd_matrix}
p\indices{^{I}_{J}}=
\left(\begin{array}{ccc}
1& \Upsilon_{j} &  0\\
0 & \delta\indices{^{i}_{j}}& 0\\
0&0& 1
\end{array}\right),\quad
d_{1}=
\left(\begin{array}{ccc}
 t' \E^{\Upsilon(x)}& 0 &  0\\
0 & \delta\indices{^{i}_{j}}& 0\\
0&0& 1
\end{array}\right),\quad
d_{2}=
\left(\begin{array}{ccc}
 t'^{-1} & 0 &  0\\
0 & \delta\indices{^{i}_{j}}& 0\\
0&0& \E^{-2\Upsilon(x)}
\end{array}\right).
\end{equation}

Since the Weyl-ambient metric is homogeneous of degree 2 under dilatations $\delta_{s}^{*}\ti g= s^{2}\ti g$, it follows that the left-hand side of \eqref{Proposition_6.5_Graham_main} satisfies
\begin{equation}\label{homog_Riemann}
\ti R'_{IJKL;M_{1}\cdots M_{r}}|_{\rho'=0,t'=1}= {\cal B}(x)^{s_{0}-2}\ti R'_{IJKL;M_{1}\cdots M_{r}}|_{\rho'=0,t'={\cal B}(x)}\,.
\end{equation}
Under the ambient Weyl diffeomorphism \eqref{eq:ambientdiff} the covariant derivatives of the ambient Riemann curvature components transform tensorially as
\begin{equation}\label{WFG_tensorially}
\ti R'_{IJKL;M_{1}\cdots M_{r}}|_{\rho',t'}= \ti R_{ABCD;F_{1}\cdots F_{r}}|_{\rho,t} (\psi)\indices{^{A}_{I}}\cdots (\psi)\indices{^{F_{r}}_{M_{r}}}\,.
\end{equation}
Evaluating both sides of \eqref{WFG_tensorially} at $\rho'=0$, $t'=e^{-\Upsilon(x)}$ and using \eqref{pd_matrix} we have
\begin{equation}\label{Riemann_tensor_tensorially} 
\ti R'_{IJKL;M_{1}\cdots M_{r}}|_{\rho'=0,t'=e^{-\Upsilon(x)}}={\cal B}(x)^{2s_{\infty}- s_{0}} \ti R_{ABCD;F_{1}\cdots F_{r}}|_{\rho=0,t=1}p\indices{^{A}_{I}}\cdots p\indices{^{F_{r}}_{M_{r}}}\,.
\end{equation}
Plugging \eqref{homog_Riemann} into \eqref{Riemann_tensor_tensorially}, we  obtain \eqref{Proposition_6.5_Graham_main}.
\end{proof}

Theorem \ref{Proposition_6.5_Graham} helps us to find  Weyl-covariant tensors on $(M,[g,a])$. First let us look at the case without derivatives. In the coordinate basis, the nonvanishing components of the Weyl-ambient Riemann tensor $\tilde R_{IJKL}$ are $\ti R_{\infty jk\infty}$, $\ti R_{\infty jkl}$ and $\ti R_{ijkl}$.  Evaluating at $\rho=0$ and $t=1$, they are

\begin{align}
\label{eq:Rijkl}
\ti R_{\infty jk\infty}|_{\rho=0,t=1}= \hat\Omega^{(1)}_{jk}\,,\qquad\ti R_{\infty jkl}|_{\rho=0,t=1}= \hat C_{jkl}\,,\qquad\ti R_{ijkl}|_{\rho=0,t=1}= \hat W_{ijkl}\,.
\end{align}
Here $\hat C_{jkl}$ and $\hat W_{ijkl}$ are the Weyl-Cotton tensor and the Weyl curvature tensor on $M$, respectively, and $\hat\Omega^{(1)}_{jk}$ for now simply denotes the tensor defined in \eqref{eq:Omega1}. Then, applying \eqref{Proposition_6.5_Graham_main} we get $\hat C'_{jkl}= \hat C_{jkl}$ and $\hat W'_{ijkl}= {\cal B}^{2}(x) \hat W_{ijkl}$  under Weyl transformation as expected, we can also read off from \eqref{Proposition_6.5_Graham_main} that the Weyl weight of $\hat\Omega^{(1)}_{jk}$ is $+2$.

\par
Now we will define Weyl-obstruction tensors as the derivatives of $\ti R_{\infty jk\infty}$. 
 \begin{defn}\label{Weyl_ext_Obst}
Suppose $k$ is a positive integer. The $k^{th}$ extended Weyl-obstruction tensor $\hat \Omega^{(k)}_{ij}$ is defined as
\begin{equation}
\label{eq:def1}
\hat \Omega^{(k)}_{ij}= \ti R_{\infty ij\infty ;\underbrace{\scriptstyle \infty\cdots\infty}_{k-1}}|_{\rho=0,t=1}.
\end{equation}
For $k=1$ we can see from \eqref{eq:Rijkl} that $\ti R_{\infty jk\infty}|_{\rho=0,t=1}= \hat\Omega^{(1)}_{jk}$ is indeed the first extended Weyl-obstruction tensor.
\end{defn}

From the symmetry of the Weyl-ambient Riemann tensor we can immediately see that $\hat \Omega^{(k)}_{ij}$ given by Definition \ref{Weyl_ext_Obst} is symmetric. From the Ricci-flatness condition $\ti Ric(\ti g)=0$ and the fact that $\ti R_{0 IJK}=0$, we can see that $\hat \Omega^{(k)}_{ij}$ is traceless. Now we will show another important property of the extended Weyl-obstruction tensors defined in this way, namely that they are Weyl covariant. 
\begin{lemma}
\label{lem:lemmaR0}
The components of the Riemann tensor of the Weyl-ambient metric $\tilde g$ satisfy
\be
\label{eq:lemmaR0}
\tilde R_{IJK0;M_1\cdots M_r}=-\frac{1}{t}\sum^r_{s=1}\tilde R_{IJKM_s;M_1\cdots\hat{M}_s\cdots M_r}\,,
\ee
where $\hat{M}_s$ means to remove $M_s$ from the indices.
\end{lemma}
\begin{proof}
Computing the Christoffel symbols of the Weyl-ambient metric $\tilde g$ in \eqref{Weyl_ambient}, one finds $\tilde\Gamma^i{}_{j0}=\frac{1}{t}\delta^i{}_j$ and $\tilde\Gamma^\infty{}_{\infty0}=\frac{1}{t}$. Differentiating $\un{\cal T}=t\un\p_t$ we have ${\cal T}^I{}_{;J}=\delta^I{}_{J}$ and ${\cal T}^I{}_{;JK}=0$, then
\begin{align*}
({\cal T}^L\tilde R_{IJKL})_{;M_1\cdots M_r}&=\tilde R_{IJKM_1;M_2\cdots M_r}+({\cal T}^L\tilde R_{IJKL,M_1})_{;M_2\cdots M_r}\\
&=\tilde R_{IJKM_1;M_2\cdots M_r}+\tilde R_{IJKM_2;M_2\cdots M_r}+({\cal T}^L\tilde R_{IJKL,M_1M_2})_{;M_3\cdots M_r}\\
&=\cdots\\
&=\tilde R_{IJKM_1;M_2\cdots M_r}+\cdots+\tilde R_{IJKM_r;M_1\cdots M_{r-1}}+{\cal T}^L\tilde R_{IJKL;M_1\cdots M_r}\,.
\end{align*}
The left-hand side of this equation vanishes since $R_{IJK0}=0$, and thus the above equation leads to \eqref{eq:lemmaR0}.
\end{proof}

\begin{prop}
The extended Weyl-obstruction tensor $\hat \Omega^{(k)}_{ij}$ defined in \eqref{eq:def1} is a Weyl-covariant tensor with Weyl weight $2k$.
\end{prop}
\begin{proof}
According to Proposition \ref{Proposition_6.5_Graham}, if we choose $(IJKL;M_{1}\dots M_{r})= (\infty, i,j, \infty;\underbrace{\infty\dots \infty}_{(k-1)})$, then $s_{\infty}= k+1$ and under a Weyl transformation we have
\begin{align}
\label{eq:obsWeylcov}
\ti R'_{\infty ij\infty ;\underbrace{\scriptstyle \infty\cdots\infty}_{k-1}}|_{\rho'=0,t'=1}={\cal B}(x)^{2k}\big(\ti R_{\infty ij\infty ;\underbrace{\scriptstyle \infty\cdots\infty}_{k-1}}+\Upsilon_{i}\ti R_{\infty 0j\infty ;\underbrace{\scriptstyle \infty\cdots\infty}_{k-1}}+\Upsilon_{j}\ti R_{\infty i0\infty ;\underbrace{\scriptstyle \infty\cdots\infty}_{k-1}}\big)|_{\rho'=0,t'=1}\,.
\end{align}
It follows from Lemma \ref{lem:lemmaR0} that
\begin{align}
R_{\infty i 0\infty;\underbrace{\scriptstyle \infty\cdots\infty}_{k-1}}=\frac{k-1}{t}R_{\infty i \infty\infty;\underbrace{\scriptstyle \infty\cdots\infty}_{k-2}}=0\,.
\end{align}
Therefore, we obtain from \eqref{eq:obsWeylcov} that $\hat \Omega'^{(k)}_{ij}={\cal B}(x)^{2k}\hat \Omega^{(k)}_{ij}$ under a Weyl transformation, i.e.\ $\hat\Omega^{(k)}_{ij}$ is a Weyl-covariant tensor with Weyl weight $2k$.
\end{proof}

Finally, we would like to show that Definition \ref{Weyl_ext_Obst} and Definition \ref{def1} are equivalent; that is, the Weyl-obstruction tensors defined by the derivatives of the ambient Riemann tensor in the frame $\{\bm e^+,\bm e^i,\bm e^-\}$ and the coordinate basis $\{\td t,\td x^i,\td\rho\}$ are equivalent. To start, let us look at the transformation between $\{\bm e^+,\bm e^i,\bm e^-\}$ and the coordinate basis $\{\td t,\td x^i,\td\rho\}$:
\begin{align}
\label{eq:basistrans}
\left(\begin{array}{c}\bm e^+ \\\bm e^j \\\bm e^-\end{array}\right)=\left(\begin{array}{ccc}1 &ta_i & 0 \\0 & \delta^j{}_i & 0 \\ \rho & -\rho t a_i & t\end{array}\right)\left(\begin{array}{c}\td t \\\td x^i \\ \td\rho\end{array}\right)\,.
\end{align}
Denote the transformation matrix as $\Lambda$, i.e.\ $\bm e^{J}=\Lambda^{J}{}_{I'}\td x^{I'}$ ($J=\{+,i,-\}$, $I'=\{0,i,\infty\}$), then the inverse matrix reads
\begin{align}
\Lambda^{-1}=\left(\begin{array}{ccc}1 &-ta_j & 0 \\0  & \delta^i{}_j & 0\\ -\frac{\rho}{t} &2\rho a_j & \frac{1}{t}\end{array}\right)\,.
\end{align}

Comparing \eqref{eq:ambRiem} and \eqref{eq:Rijkl}, we can see that the components $R_{ijkl}$, $R_{-ijk}$ and $R_{-ij-}$ in the null frame match the corresponding components $R_{ijkl}$, $R_{\infty ijk}$ and $R_{\infty ij\infty}$ in the coordinate basis when $\rho=0$ and $t=1$. Now let us show that any Weyl-obstruction tensor defined in \eqref{eq:def2} is equivalent to that in \eqref{eq:def1}. First, notice that although the components $\tilde R_{-+MN}$ of $\tilde R_{IJKL}$ in the frame $\{+,i,-\}$ vanish, the components $\ti\nabla_{P} \ti R_{-+MN}$ are not necessarily zero. (Using the notation in Section \ref{sec:topdown}, here we denote $\tilde\nabla_{\un D_P}$ as $\tilde\nabla_P$ for $P=+,i,-$.) The following lemma will be used in the proof of Proposition \ref{prop:obsequiv}.
\begin{lemma}
\label{lem:Riem}
$\tilde\nabla_{P}\underbrace{\tilde\nabla_-\cdots\tilde\nabla_-}_{n}\tilde R_{-+MN}=-\frac{1}{t}\delta^{i}{}_{P}\underbrace{\tilde\nabla_-\cdots\tilde\nabla_-}_{n}\tilde R_{-iMN}$ for any integer $n\geqslant0$.
\end{lemma}
\begin{proof}
See Appendix \ref{App:Lemproof}.
\end{proof}

\begin{prop}
\label{prop:obsequiv}
$\tilde R_{\infty ij\infty;\underbrace{\scriptstyle\infty\cdots\infty}_{n}}=t^{n+2}\underbrace{\tilde\nabla_-\cdots\tilde\nabla_-}_{n}\tilde R_{-ij-}$ for any integer $n\geqslant0$.
\end{prop}

\begin{proof}
For $n=0$ one can see this readily from \eqref{eq:Rijkl}. Since $\un\p_{\tilde N'}=\Lambda^{M}{}_{N'}\un D_{M}$, for $n\geqslant1$ the left-hand side of the above equation can be written as (primes are dropped for simplicity)
\begin{align}
\label{eq:Riem1}
\tilde R_{\infty ij\infty;\underbrace{\scriptstyle\infty\cdots\infty}_{n}}&=\Lambda^{M_1}{}_{\infty}\cdots\Lambda^{M_n}{}_{\infty}\Lambda^{K}{}_{\infty}\Lambda^{I}{}_{i}\Lambda^{J}{}_{j}\Lambda^{L}{}_{\infty}\tilde\nabla_{M_1}\cdots\tilde\nabla_{M_n}\tilde R_{KIJL}\nn\\
&=t^{n+2}\Lambda^{I}{}_{i}\Lambda^{J}{}_{j}\underbrace{\tilde\nabla_{-}\cdots\tilde\nabla_{-}}_{n}\tilde R_{-IJ-}\,,
\end{align}
where $\Lambda^{M}{}_\infty=t\delta^{M}{}_-$ [see \eqref{eq:basistrans}] is used in the second equality. Using the symmetries of the Riemann tensor, we have
\begin{align}
\label{eq:Riem2}
\Lambda^{I}{}_{i}\Lambda^{J}{}_{j}\underbrace{\tilde\nabla_{-}\cdots\tilde\nabla_{-}}_{n}\tilde R_{-IJ-}={}&\underbrace{\tilde\nabla_{-}\cdots\tilde\nabla_{-}}_{n}\tilde R_{-ij-}+\Lambda^{+}{}_{i}\underbrace{\tilde\nabla_{-}\cdots\tilde\nabla_{-}}_{n}\tilde R_{-+j-}\nn\\
&+\Lambda^{+}{}_{j}\underbrace{\tilde\nabla_{-}\cdots\tilde\nabla_{-}}_{n}\tilde R_{-i+-}+\Lambda^{+}{}_{i}\Lambda^{+}{}_{j}\underbrace{\tilde\nabla_{-}\cdots\tilde\nabla_{-}}_{n}\tilde R_{-++-}\nn\\
={}&\underbrace{\tilde\nabla_{-}\cdots\tilde\nabla_{-}}_{n}\tilde R_{-ij-}\,,
\end{align}
where $\Lambda^{i}{}_j=\delta^{i}{}_j$ is used in the first equality and Lemma \ref{lem:Riem} is used in the second equality. Plugging \eqref{eq:Riem2} into \eqref{eq:Riem1} completes the proof.
\end{proof}
From Proposition \ref{prop:obsequiv} we can directly see that the $\hat\Omega^{(k)}_{ij}$ defined in \eqref{eq:def2} is equivalent to \eqref{eq:def1}. Therefore, the descriptions of the Weyl-obstruction tensors in the first order and second order formalisms are equivalent. Each of these two formalisms have their own advantages. The first order formalism is suited for the top-down approach as the metric $\tilde g$ has a simple form in the dual frame $\{\bm e^I\}$. It is also more convenient to construct Weyl-covariant tensors in the first order formalism since \eqref{eq:prop1} gives a covariant transformation while \eqref{Proposition_6.5_Graham_main} has the matrix $p$ with an off-diagonal element. On the other hand, the second order formalism is designed for the bottom-up approach, as one can evaluate the initial value problem more naturally in the coordinate basis.

 \section{Conclusions}
 \label{sec:conclu}
In this paper we have generalized the ambient construction for conformal manifolds to that for Weyl manifolds. Inspired by the WFG gauge for AlAdS \cite{Ciambelli:2019bzz}, we introduced the Weyl-ambient metric $\tilde g$ in \eqref{Weyl_ambient}. From a top-down perspective we showed how the Weyl-ambient space $(\tilde M,\tilde g)$ induces a Weyl geometry on a codimension-2 manifold $M$. The metric $\tilde g$ and the LC connection on $\tilde M$ give rise to a Weyl class $[\gamma^{(0)},a^{(0)}]$ on $M$, in which a representative includes an induced metric $\gamma^{(0)}_{ij}$ together with a Weyl connection $a^{(0)}_i$. The ambient Weyl diffeomorphisms on $\tilde M$ act as Weyl transformations on the $M$. This enhances the codimension-2 conformal geometry in the usual ambient construction to a Weyl geometry $(M,[\gamma^{(0)},a^{(0)}])$. 
\par
From a bottom-up perspective, we formulated the $(d+2)$-dimensional Weyl-ambient space from a $d$-dimensional Weyl manifold $(M,[g,a])$. We first introduced a Weyl structure ${\cal P}_W$ on $M$ together with a Weyl connection. We then generalized the definition of ambient spaces to Weyl-ambient spaces, and proved that any Weyl-ambient space can be put in Weyl-normal form by a diffeomorphism. Besides assigning the Weyl connection $a_{i}$ on ${\cal P}_W$, the $\rho$-coordinate lines of a Weyl-ambient space in Weyl-normal form are not required to be geodesics but can acquire an acceleration $\un{\cal A}$. By taking the Weyl structure as an initial surface, we have shown that there exists a unique Weyl-ambient space in Weyl-normal form for any given Weyl manifold provided the data $(g_{ij},a_i,\un{\cal A})$ is given. The metric generated order by order from the initial value problem is exactly the $\tilde g$ we introduced in \eqref{Weyl_ambient} from the top-down approach, where $g_{ij}$ corresponds to $\gamma^{(0)}_{ij}$, and $(a_i,\un{\cal A})$ corresponds to $a_i(x,\rho)$. 
\par
Based on the Weyl-ambient construction, we investigated Weyl-covariant quantities induced by the ambient tensors in both first and second order formalisms. As an important example, the extended Weyl-obstruction tensor $\hat\Omega^{(k)}_{ij}$ is defined through covariant derivatives of the ambient Riemann tensor, and its definition in the first and second order formalisms are shown to be equivalent. We also proved that $\hat\Omega^{(k-1)}_{ij}$ corresponds to the pole of $\gamma^{(k)}_{ij}$ at $d=2k$ in the ambient metric expansion, which justifies the description of Weyl-obstruction tensors in \cite{Jia:2021hgy}. Compared with the extended obstruction tensor $\Omega^{(k-1)}_{ij}$, whose residue is only conformally covariant in $d=2k$, the extended Weyl-obstruction tensor $\hat\Omega^{(k-1)}_{ij}$ is Weyl covariant in any dimension. 
\par

There are many possible extensions of this work. It would be interesting to reframe the Weyl-ambient geometry in the framework of Atiyah Lie algebroids, as set up in \cite{Ciambelli:2021ujl}. Using the machinery of Lie algebroids, one can naturally formulate the BRST cohomology and quantum anomalies in a geometrical way \cite{JKL}. Since obstruction tensors, which are defined canonically through the ambient construction, also appear in the type B Weyl anomaly, we hope that the Weyl-ambient space can provide a geometric interpretation for the Weyl anomaly in the Lie algebroid language. We will explore this in future work.
\par
The Weyl-ambient space induces the $\text{Diff}(M)\ltimes\text{Weyl}$ symmetry on the codimension-2 manifold $M$, which can be regarded as an asymptotic corner symmetry \cite{Ciambelli:2022cfr,Ciambelli:2022vot}. The algebra of corner symmetries and their Noether charges have been studied in \cite{Ciambelli:2021vnn,Ciambelli:2022cfr} (see also \cite{Freidel:2021cjp}), it is possible to apply the results therein to the Weyl-ambient space and study the asymptotic corner symmetries of the Weyl-ambient space. Moreover, since the surface $\cal N$ at $\rho=0$ of the Weyl-ambient space is null, there is an induced Carroll structure  \cite{Ciambelli:2018xat, Ciambelli:2019lap}. This is evident from the fact that the ambient Weyl diffeomorphism acts on the null surface as (a special case of) a Carrollian diffeomorphism. It would be interesting to see if the results of \cite{Mittal:2022ywl} for a Ricci-flat space with a Killing vector can be emulated for the Weyl-ambient space, which is Ricci-flat but possesses a conformal Killing vector $\un T=t\un\p_t$.
\par
One also expects intriguing holographic applications of the Weyl-ambient construction, for example in the context of celestial holography \cite{Pasterski:2016qvg,Raclariu:2021zjz,Pasterski:2021raf} and codimension-2 holography \cite{Akal:2020wfl,Ogawa:2022fhy}. In particular, the $\text{Diff}(M)\ltimes\text{Weyl}$ symmetry on $M$ corresponds to the Weyl-BMS symmetry on $\tilde M$ \cite{Freidel:2021fxf} (with supertranslations turned off). Therefore, we expect that the Weyl-ambient construction will provide a new arena for realizing the holographic principle.
\par
The Weyl-ambient metric construction is part of a bigger program of introducing the Weyl connection back into physics. Viewed as an ordinary gauge symmetry, the Weyl symmetry can provide an organizing principle for constructing effective field theories (e.g.,\ for conformal hydrodynamics). Weyl manifolds would be the proper geometric setup for such future investigations. More recently, the ambient construction was used to study correlators of CFTs on general curved backgrounds \cite{Parisini:2022wkb}. We hope the Weyl-ambient geometries can be utilized in similar contexts.

\section*{Acknowledgements}
We thank Luca Ciambelli and Tassos Petkou for constructive comments on the manuscript. We would also like to thank George Katsianis, Marc Klinger, Pin-Chun Pai and Antony Speranza for discussions. This work was supported by the U.S. Department of Energy under contract DE-SC0015655.

\appendix
\renewcommand{\theequation}{\thesection.\arabic{equation}}
\setcounter{equation}{0}

\section{Coordinate Systems of the Flat Ambient Space}\label{AppA}
In this appendix we demonstrate the transformation between the flat ambient metric in different coordinate systems introduced in Section \ref{Sec2}.
\par
Start with Minkowski spacetime $\mathbb{R}^{1,d+1}$ in Lorentzian coordinates $\{X^{0},X^{i}\}$ with $i=1,\dots , d+1$:
\begin{equation}\label{app:1}
\eta= - (\td X^{0})^{2}+\sum_{i=1}^{d+1}(\td X^{i})^{2}\,.
\end{equation}
First, we can define a stereographic coordinate system $\{\ell,r,x^{i}\}$ as follows:
\begin{equation}\label{Lorentzian_to_stereo}
X^{0}= \ell \frac{L^2+ r^2}{L^2-r^2}\,,\qquad X^{i}= \ell \frac{2L}{L^2-r^2} x^{i},\qquad i=1,\dots, d+1\,,
\end{equation}
where $r^2 = \sum\limits_{i=1}^{d+1} (x^{i})^2 $ and $L$ is a positive constant. In this system, the Minkowski metric \eqref{app:1} becomes
\begin{equation}\label{app:2}
\eta= -\td \ell^2 + \frac{\ell^2}{L^2} \frac{4}{(1- (r/L)^2)^2}\sum_{i=1}^{d+1}(\td x^i)^2 = -\td \ell^2 + \frac{\ell^2}{L^2} \frac{4}{(1- (r/L)^2)^2}\left(\td r^2 + r^2 \td\Omega^{2}_{d}\right)\,,
\end{equation}
where in the second equality we expressed $\{x^{i}\}$ in the spherical coordinates. The coordinate patch is $\ell>0$, $0 \leqslant r<L$, which covers the interior of the future light cone. Notice that in  these coordinates the metric has a ``cone'' form \eqref{Flat_Ambient_4}, with $g^+$ given in \eqref{eq:adsglob}, which is the $(d+1)$-dimensional Euclidean AdS metric $g^+_G$ in global coordinates. This AdS metric can be converted into the FG from by transforming the coordinate $r$ to a coordinate $z$
\begin{equation}\label{stereo_to_FG}
r= L\left(\frac{2L-z}{2L+z}\right)\,.
\end{equation}
Then, the metric \eqref{app:2} takes the form
\begin{equation}
\label{eq:etaFG}
\eta = - \td \ell^2 + \frac{\ell^2}{z^2} \left(\td z^2 + L^2(1- \frac{1}{4}(z/L)^2)^2\td\Omega_{d}^2\right)\,,
\end{equation}
and the interior of the future light cone is now covered by  $\ell>0$, $ 0<z< 2L$. We can further convert \eqref{eq:etaFG} into the ambient form \eqref{ambient_metric} by setting
\begin{equation}\label{FG_to_ambient}
\ell=z t\,,\qquad z^2 =-2\rho\,,
\end{equation}
and the metric turns into the form shown in \eqref{Flat_Ambient_2}:
\begin{equation}\label{app:3}
\eta = 2\rho \td t^2 + 2t \td t \td\rho + t^2 (1+ \frac{\rho}{2 L^2})^2 L^2 \td\Omega_{d}^2 \,.
\end{equation} 
Plugging \eqref{FG_to_ambient} and \eqref{stereo_to_FG} into \eqref{Lorentzian_to_stereo} we find that
\beq
X^{0}+ R= 2L t\,,\qquad\tan\alpha\equiv\frac{R}{X^{0}}= \frac{1+ \frac{\rho}{2L^2}}{1- \frac{\rho}{2L^2}}\,,
\eeq
where $R^2 = \sum\limits_{i=1}^{d+1} (X^{i})^2 $. From the above equation one can see that the constant-$t$ and constant-$\rho$ surfaces are indeed the cones depicted in Figure~\ref{fig:cones}, with  $\alpha$ the angle of the constant-$\rho$ cone with respect to the $X^{0}$-axis.
\par
The Minkowski metric \eqref{app:1} can also be written in the cone form with $g^+=g^+_P$ the Euclidean AdS metric in Poincar\'e coordinates given in \eqref{eq:adsPoin}. Introduce another coordinate system $\{\ell,x^{i},z\}$ as follows:
\begin{equation}
X^{0}= \frac{\ell}{2L z}\left(L^{2}+ \sum_{i=1}^{d}(x^{i})^{2} + z^{2}\right)\,,\quad X^{d+1}= \frac{\ell}{2Lz}\left(L^{2}- \sum_{i=1}^{d}(x^{i})^{2}- z^{2}\right)\,,\quad X^{i}= \frac{\ell x^{i}}{z}\,.
\end{equation}
The metric \eqref{app:1} becomes 
\begin{equation}\label{app:4}
\eta= -\td \ell^2 + \frac{\ell^2}{z^2}\left(\td z^2+ \delta_{ij} \td x^{i}\td x^{j}\right),\qquad i=1,\cdots,d\,,\qquad z>0\,.
\end{equation}
Define the ambient coordinate system $\{t,x^{i},\rho\}$ as
\begin{equation}
\ell=zt\,,\qquad z^2 =-2\rho\,,
\end{equation}
then the metric \eqref{app:4} will have the form shown in \eqref{Flat_Ambient_3}
\begin{equation}\label{app:5}
\eta = 2\rho \td t^2 + 2t\td t \td\rho + t^2\delta_{ij} \td x^{i}\td x^{j}\,,\qquad i=1,\cdots,d\,.
\end{equation}

\section{Details of Null Frame Calculations}
\label{App:Null}
In Section \ref{sec:topdown} we introduced the following frame:
\begin{align}
\label{e+-}
\bm e^+&=\td t+ta_i\td x^i\,,\qquad\bm e^-=t\td\rho+\rho \td t-t\rho a_i\td x^i\,,\qquad \bm e^i=\td x^i\,,\\
\un D_+&=\un\p_t-\frac{\rho}{t}\un\p_\rho\,,\qquad\un D_-=\frac{1}{t}\un\p_\rho\,,\qquad\un D_i=\un\p_i-ta_i\un\p_t+2\rho a_i\un\p_\rho\,.
\end{align}
The metric \eqref{Weyl_ambient} can be written in this frame as
\begin{align*}
\tilde g=\bm e^+\otimes\bm e^-+\bm e^-\otimes\bm e^++t^2\gamma_{ij}\bm e^i\otimes\bm e^j\,,
\end{align*}
and the metric components read
\begin{align*}
\tilde g_{+-}&=\tilde g_{-+}=1\,,\qquad \tilde g_{ij}=t^2\gamma_{ij}\,,\qquad \tilde g^{+-}=\tilde g^{-+}=1\,,\qquad \tilde g^{ij}=\frac{1}{t^2}\gamma^{ij}\,.
\end{align*}
The commutation relations of the frame are as follows:
\begin{align}
\begin{split}
[\un D_+,\un D_i]&=-(a_i-\rho \varphi_i)\un D_+-\rho^2 \varphi_i\un D_-\,,\qquad[\un D_+,\un D_-]=0\,,\\
[\un D_-,\un D_i]&=(a_i+\rho\varphi_i)\un D_--\varphi_i\un D_+\,,\qquad[\un D_i,\un D_j]=-tf_{ij}\un D_++t\rho f_{ij}\un D_-\,,
\end{split}
\end{align}
where $\varphi=\p_\rho a_i$, and $f_{ij}=D_ia_j-D_ja_i$. From the above commutators we can read off the commutation coefficients:
\begin{align}
\begin{split}
C_{+i}{}^+&=-a_i+\rho \varphi_i\,,\qquad C_{+i}{}^-=-\rho^2 \varphi_i\,,\qquad C_{-i}{}^+=-\varphi_i\,,\\
C_{-i}{}^-&=a_i+\rho \varphi_i\,,\qquad C_{ij}{}^+=-tf_{ij}\,,\qquad C_{ij}{}^-=t\rho f_{ij}\,.
\end{split}
\end{align}
Then, we can compute the connection coefficients $\tilde\Gamma^{P}{}_{MN}$ of the ambient LC connection:
\begin{align}
\begin{split}
\tilde\Gamma^{P}{}_{MN}=&\frac{1}{2}\tilde g^{PQ}(D_M\tilde g_{NQ}+D_N\tilde g_{QM}-D_Q\tilde g_{MN})\\
&-\frac{1}{2} \tilde g^{PQ}(C_{MQ}{}^{R}\tilde g_{RN}+C_{NM}{}^{R}\tilde g_{RQ}-C_{QN}{}^{R}\tilde g_{RM})\,.
\end{split}
\end{align}
The nonvanishing components are
\begin{align}
\tilde\Gamma^{+}{}_{i+}&=a_i\,,\qquad
\tilde\Gamma^{+}{}_{ij}=-\frac{t}{2}(\p_\rho\gamma_{ij}+f_{ij})\,,\qquad
\tilde\Gamma^{-}{}_{ij}=-t\gamma_{ij}+\frac{\rho t}{2}(\p_\rho \gamma_{ij}+f_{ij})\,,\nn\\
\tilde\Gamma^{-}{}_{i-}&=-a_i\,,\qquad\tilde\Gamma^{i}{}_{j-}=\frac{1}{2t}\gamma^{ik}(\p_\rho\gamma_{jk}+f_{jk})\,,\qquad\tilde\Gamma^{i}{}_{j+}=\frac{1}{t}\delta^i{}_j-\frac{\rho}{2t}\gamma^{ik}(\p_\rho\gamma_{jk}+f_{jk})\,,\nn\\
\tilde\Gamma^{i}{}_{jk}&=\frac{1}{2}\gamma^{il}(\p_j\gamma_{lk}+\p_k\gamma_{jl}-\p_l\gamma_{jk})-(a_j\delta^i{}_k+a_k\delta^i{}_j-a^i\gamma_{jk})+\rho\gamma^{il}(a_j\p_\rho \gamma_{lk}+a_k\p_\rho \gamma_{jl}-a_l\p_\rho \gamma_{jk})\,,\nn\\
\tilde\Gamma^{+}{}_{+i}&=\rho \varphi_i\,,\qquad\tilde\Gamma^{i}{}_{++}=\frac{\rho^2}{t^2}\gamma^{ij}\varphi_j\,,\qquad\tilde\Gamma^{-}{}_{+i}=-\rho^2\varphi_i\,,\qquad\tilde\Gamma^{i}{}_{+-}=-\frac{\rho}{t^2}\gamma^{ij}\varphi_j\,,\nn\\
\tilde\Gamma^{+}{}_{-i}&=-\varphi_i\,,\qquad\tilde\Gamma^{i}{}_{-+}=-\frac{\rho}{t^2}\gamma^{ij}\varphi_j\,,\qquad\tilde\Gamma^{-}{}_{-i}=\rho \varphi_i\,,\qquad\tilde\Gamma^{i}{}_{--}=\frac{1}{t^2}\gamma^{ij}\varphi_j\,,\nn\\
\tilde\Gamma^{i}{}_{+j}&=\frac{1}{t}\delta^i{}_j-\frac{\rho}{2t}\gamma^{ik}(\p_\rho\gamma_{jk}+f_{jk})\,,\qquad\tilde\Gamma^{i}{}_{-j}=\frac{1}{2t}\gamma^{ik}(\p_\rho\gamma_{jk}+f_{jk})\,,
\end{align}
which constitute the connection 1-form $\tilde{\bm\omega}^{M}{}_{N}$ presented in \eqref{eq:conn1form}. Then, using Cartan's second structure equation
\begin{align}
\tilde{\bm R}^{M}{}_{N}=\td\tilde{\bm\omega}^{M}{}_{N}+\tilde{\bm\omega}^{M}{}_{P}\wedge\tilde{\bm\omega}^{P}{}_{N}\,,
\end{align}
we can find the ambient curvature 2-form, the nonvanishing components are
\begin{align}
\tilde{\bm R}^+{}_i={}&-t(\hat\nabla_j\psi_{ki}-\rho \varphi_if_{jk})\bm e^j\wedge \bm e^k+(\p_\rho\psi_{ji}-\psi_{jk}\psi_i{}^k-\hat\nabla_j\varphi_i-2\rho\varphi_i\varphi_j)\bm e^j\wedge(\bm e^--\rho\bm e^+)\nn\,,\\
\tilde{\bm R}^-{}_i={}&\rho t(\hat\nabla_j\psi_{ki}-\rho \varphi_if_{jk})\bm e^j\wedge \bm e^k-\rho(\p_\rho\psi_{ji}-\psi_{jk}\psi_i{}^k-\hat\nabla_j\varphi_i-2\rho\varphi_i\varphi_j)\bm e^j\wedge(\bm e^--\rho\bm e^+)\nn\,,\\
\tilde{\bm R}^i{}_+={}&-\frac{\rho}{t}(\hat\nabla_j\psi_k{}^i-\rho\varphi^if_{jk})\bm e^j\wedge \bm e^k+\frac{\rho}{t^2}(\p_\rho\psi_j{}^i+\psi_k{}^i\psi_j{}^k-\hat\nabla_j\varphi^i-2\rho\varphi^i\varphi_j) \bm e^j\wedge(\bm e^--\rho\bm e^+)\nn\,,\\
\tilde{\bm R}^i{}_-={}&\frac{1}{t}(\hat\nabla_j\psi_k{}^i-\rho\varphi^if_{jk})\bm e^j\wedge \bm e^k\nn-\frac{1}{t^2}(\p_\rho\psi_j{}^i+\psi_k{}^i\psi_j{}^k-\hat\nabla_j\varphi^i-2\rho\varphi^i\varphi_j)\bm e^j\wedge(\bm e^--\rho\bm e^+)\nn\,,\\
\tilde{\bm R}^i{}_j={}&\frac{1}{2}(\bar R^i{}_{jkl}+\delta^i{}_jf_{kl})\bm e^k\wedge \bm e^l-(\delta_k{}^i\psi_{lj}+\psi_k{}^i\gamma_{lj}-2\rho\psi_k{}^i\psi_{lj}+\rho\psi_j{}^if_{kl})\bm e^k\wedge \bm e^l\nn\\
\label{eq:tildeRcomp}
&+\frac{1}{t}\gamma^{il}(\hat\nabla_l\psi_{jk}-\hat\nabla_j\psi_{lk}+2\rho f_{jl}\varphi_k) \bm e^k\wedge(\bm e^--\rho\bm e^+)\,,
\end{align}
where $\hat\nabla$ is introduced in \eqref{eq:hatnabla}, $\psi_{ij}\equiv\frac{1}{2}(\p_\rho\gamma_{ij}+f_{ij})$, and
\begin{align}
\bar R^{i}{}_{jkl}\equiv D_{k}\ti \Gamma^{i}{}_{lj}-D_{l}\ti \Gamma^{i}{}_{kj}+\ti \Gamma^{i}{}_{km}\ti \Gamma^{m}{}_{lj}-\ti \Gamma^{i}{}_{lm}\ti \Gamma^{m}{}_{kj}\,.
\end{align}
The components in \eqref{eq:tildeRcomp} constitute the curvature 2-form $\tilde{\bm R}^{M}{}_{N}$ presented in \eqref{eq:curv2form}. 
\par
Now one can derive the extended Weyl-obstruction tensors according to Definition \ref{def1}. For example, $\hat\Omega^{(1)}_{ij}$ and $\hat\Omega^{(2)}_{ij}$ can be computed as follows:
\begin{align*}
\tilde R_{-ij-}={}&\p_\rho\gamma_{ij}-\psi_{ik}\psi_j{}^k-\hat\nabla_{(i}\varphi_{j)}-2\rho\varphi_i\varphi_j\,,\\
\nabla_-\tilde R_{-ij-}={}&\frac{1}{t}\Big[\p^2_\rho\gamma_{ij}-2\psi_j{}^k{\cal B}_{ki}-2\psi_{i}{}^{k}{\cal B}_{kj}-\hat\nabla_{(i}(\p_\rho\varphi_{j)})-6\varphi_i\varphi_j+\varphi^k\varphi_k\gamma_{ji}-\psi_i{}^k\hat\nabla_j\varphi_k-\psi_j{}^k\hat\nabla_i\varphi_k\\
&\quad+\varphi^k(\hat\nabla_i\psi_{jk}+2\hat\nabla_j\psi_{ki}-2\hat\nabla_k\psi_{ji}+\hat\nabla_i\psi_{kj}-\hat\nabla_k\psi_{ij})\\
&\quad+2\rho\big(\varphi^k(\varphi_j\psi_{ik}+\varphi_i\psi_{kj}-\varphi_k\psi_{ij})-2\varphi^k\varphi_{(i}\psi_{j)k}-3\p_\rho\varphi_{(i}\varphi_{j)}-2\varphi^k\varphi_{(i}f_{j)k}\big)\Big]\,.
\end{align*}
Plugging the on-shell solution \eqref{Pgf}--\eqref{g6} in to the above expressions, one obtains the extended Weyl-obstruction tensors $\hat\Omega^{(1)}_{ij}$ and $\hat\Omega^{(2)}_{ij}$ given in \eqref{eq:Omega1} and \eqref{eq:Omega2}.
\par
From the components of the ambient Riemann curvature, we can also find the Ricci components in this frame:
\begin{align*}
\tilde R_{++}&=-\rho\tilde  R_{+-}=-\rho\tilde  R_{-+}=\rho^2\tilde R_{--}=-\frac{\rho^2}{t^2}(\gamma^{ij}\p_\rho\psi_{ji}+\psi_k{}^i\psi_i{}^k-\hat\nabla_i\varphi^i-2\rho\varphi^i\varphi_i)\,, \\
\tilde R_{i+}&=\tilde R_{+i}=-\rho \tilde R_{i-}=-\rho \tilde R_{-i}=-\frac{\rho}{t}(\hat\nabla_j\psi_i{}^j-\hat\nabla_i\theta-2\rho\varphi^jf_{ji})\,,\\
\tilde R_{ij}&=\bar R_{ij}+f_{ij}-(d-2)\psi_{ji}-\theta\gamma_{ji}+2\rho({\cal B}_{ij}+\theta\psi_{ji}-\psi_j{}^k\psi_{ki}-\psi_i{}^kf_{kj})\,,
\end{align*}
where ${\cal B}_{ij}$ is defined in \eqref{eq:Bij}. The Ricci-flatness condition gives the following three equations:
\begin{align}
\label{eq:R++0}
0&=\gamma^{ij}\p_\rho\psi_{ji}+\psi_k{}^i\psi_i{}^k-\hat\nabla_i\varphi^i-2\rho\varphi^i\varphi_i \,,\\
\label{eq:R+i0}
0&=\hat\nabla_j\psi_i{}^j-\hat\nabla_i\theta-2\rho\varphi^jf_{ji}\,,\\
\label{eq:Rij0}
0&=\bar R_{ij}+f_{ij}-(d-2)\psi_{ji}-\theta\gamma_{ji}+2\rho({\cal B}_{ij}+\theta\psi_{ji}-\psi_j{}^k\psi_{ki}-\psi_i{}^kf_{kj})\,.
\end{align}
In the leading order when $\rho=0$, the condition \eqref{eq:R++0} leads to the fact that $\hat\Omega^{(1)}_{ij}$ is traceless, and \eqref{eq:Rij0} gives the Bianchi identity $\hat\nabla^{(0)}_i\hat P^i{}_j=\hat\nabla^{(0)}\hat P$, where $\hat P$ is the trace of $\hat P_{ij}$.
\par
Differentiating $\bar R_{ij}$ with respect to $\rho$ yields
\begin{align}
\p_\rho\bar R_{ij}={}&\hat\nabla_k\hat\nabla_j\psi^{k}{}_i+\hat\nabla_k\hat\nabla_i\psi_j{}^k-\hat\nabla_k\hat\nabla^k\psi_{ji}-\hat\nabla_j\hat\nabla_i\theta-\hat\nabla_i\varphi_j+(d-1)\hat\nabla_i\varphi_j+\gamma_{ij}\hat\nabla_k\varphi^k\nn\\
&+4\rho a_k(\varphi_j\psi^k{}_i+\varphi_i\psi^k{}_j-\varphi^k\psi_{ij})-4\rho a_j\varphi_i\theta+2\rho\hat\nabla(\varphi_j\psi^k{}_i+\varphi_i\psi^k{}_j-\varphi^k\psi_{ij})-2\rho\hat\nabla_i\theta\nn\\
&+2\rho\varphi_k(\nabla_j\psi^k{}_i+\nabla_i\psi_j{}^k-\nabla^k\psi_{ji})-2\rho\varphi_j\hat\nabla_i\theta-2\rho\big((d+2)\varphi_i\varphi_j-\varphi_k\varphi^k\gamma_{ij}\big)\nn\\
\label{eq:Rijprho}
&+2\rho \varphi_k(\varphi_j\psi^k{}_i+\varphi_i\psi^k{}_j-\varphi^k\psi_{ij})-2\rho \varphi_j\varphi_i\theta\,,
\end{align}
which leads to \eqref{eq:psipole2} when differentiating \eqref{eq:Rij0}.

\section{Proofs}\label{AppC}
\subsection{Proof of Theorem \ref{diff_Weyl_normal_form}}
\label{Appthm5pt1}
To prove Theorem \ref{diff_Weyl_normal_form}, we first need to introduce a \emph{$(g,a)$-transversal vector} (generalized from the concept of a $g$-transversal vector in \cite{Fefferman:2007rka}), where the horizontal subspace $H_p$ defined by the Weyl connection plays an important role. Once we pick a representative $(g,a)$ in the Weyl class, $g$ induces an isomorphism between ${\cal P}_W$ and $\cal G$ through \eqref{eq:trivial}, which determines the fibre coordinate $t$ of ${\cal P}_W$; $a$ defines for any $p\in {\cal P}_W$ a horizontal subspace $H_p\subset T_p{\cal P}_W$ given in \eqref{eq:HV}, which can also be viewed as a subspace of $T_{(p,0)}({\cal P}_W\times \mathbb{R})$ via the inclusion map $\iota : {\cal P}_W\to {\cal P}_W\times \mathbb{R}$. We define a vector $\un {\cal V}\in T_{(p,0)}({\cal P}_W\times \mathbb{R})$ to be a \emph{$(g,a)$-transversal vector} for $\tilde g$ if it satisfies the following three conditions at $(p,0)$:
\begin{equation}\label{Weyl g-transversal vector}
\text{\ding{172} }\ti g(\un {\cal V}, \un{\cal T})= t^2\,,\qquad \text{\ding{173} }\ti g(\un{\cal V},\un{\cal H})=0\quad\forall\, \un{\cal H}\in H_{p}\,,\qquad\text{\ding{174} }\ti g(\un{\cal V},\un{\cal V})=0\,.
\end{equation}
When $a_i(x)=0$ in \eqref{projector_triv_1}, i.e.,\ $\mathbf{a}=\un\pa_{t} \otimes \td t$, the $(g,a)$-transversal vector for $\tilde g$ goes back to the $g$-transversal vector for $\tilde g$ defined in \cite{Fefferman:2007rka}. From \eqref{form_1} one can see that for $(\tilde M,\tilde g)$ in Weyl-normal form, $\un \pa_{\rho}$ is $(g,a)$-transversal for $\tilde g$ at $(p,0)$.
Following the proof of Lemma 2.10 in \cite{Fefferman:2007rka}, it is straightforward to show that the $(g,a)$-transversal vector is unique and dilatation-invariant (i.e.\ $\delta_{s*}V_p=V_{\delta_s(p)}$) for $\tilde g$ at $(p,0)$.  

The proof of Theorem \ref{diff_Weyl_normal_form} proceeds similar to the proof of Proposition 2.8 in \cite{Fefferman:2007rka}; one only has to let the $g$-transversal vector $\un{\cal V}$ to be a $(g,a)$-transversal vector. Here we will not repeat all the details but only outline the proof and elaborate on the steps when the Weyl connection $a$ is relevant. 
\begin{proof}[Proof of Theorem \ref{diff_Weyl_normal_form}]
Suppose $p\in{\cal P}_W$ and let $\un {\cal V}_{p}$ be the $(g,a)$-transversal vector for $\tilde g$ at $(p,0)$.
One can parametrize the (non-geodesic) curve $C_p:\lambda\mapsto \phi(p,\lambda)\in \ti M$ with initial conditions
\begin{equation}
\label{phi_initial_conditions}
\phi(p,0)= (p,0)\,,\qquad \pa_{\lambda}\phi(p,\lambda)|_{\lambda=0}= \un{\cal V}_{p}\,,
\end{equation} 
with the ``equation of motion'' $\nabla_{\un{\cal U}}\un{\cal U}=\un{\cal A}$, where $\un{\cal U}= \frac{\td}{\td\lambda}$ is the tangent vector to the accelerated curve $C_p$, and the acceleration vector $\un{\cal A}$ satisfies $\tilde g(\un{\cal T}, \un{\cal A})=0$. Suppose the domain of $\phi$ is $\tilde U_0\subset{\cal P}_W\times\bb R$, which is dilatation-invariant. Then $\phi:\tilde U_0\to\tilde M$ is a smooth map commuting with dilatation, and it can be proved that there exists  $\tilde U_1\subset \tilde U_0$ as a dilatation-invariant neighborhood of ${\cal P}_W\times\bb \{0\}$ such that $\phi:\tilde U_1\to \tilde M$ is a diffeomorphism (see \cite{Fefferman:2007rka}).

\par
Furthermore, one can define $\tilde M'=\{(p,\lambda)\in\tilde U_1|\,(p,\mu)\in\tilde U_1, \forall\mu\in\bb R$ satisfying $|\mu|\leqslant|\lambda|\}$. It is easy to verify that $(\tilde M',\phi^*\tilde g)$ satisfies the conditions of Definition \ref{Ambient} and thus is a Weyl pre-ambient space for $(M,[g,a])$. It follows that for each $p\in {\cal P}_W$, the set for $\lambda$ such that $(p,\lambda)\in\tilde M'$ is an open interval $I_p$ containing $0$, and the parametrized curve $C'_p:\lambda\mapsto(p,\lambda)$ with tangent vector $\un{\cal U}'$ and the acceleration $\un {\cal A}'=\nabla'_{\cal U'}{\cal U'}$ satisfies $\phi^*\tilde g(\un {\cal T}', \un {\cal A}')=0$, where $\un{\cal T}'\equiv \phi^*\un{\cal T}$, and $\nabla'$ is the Levi-Civita connection associated with $\phi^*\tilde g$. Hence, conditions (5.1) and (5.2) of Definition \ref{Weyl_normal_form} are satisfied by $(\tilde M',\phi^*\tilde g)$.
\par
Finally let us verify condition (5.3) of Definition \ref{Weyl_normal_form}. Since $\un{\cal V}$ satisfies the conditions in \eqref{Weyl g-transversal vector} and $\phi$ satisfies \eqref{phi_initial_conditions}, under the identification $\mathbb{R}_{+}\times M\times \mathbb{R}\simeq {\cal P}_W\times \mathbb{R}$ induced by $g$ we have at $(\lambda=0,p)$:
 \begin{align}
 (\phi^{*}\ti g)(\un\pa_{\lambda},\un{\cal T})&=t^2\,\nn\\
 (\phi^{*}\ti g)(\un\pa_{\lambda},\un{\cal H})&=0\qquad\forall\,\un{\cal H}\in {\cal H}_{p}\,,\\
 (\phi^{*}\ti g)(\un\pa_{\lambda},\un\pa_{\lambda})&= 0\,.\nn
 \end{align}
For a given connection $\mathbf{a}=t\un\pa_t\otimes \big(t^{-1}\td t+a_i(x)\td x^i\big)$ on ${\cal P}_W$, the horizontal subspace ${\cal H}_p$ at $(p,0)$ is spanned by $\un D_i=\un\p_i-ta_i\un\p_t$. Since $(\tilde M',\phi^*\tilde g)$ is a Weyl pre-ambient space for $(M,[g,a])$, $\iota^*(\phi^*g)$ is the tautological tensor $\mathbf g_0$ on ${\cal P}_W$. Then, the above equations give that $\phi^{*}\ti g|_{\lambda=0}=  t^2 \mathbf{g}_{0}  + 2t(\td t + ta_{i}(x)\td x^{i})\td\lambda $. Therefore, all the conditions in Definition \ref{Weyl_normal_form} are satisfied by $(M',\phi^*g)$, which completes the existence part of the Proposition. The uniqueness part follows from the fact that the above construction of $\phi$ is forced. Suppose $\phi:M\to M'$ is a diffeomorphism such that $(M',\phi^*g)$ is a pre-ambient space in Weyl-normal form, then $\un{\cal V}_p$ must be $(g,a)$-transversal for $\ti g$ at $(p,0)$, and the curve $C'_p:\lambda\mapsto \phi(z,\lambda)$ must be the unique curve satisfying the initial conditions \eqref{phi_initial_conditions} and having the acceleration $\un{\cal A}$, which determines $\phi:\tilde M\to \tilde M'$ uniquely.
\end{proof}

\subsection{Proof of Theorem \ref{Weyl_ambient_existence}}

 \par
 \begin{proof}[Proof of Theorem \ref{Weyl_ambient_existence}]
 The proof of this theorem has two main parts. First, from $\ti Ric(\ti g)=0$ and the initial value of $\tilde g$ at $\rho=0$ we will determine the first $\rho$-derivative of the metric components at $\rho=0$. Then, using an inductive argument we will show that all higher derivatives (to infinite order) at $\rho=0$ can also be determined from the Ricci-flatness condition. Let us write the unknown components of $\tilde g$ as
 \begin{equation}
 \ti g_{00}= c(x,\rho)\,,\qquad\ti g_{0 i}= t b_{i}(x,\rho)\,,\qquad\ti g_{ij}= t^2 g_{ij}(x,\rho)\,,
 \end{equation}
where $g_{ij}(x,\rho)$ can be considered as a one-parameter family of metrics on $M$. From property (2) above we have the initial values $c(x,0)=0$ and $b_{i}(x,0)=0$. The general metric has the form
\begin{equation}\label{_metric_formal_second_order}
\ti g_{IJ}= \begin{blockarray}{cccc}
 & 0 & j & \infty  \\
\begin{block}{c(ccc)}
  0 & c(x,\rho)& tb_i(x,\rho) & t\\
   i &  tb_i(x,\rho)  & t^2g_{ij}(x,\rho)&  t^2a_i(x,\rho)\\
  \infty & t & t^2a_i(x,\rho) & 0\\
\end{block}
\end{blockarray}
\,\,,
 \end{equation}
 and the inverse metric is
\begin{equation}\label{inverse_metric_formal}
\ti g^{IJ}=\left(\begin{array}{ccc}\frac{a^2}{\chi}& -\frac{(1- a\cdot b)a^{j}+ a^2 b^{j}}{t\chi} & \frac{1- a\cdot b}{t\chi}\\
 -\frac{(1- a\cdot b)a^{i}+ a^2 b^{i}}{t\chi}  & \frac{g^{ij}}{t^2}+ \frac{(1-a\cdot b)(a^{i}b^{j}+ a^{j}b^{i}) + a^2 b^{i}b^{j}- (c- b^2)a^{i}a^{j}}{t^2 \chi}& \frac{(c- b^2)a^{i}- (1- a\cdot b)b^{i}}{t^2 \chi} \\
 \frac{1- a\cdot b}{t\chi}&\frac{(c- b^2)a^{j}- (1- a\cdot b)b^{j}}{t^2 \chi}& \frac{b^2 - c}{t^2 \chi}\\\end{array}\right)\, ,
 \end{equation}
 where $a^{i}\equiv g^{im}a_{m}$, $b^{i}\equiv g^{im}b_{m}$ and $\chi= a^{2}(c-b^2)+ (1- a\cdot b)^2$, with $a^{2}= a_{k}a^{k}$, $b^2= b_{k}b^{k}$ and $a\cdot b= a_{k}b^{k}$. 
The Christoffel symbols $\ti \Gamma_{IJK}\equiv \ti g_{KM}\ti \Gamma\indices{^{M}_{IJ}}$ are
 
 \begin{equation}\label{Christoffel_Initial_value}
 \begin{split}
2\ti \Gamma_{IJ 0}&=
\left(\begin{array}{ccc}
0& \pa_{j}c &\pa_{\rho}c\\
\pa_{i}c  & t(\pa_{i}b_{j}+ \pa_{j}b_{i}- 2 g_{ij})& t(\pa_{\rho}b_{i}- 2a_{i})\\
\pa_{\rho}c &t(\pa_{\rho}b_{j}- 2a_{j})& 0
\end{array}\right)\,,\\
2\ti \Gamma_{IJ k}&=
\left(\begin{array}{ccc}
2b_{k}- \pa_{k}c & t\left(2 g_{jk}+ \pa_{j}b_{k}- \pa_{k}b_{j}\right)&t(2a_{k} + \pa_{\rho}b_{k})\\
t(2 g_{ik}+ \pa_{i}b_{k}- \pa_{k}b_{i})  &2 t^2 \gamma_{ijk}& t^2 (\pa_{\rho}g_{ik}+ F_{ik}) \\
t(2a_{k}+ \pa_{\rho}b_{k}) & t^2 \left(\pa_{\rho}g_{jk}+ F_{jk}\right) & 2t^2 \pa_{\rho}a_{k}
\end{array}\right)\,,\\
2\ti \Gamma_{IJ \infty}&=
\left(\begin{array}{ccc}
2- \pa_{\rho}c& t(2a_{j}- \pa_{\rho}b_{j}) & 0\\
t(2a_{i}- \pa_{\rho}b_{i})  & t^2 (\pa_{i}a_{j}+ \pa_{j}a_{i}- \pa_{\rho}g_{ij})& 0\\
0 &0& 0
\end{array}\right)\,,
\end{split}
\end{equation}
where $\gamma_{ijk}= g_{km}\gamma\indices{^{m}_{ij}}$ with $\gamma\indices{^{m}_{ij}}= \frac{1}{2}g^{mk}\left(\pa_{i}g_{jk}+ \pa_{j}g_{ik}-\pa_{k}g_{ij}\right)$ and  $F_{jk}= \pa_{j}a_{k}- \pa_{k}a_{j}$. Calculating the components $\ti R_{IJ}$ of $\tilde Ric(\tilde g)$ to the leading order in $\rho$-expansion from
\begin{equation}\label{Ricci}
\ti R_{IJ}= \frac{1}{2}\ti g^{KL}\left(\pa^{2}_{IL}\ti g_{JK}+ \pa^{2}_{JK}\ti g_{IL}- \pa^{2}_{KL}\ti g_{IJ}- \pa^{2}_{IJ}\ti g_{KL}\right)+ \ti g^{KL}\ti g^{PQ}\big(\ti\Gamma_{ILP}\ti\Gamma_{JKQ}- \ti\Gamma_{IJP}\ti\Gamma_{KLQ}\big) \,,
\end{equation}
and setting them to zero as the Ricci-flatness condition demands, we obtain
\begin{equation}\label{metric_components_order_1}
\begin{split}
c(x,\rho)&= 2\rho + O(\rho^{2})\,,\qquad b_{i}(x,\rho)= O(\rho^2)\,,\\ g_{ij}(x,\rho)&=g_{ij}(x)+ \rho\big(2\hat P_{(ij)}- 2a_{i}(x)a_{j}(x)\big)+ O(\rho^{2})\,,
\end{split}
\end{equation} 
where $\hat P_{ij}$ is the Weyl-Schouten tensor. One can observe that this agrees with \eqref{Weyl_ambient}, where $g_{ij}(x)$ corresponds to $\gamma^{(0)}_{ij}$ in the expansion \eqref{eq:gexpan}, and the order $O(\rho)$ matches $\gamma^{(1)}_{ij}$ given in \eqref{Pgf}. Note that the above components of a Weyl-ambient metric reduce to the components of an ambient metric in \cite{Fefferman:2007rka} when the Weyl connection $a_i$ is turned off. 
\par
The next stage of the proof is to carry out an inductive perturbation calculation for higher orders in $\rho$. The purpose of this calculation is to prove (inductively) that the Ricci-flatness condition can be used to determine the unknown components of $\tilde g$ in Weyl-normal form to infinite order in $\rho$.\par
Let $\ti g^{[k]}$ represent a metric that includes the terms of the $\rho$-expansion of $\tilde g$ up to (including) order $O(\rho^{k})$, i.e.,\ $\tilde g=\tilde g^{[k]}+O(\rho^{k+1})$. Then, the Ricci-flatness condition of $\tilde g$ implies that the components $\ti R^{[k]}_{IJ}$ of $Ric(\tilde g^{[k]})$ satisfy
\begin{align}
\label{eq:Ricciflatm}
\begin{split}
\ti R_{IJ}(\ti g^{[k]})= O(\rho^{k})\quad I,J\neq\infty\,,\qquad
\ti R_{I\infty}(\ti g^{[k]})= O(\rho^{k-1})\,.
\end{split}
\end{align}
To carry out the induction, we assume that $\ti g^{[m-1]}$ has been uniquely determined from the condition \eqref{eq:Ricciflatm} with $k=m-1$. We have seen this is true for $m=2$ above by explicit calculation. Now we want to show that $\ti g^{[m]}$ then can be uniquely determined from the condition \eqref{eq:Ricciflatm} with $k=m$. Set $\ti g^{[m]}_{IJ}= \ti g^{[m-1]}_{IJ}+ \Phi_{IJ}$, with 
\begin{equation}
\label{eq:g+phi}
\Phi_{IJ}:=
\left(\begin{array}{ccc}
\Phi_{00}&\Phi_{0j}&0\\
\Phi_{i0}&\Phi_{ij}&\Phi_{i\infty}\\
0&\Phi_{j\infty}&0
\end{array}\right)
= \rho^m 
\left(\begin{array}{ccc}
\phi_{00}(x)&t\phi_{0j}(x)&0\\
t\phi_{0i}(x)&t^2 \phi_{ij}(x)&t^{2}a^{(m)}_{i}(x)\\
0&t^{2}a^{(m)}_{j}(x)&0
\end{array}\right)\,,
\end{equation}
where $a_{i}^{(m)}(x)$ is the $m^{th}$ order term of $a_{i}(x,\rho)$ [see \eqref{eq:aexpan}], and we have considered the fact that $\ti g^{[m]}_{IJ}$ satisfies \eqref{lemma_conditions}. All we have to show is that $\phi_{00}$, $\phi_{0i}$ and $\phi_{ij}$ can all be uniquely determined. From \eqref{Ricci} one finds that 
\begin{equation}\label{Ricci_2}
\begin{split}
\ti R^{[m]}_{IJ}={}&\ti R^{[m-1]}_{IJ}+  \frac{1}{2}\ti g_{[m]}^{KL}\left(\pa^{2}_{IL} \Phi_{JK}+ \pa^{2}_{JK} \Phi_{IL}- \pa^{2}_{KL} \Phi_{IJ}- \pa^{2}_{IJ} \Phi_{KL}\right) \\&
+ \ti g_{[m]}^{KL}\ti g_{[m]}^{PQ}\left(\ti \Gamma^{[m]}_{ILP}\Gamma^{\Phi}_{JKQ}+  \Gamma^{\Phi}_{ILP}\ti\Gamma^{[m]}_{JKQ}- \ti\Gamma^{[m]}_{IJP} \Gamma^{\Phi}_{KLQ}- \Gamma^{\Phi}_{IJP}\ti \Gamma^{[m]}_{KLQ}\right)+ O(\rho^{m})\,,
\end{split}
\end{equation}
where $\ti g_{[m]}^{KL}$ and $\ti \Gamma_{IJK}^{[m]}$ are the inverse and Christoffel symbols of $\ti g^{[m]}_{KL}$, respectively, and $\Gamma^{\Phi}_{IJK}\equiv\frac{1}{2}( \pa_{J}\Phi_{IK}+ \pa_{I}\Phi_{JK}- \pa_{K}\Phi_{IJ})$. The components of $\Gamma^{\Phi}_{IJK}$ can be expressed as follows:
\begin{equation}\label{GammaPhi}
\begin{split}
2\Gamma^{\Phi}_{IJ0}&= 
\left(\begin{array}{ccc}
0&0&\p_\rho\Phi_{00}\\
0&0&\p_\rho\Phi_{i0}\\
\p_\rho\Phi_{00}&\p_\rho\Phi_{0j}&0
\end{array}\right)+O(\rho^m)\,,\\
2\Gamma^{\Phi}_{IJk}&= 
\left(\begin{array}{ccc}
0&0&\p_\rho\Phi_{0k}\\
0&0&\p_\rho\Phi_{ik}\\
\p_\rho\Phi_{0k}&\p_\rho\Phi_{jk}&2\p_\rho\Phi_{\infty k}
\end{array}\right)+O(\rho^m)\,,\\
2\Gamma^{\Phi}_{IJ\infty}&= 
\left(\begin{array}{ccc}
-\p_\rho\Phi_{00}&-\p_\rho\Phi_{0j}&0\\
-\p_\rho\Phi_{i0}&-\p_\rho\Phi_{ij}&0\\
0&0&0
\end{array}\right)+O(\rho^m)\,.
\end{split}
\end{equation}

Substituting \eqref{GammaPhi} and the leading order of $\tilde\Gamma^{[m]}_{IJK}$ and $\tilde g_{[m]}^{IJ}$ [i.e.,\ the leading order of $\tilde g^{IJ}$, $\tilde\Gamma_{IJK}$ in \eqref{inverse_metric_formal},\eqref{Christoffel_Initial_value}] into \eqref{Ricci_2}, one finds 
\begin{equation}\label{Ricci_inductive}
\begin{split}
t^2 \ti R^{[m]}_{00}&=t^2 \ti R^{[m-1]}_{00} + m\rho^{m-1}\left(m- 1 - \frac{d}{2}\right)\phi_{00}+O(\rho^{m})\,,\\
t\ti R^{[m]}_{0i} &= t\ti R^{[m-1]}_{0i}+ m\rho^{m-1}\left[\frac{1}{2}\pa_{i}\phi_{00}+ \left(m-1-\frac{d}{2}\right)\phi_{0i}\right]+ O(\rho^{m})\,,\\
\ti R_{ij}^{[m]}&= \ti R_{ij}^{[m-1]}+ m\rho^{m-1}\left[(m- \frac{d}{2})\phi_{ij}- \frac{1}{2}g_{ij}g^{km}\phi_{km}+ \mathring\nabla_{(i}\phi_{j)0}+ \mathring P_{ij}\phi_{00}\right] + O(\rho^{m})\,,\\
t\ti R^{[m]}_{0\infty}&=t\ti R^{[m-1]}_{0\infty}+ \frac{1}{2}m(m-1)\rho^{m-2}\phi_{00} + O(\rho^{m-1})\,,\\
\ti R^{[m]}_{i\infty}&=\ti R^{[m-1]}_{i\infty}+ \frac{1}{2}m(m-1)\rho^{m-2}\phi_{i0}+ O(\rho^{m-1})\,,\\
\ti R^{[m]}_{\infty \infty}&=\ti R^{[m-1]}_{\infty\infty} -m(m-1)\rho^{m-2}\left(\frac{1}{2}a^{2}\phi_{00}- a^{k}\phi_{k0} +\frac{1}{2}g^{km}\phi_{km}\right)+ O(\rho^{m-1})\,,
\end{split}
\end{equation}
where $\mathring P_{ij}$, $\mathring\nabla$ are the LC Schouten tensor and LC connection associated with the metric $g_{ij}(x)$. Although the Weyl connection $a^{(0)}_{i}(x)$ appears throughout the calculation, it cancels itself out rather unexpectedly, except for the terms in $\ti R^{(m)}_{\infty\infty}$. The inductive argument then proceeds in the same way as \cite{Fefferman:2007rka}. First we consider the Ricci components with $I,J\neq \infty$. From the first two equations in \eqref{Ricci_inductive} one can uniquely determine $\phi_{00}$ and $\phi_{0i}$ such that $\ti R^{[m]}_{00}$ and $\ti R^{[m]}_{0i}$ both vanish up to order $O(\rho^{m})$. Then, from the third equation in \eqref{Ricci_inductive}  one can uniquely solve for $\phi_{ij}$ such that the order $O(\rho^{m-1})$ of $\ti R^{[m]}_{ij}$ vanishes. Therefore, $\tilde g^{[m]}$ will be uniquely determined by $\ti R^{[m]}_{IJ}=O(\rho^m)$ for $I,J\neq \infty$ once $\tilde g^{[m-1]}$ is determined, and hence the unknown components of $\ti g_{IJ}$ can be determined to infinite order. 
\par
Note that when $d=2m$, the situation becomes subtle because the term $\phi_{ij}$ vanishes in $\tilde R^{[m]}$. In \cite{Fefferman:2007rka}, this is attributed to the obstruction of the Ricci-flatness condition at $O(\rho^{d/2-1})$ when $d$ is an even integer, and one has to carefully consider even and odd $d$ separately. Nevertheless, since we consider the dimension $d$ as a continuous parameter, we can always solve for $\phi_{ij}$ from the Ricci-flatness condition for any $d$, and the information regarding these obstructions is not lost but takes the form of poles in $\phi_{ij}$ at $d=2m$. As is shown in Proposition \ref{prop:obspole}, since $\phi_{ij}$ represents the order $O(\rho^m)$ of $g_{ij}(x,\rho)$ in the Weyl-ambient metric \eqref{Weyl_ambient}, this pole represents exactly the Weyl-obstruction tensor.
\par
So far we have proved that the unknown components of $\ti g$ are determined to infinite order by the Ricci-flatness condition for $I,J\neq \infty$. To finish the analysis we also need to show that the remaining Ricci components $\ti R_{I\infty}$ also vanish to infinite order when we plug in the solution for $\tilde g$ obtained from $\ti R_{IJ}=0$ for $I,J\neq \infty$. Consider the Bianchi identity $\ti g^{JK}\nabla_{I}\ti R_{JK}= 2\ti g^{JK} \nabla_{J}\ti R_{IK}$. Expanding the covariant derivative in terms of the Christoffel symbols we get
\begin{equation}\label{Bianchi}
2\ti g^{JK}\pa_{J}\ti R_{IK}- \ti g^{JK}\pa_{I}\ti R_{JK}- 2 \ti g^{JK}\ti g^{PQ}\ti\Gamma_{JKP}\ti R_{QI}=0\,.
\end{equation}
Since $\ti R_{I\infty}= {\cal O}(\rho^{m-2})$ is trivially true for $m=2$, now we want to show that $\ti R_{I\infty}= {\cal O}(\rho^{m-2})$ leads to $\ti R_{I\infty}= {\cal O}(\rho^{m-1})$ by means of the Bianchi identity. Expanding \eqref{Bianchi} for $I=0,i,\infty$ and making use of the homogeneity property of the metric we get
 \begin{equation}\label{Bianchi_2}
 \begin{split}
& \left(d -2 -2 \rho\pa_{\rho}\right)\ti R_{0\infty} = {\cal O}(\rho^{m-1})\\&
 (d-2 -2 \rho\pa_{\rho})\ti R_{i\infty}- t\pa_{i}\ti R_{0\infty}={\cal O}(\rho^{m-1})\\&
 a^{2}\left( t^{-1}d \ti R_{\infty 0}+2\pa_{0}\ti R_{\infty 0}\right)- 2 t ^{-1}a^{m}\left( \pa_{m}\ti R_{\infty 0}-(2-d)t^{-1}\ti R_{\infty m}\right) \\&
+2 t^{-2}\left(d-2 - \rho\pa_{\rho}\right)\ti R_{\infty\infty}+2 t^{-2}g^{mk}\mathring \nabla_{m}\ti R_{\infty k}+ 2t^{-1}\mathring{P}\ti R_{\infty 0} ={\cal O}(\rho^{m-1})\,.
 \end{split}
 \end{equation}
We can see that the Weyl connection appears only in the last equation of \eqref{Bianchi_2}. Note that all the Ricci terms $\ti R_{IJ}$ with $I,J \neq \infty$ has been dropped from \eqref{Bianchi_2} since they vanish to infinite order. Suppose $\ti R_{I\infty}= \gamma_{I}\rho^{m-2}$. The first equation in \eqref{Bianchi_2} gives $(d+2 -2m)\gamma_{0}= {\cal O}(\rho)$, and thus $\ti R_{0\infty}={\cal O}(\rho^{m-1})$. The second equation in \eqref{Bianchi_2}  gives  $(d+2 -2 m)\gamma_{i}= {\cal O}(\rho)$, and thus $\ti R_{i\infty}= {\cal O}(\rho^{m-1})$. The last equation then gives $(d-m)\gamma_{\infty}= {\cal O}(\rho)$, so $\ti R_{\infty\infty}={\cal O}(\rho^{m-1})$. This completes the inductive argument and thus $\ti R_{I\infty}$ can also be made to vanish to infinite order.
\par
To summarize, we have shown by an inductive argument that there exists a Weyl-ambient space $(\tilde M,\tilde g)$ for $(M,[g,a])$ in Weyl-normal form with acceleration $\un{\cal A}$. Some components of $\tilde g$ have the form in \eqref{lemma_conditions}, and all the unknown components are determined uniquely to infinite order of $\rho$ at ${\cal P}_W\times \{0\}$ by the Ricci-flatness condition.
\end{proof}

\subsection{Proof of Lemma \ref{lem:Riem}}
\label{App:Lemproof}
\begin{proof}[Proof of Lemma \ref{lem:Riem}]
We will prove this identity by induction. First, noticing that $\tilde R_{-+MN}=0$, when $n=0$ we have
\begin{align*}
\tilde\nabla_i\tilde R_{-+MN}&=-\tilde\Gamma^j{}_{i-}\tilde R_{j+MN}-\tilde\Gamma^{j}{}_{i+}\tilde R_{-jMN}=\frac{1}{t}\psi_i{}^j \tilde R_{+jMN}-\frac{1}{t}(\delta^j{}_i-\rho\psi^j{}_i )\tilde R_{-jMN}\\
&=-\frac{\rho}{t}\psi_i{}^j \tilde R_{-jMN}-\frac{1}{t}(\delta^j{}_i-\rho\psi_i{}^j)\tilde R_{-jMN}=-\frac{1}{t}\tilde R_{-iMN}\,,\\
\tilde\nabla_-\tilde R_{-+MN}&=-\tilde\Gamma^j{}_{--}\tilde R_{j+MN}-\tilde\Gamma^{j}{}_{-+}\tilde R_{-jMN}=0\,,\\
\tilde\nabla_+\tilde R_{-+MN}&=-\tilde\Gamma^j{}_{+-}\tilde R_{j+MN}-\tilde\Gamma^{j}{}_{++}\tilde R_{-jMN}=0\,,
\end{align*}
where we used the fact that $\tilde\Gamma^i{}_{M+}=-\rho\tilde\Gamma^i{}_{M-}$ and $\tilde R_{+jMN}=-\rho\tilde R_{-jMN}$, which can be seen from
\eqref{eq:conn1form} and \eqref{eq:curv2form}, respectively. Thus, for $n=0$ we have $\nabla_{P}\tilde R_{-+MN}=-\frac{1}{t}\delta^{i}{}_{P}\tilde R_{-iMN}$. Assuming that this lemma holds for all $n\leqslant k-1$, now we show that it will hold for $n=k>0$:
\begin{align*}
&\tilde\nabla_i\underbrace{\tilde\nabla_-\cdots\tilde\nabla_-}_{k}\tilde R_{-+MN}\\
={}&D_i\underbrace{\tilde\nabla_-\cdots\tilde\nabla_-}_{k-1}\tilde R_{-+MN}-\tilde\Gamma^j{}_{i-}\tilde\nabla_j\underbrace{\tilde\nabla_-\cdots\tilde\nabla_-}_{k-1}\tilde R_{-+MN}-\cdots-\tilde\Gamma^j{}_{i-}\underbrace{\tilde\nabla_-\cdots\tilde\nabla_-}_{k-1}\tilde\nabla_j\tilde R_{-+MN}
\\
&-\tilde\Gamma^+{}_{i-}\tilde\nabla_+\underbrace{\tilde\nabla_-\cdots\tilde\nabla_-}_{k-1}\tilde R_{-+MN}-\cdots-\tilde\Gamma^+{}_{i-}\underbrace{\tilde\nabla_-\cdots\tilde\nabla_-}_{k-1}\tilde\nabla_+\tilde R_{-+MN}
\\
&-\tilde\Gamma^j{}_{i-}\underbrace{\tilde\nabla_-\cdots\tilde\nabla_-}_{k}\tilde R_{j+MN}-\tilde\Gamma^j{}_{i+}\underbrace{\tilde\nabla_-\cdots\tilde\nabla_-}_{k}\tilde R_{-jMN}\\
&-\tilde\Gamma^{P}{}_{i M}\underbrace{\tilde\nabla_-\cdots\tilde\nabla_-}_{k}\tilde R_{-+ P N}-\tilde\Gamma^{P}{}_{i N}\underbrace{\tilde\nabla_-\cdots\tilde\nabla_-}_{k}\tilde R_{-+MP}\\
={}&\frac{k}{t^2}\psi_i{}^j \underbrace{\tilde\nabla_-\cdots\tilde\nabla_-}_{k-1}\tilde R_{-jMN}- \frac{1}{t}\psi_i{}^j\underbrace{\tilde\nabla_-\cdots\tilde\nabla_-}_{k}(\rho \tilde R_{-jMN})-\frac{1}{t}(\delta^j{}_i-\rho\psi_i{}^j )\underbrace{\tilde\nabla_-\cdots\tilde\nabla_-}_{k}\tilde R_{-jMN}\\
={}&-\frac{1}{t}\underbrace{\tilde\nabla_-\cdots\tilde\nabla_-}_{k}\tilde R_{-jMN}\,,\\
&\tilde\nabla_-\underbrace{\tilde\nabla_-\cdots\tilde\nabla_-}_{k}\tilde R_{-+MN}\\
={}&D_-\underbrace{\tilde\nabla_-\cdots\tilde\nabla_-}_{k-1}\tilde R_{-+MN}-\tilde\Gamma^j{}_{--}\tilde\nabla_j\underbrace{\tilde\nabla_-\cdots\tilde\nabla_-}_{k-1}\tilde R_{-+MN}-\cdots-\tilde\Gamma^j{}_{--}\underbrace{\tilde\nabla_-\cdots\tilde\nabla_-}_{k-1}\tilde\nabla_j\tilde R_{-+MN}
\\
&-\tilde\Gamma^j{}_{--}\underbrace{\tilde\nabla_-\cdots\tilde\nabla_-}_{k}\tilde R_{j+MN}-\tilde\Gamma^j{}_{-+}\underbrace{\tilde\nabla_-\cdots\tilde\nabla_-}_{k}\tilde R_{-jMN}\\
&-\tilde\Gamma^{P}{}_{- M}\underbrace{\tilde\nabla_-\cdots\tilde\nabla_-}_{k}\tilde R_{-+PN}-\tilde\Gamma^{P}{}_{-N}\underbrace{\tilde\nabla_-\cdots\tilde\nabla_-}_{k}\tilde R_{-+MP}\\
={}&\frac{k}{t^2}\varphi^j\underbrace{\tilde\nabla_-\cdots\tilde\nabla_-}_{k-1}\tilde R_{-jMN}-\frac{1}{t}\varphi^j\underbrace{\tilde\nabla_-\cdots\tilde\nabla_-}_{k}(\rho \tilde R_{-jMN})+\frac{\rho}{t}\varphi^j\underbrace{\tilde\nabla_-\cdots\tilde\nabla_-}_{k}\tilde R_{-jMN}=0\,,\\
&\tilde\nabla_+\underbrace{\tilde\nabla_-\cdots\tilde\nabla_-}_{k}\tilde R_{-+MN}\\
={}&D_+\underbrace{\tilde\nabla_-\cdots\tilde\nabla_-}_{k-1}\tilde R_{-+MN}-\tilde\Gamma^j{}_{+-}\tilde\nabla_j\underbrace{\tilde\nabla_-\cdots\tilde\nabla_-}_{k-1}\tilde R_{-+MN}-\cdots-\tilde\Gamma^j{}_{+-}\underbrace{\tilde\nabla_-\cdots\tilde\nabla_-}_{k-1}\tilde\nabla_j\tilde R_{-+MN}
\\
&-\tilde\Gamma^j{}_{+-}\underbrace{\tilde\nabla_-\cdots\tilde\nabla_-}_{k}\tilde R_{j+MN}-\tilde\Gamma^j{}_{++}\underbrace{\tilde\nabla_-\cdots\tilde\nabla_-}_{k}\tilde R_{-jMN}\\
&-\tilde\Gamma^{P}{}_{+M}\underbrace{\tilde\nabla_-\cdots\tilde\nabla_-}_{k}\tilde R_{-+PN}-\tilde\Gamma^{P}{}_{+N}\underbrace{\tilde\nabla_-\cdots\tilde\nabla_-}_{k}\tilde R_{-+MP}\\
={}&-\frac{k\rho}{t^2}\varphi^j\underbrace{\tilde\nabla_-\cdots\tilde\nabla_-}_{k-1}\tilde R_{-jMN}+\frac{\rho}{t}\varphi^j\underbrace{\tilde\nabla_-\cdots\tilde\nabla_-}_{k}(\rho \tilde R_{-jMN})-\frac{\rho^2}{t}\varphi^j\underbrace{\tilde\nabla_-\cdots\tilde\nabla_-}_{k}\tilde R_{-jMN}=0\,.
\end{align*}
Therefore, $\tilde\nabla_{P}\underbrace{\tilde\nabla_-\cdots\tilde\nabla_-}_{n}\tilde R_{-+MN}=-\frac{1}{t}\delta^{i}{}_{P}\underbrace{\tilde\nabla_-\cdots\tilde\nabla_-}_{n}\tilde R_{-iMN}$ holds for $n=k$ if it is valid for all $n\leqslant k-1$, which completes the proof.
\end{proof}

\providecommand{\href}[2]{#2}\begingroup\raggedright\endgroup

\end{document}